\documentclass[11pt]{article}

\usepackage[margin = 1in]{geometry}

\usepackage{hyperref}
\usepackage{amsthm}
\usepackage{caption}
\usepackage{subcaption}
\usepackage{graphics}
\usepackage{enumitem}
\usepackage{microtype}

\usepackage{filecontents}
\usepackage{amsmath,amssymb,color}
\usepackage{multirow}
\usepackage{mdframed}
\hypersetup{colorlinks=true,citecolor=blue, linkcolor=blue, urlcolor=blue}
\usepackage{tikz}

\newtheorem{theorem}{Theorem}
\newtheorem{lemma}[theorem]{Lemma}

\newtheorem{observation}[theorem]{Observation}
\newtheorem{corollary}[theorem]{Corollary}
\newtheorem{remark}[theorem]{Remark}

\newtheorem{definition}[theorem]{Definition}

\newtheorem*{lem:treespathPotts}{Lemma~\ref{lem:treespathPotts}}
\newtheorem*{lem:treespath}{Lemma~\ref{lem:treespath}}
\newtheorem*{lem:corrdecayIsing}{Lemma~\ref{lem:corrdecayIsing}}
\newtheorem*{thm:sampleising}{Theorem~\ref{thm:sampleising}}

\def\ra{\rightarrow}
\def\conn{\leftrightarrow}

\def\open{\mathrm{op}}
\def\closed{\mathrm{cl}}
\def\TV{\mathrm{TV}}

\def\emm{\mathrm{e}}

\def\Gc{{\mathcal{G}}}

\def\Eb{{\mathbf{E}}}
\def\E{{\mathbb{E}}}

\newcommand{\cons}{{10}}

\newcommand{\TreeD}{\mathbb{T}_{\Delta}}

\newcommand{\norm}[1]{\left\lVert#1\right\rVert}

\newcommand{\SampleRC}{\mbox{\textsc{SampleRC}}}
\newcommand{\SamplePotts}{\mbox{\textsc{SampleAntiPotts}}}
\newcommand{\ReSample}{\mbox{\textsc{ReSample}}}
\newcommand{\IdealReSample}{\mbox{\textsc{IdealReSample}}}

\newcommand{\Corr}{\mathsf{Corr}}
\newcommand{\dist}{\mathsf{dist}}
\newcommand{\noteq}{\mathsf{neq}}
\newcommand{\eq}{\mathsf{eq}}

\title{Sampling in Uniqueness from the Potts and Random-Cluster Models on Random Regular Graphs\footnote{A preliminary short version of the
manuscript (without the proofs) appeared in the proceedings of RANDOM/APPROX 2018.}}

\author{Antonio Blanca\thanks{School of Computer Science, Georgia
Institute of Technology, Atlanta GA 30332.
Research supported in part by NSF grants CCF-1617306 and CCF-1563838. \texttt{\{ablanca,vigoda\}@cc.gatech.edu}} \and Andreas Galanis\thanks{University of Oxford,
  Wolfson Building, Parks Road, Oxford, OX1~3QD, UK.
The research leading to these results has received funding from the European Research Council under
the European Union's Seventh Framework Programme (FP7/2007-2013) ERC grant agreement no. 334828. The paper
reflects only the authors' views and not the views of the ERC or the European Commission. The European Union is not liable for any use that may be made of the information contained therein. \texttt{\{andreas.galanis,leslie.goldberg,kuan.yang\}@cs.ox.ac.uk}
}
\and Leslie Ann Goldberg$^\ddag$ \and 
 Daniel \v{S}tefankovi\v{c}\thanks{Department of Computer Science, University of Rochester,
Rochester, NY 14627.   
Research supported in part by NSF grant CCF-1563757. \texttt{stefanko@cs.rochester.edu}} 
 \and Eric Vigoda$^\dag$ \and Kuan Yang$^\ddag$
}

\begin{document}

\maketitle

\begin{abstract}
We consider the problem of sampling from the Potts model on random regular graphs. It is conjectured that sampling is possible when the temperature of the model is in the so-called uniqueness regime of the regular tree, but positive algorithmic results have been for the most part elusive. In this paper, for all integers $q\geq 3$ and $\Delta\geq 3$,  we develop algorithms that produce samples within error $o(1)$ from the $q$-state Potts model on random $\Delta$-regular graphs,  whenever the temperature is in uniqueness, for both the ferromagnetic and antiferromagnetic cases.  

The algorithm for the antiferromagnetic Potts model is based on iteratively adding the edges of the graph and resampling a bichromatic class that contains the endpoints of the newly added edge. Key to the algorithm is how to perform the resampling step efficiently since  bichromatic classes can potentially induce linear-sized components. To this end, we exploit the tree uniqueness to show that the average growth of bichromatic components is typically small, which allows us  to  use correlation decay algorithms for the resampling step. While the precise uniqueness threshold on the tree is not known for general values of $q$ and $\Delta$ in the antiferromagnetic case, our algorithm works throughout uniqueness regardless of its value. 

In the case of the ferromagnetic Potts model, we are able to simplify the algorithm significantly by utilising the random-cluster representation of the model. In particular, we demonstrate that a percolation-type algorithm succeeds in sampling from the random-cluster model with parameters $p,q$ on random $\Delta$-regular graphs for all values of $q\geq 1$ and $p<p_c(q,\Delta)$, where $p_c(q,\Delta)$ corresponds to a uniqueness threshold for the model on the $\Delta$-regular tree. When restricted to integer values of $q$, this yields a simplified algorithm for the ferromagnetic Potts model on random $\Delta$-regular graphs. 
\end{abstract}

\thispagestyle{empty}

\newpage

\setcounter{page}{1}
\section{Introduction}

Random constraint satisfaction problems have been thoroughly studied in computer science in an effort to analyse the limits of satisfiability algorithms and understand the structure of hard instances. Analogously, understanding  spin systems on random graphs \cite{MosselSly, Mossel2009, YZ, BCP, MS22, Dembo2014, GSVY16, Efthymiou, EHSV} gives insights about the complexity of counting and the efficiency of approximate sampling algorithms. In this paper, we design approximate sampling algorithms for the Potts model on \mbox{random regular graphs.}

The Potts model is a fundamental spin system studied in statistical physics and computer science. The model has two parameters: an integer $q\geq 3$, which represents the number of states/colours of the model, and a real parameter $B>0$, which corresponds to the so-called ``temperature''. We denote the set of colours by $[q]:=\{1,\hdots,q\}$. For a graph $G=(V,E)$, configurations of the model are all possible  assignments of colours to the vertices of the graph. Each assignment $\sigma:V\rightarrow [q]$  has a weight $w_G(\sigma)$ which is determined by the number $m(\sigma)$ of monochromatic edges under $\sigma$; namely, $w_G(\sigma)=B^{m(\sigma)}$. The Gibbs distribution $\mu_G$ is defined on the space of all configurations $\sigma$ and is  given by
\[\mu_G(\sigma)=B^{m(\sigma)}/Z_G,\mbox{ where } Z_G=\sum_\sigma B^{m(\sigma)}.\]
We also refer to $\mu_G$ as the Potts distribution; the quantity $Z_G$ is known as the partition function. Well-known models closely related to the Potts model are the  Ising and colourings models. The Ising model is the special case $q=2$ of the Potts model, while the $q$-colourings model is the ``zero-temperature'' case $B=0$ of the Potts model, where the distribution is supported on the set of proper $q$-colourings.

The behaviour of the Potts model has significant differences depending on whether $B$ is less or larger than~1. When $B<1$, configurations where most neighbouring vertices have different colours have large weight and the model is  called \emph{antiferromagnetic}; in contrast, when $B>1$, configurations where most neighbouring vertices have the same colours have large weight and the model is  called \emph{ferromagnetic}. One difference between the two cases that will be relevant later is that the ferromagnetic Potts model admits a random-cluster representation -- the details of this representation are given in Section~\ref{sec:rcdef}.

Sampling from the Potts model is a problem that is frequently encountered in running simulations in statistical physics or inference tasks in computer science.  To determine the efficiency and accuracy of sampling methods, it is relevant to consider the underlying phase transitions, which signify abrupt changes in the properties of the Gibbs distribution when the underlying parameter changes. The so-called uniqueness phase transition captures the sensitivity of the state of a vertex to fixing far-away boundary conditions. As an example, in the case of the ferromagnetic Potts model on the  $\Delta$-regular  tree, uniqueness holds when root-to-leaves correlations  in the Potts distribution vanish as the height of the tree goes to infinity; it is known that this holds  iff $B<B_c(q,\Delta)$, where $B_c(q,\Delta)$ is the  ``uniqueness threshold'' (cf. \eqref{eq:BcqDelta} for its value).  Connecting the uniqueness phase transition with the performance of algorithms is a difficult task that is largely under development. This connection is well-understood on the grid, where it is known that the mixing time of  local Markov chains, such as the Glauber dynamics, switches from polynomial to exponential at the corresponding uniqueness threshold, see for example \cite{MartinelliOlivieri2, MartinelliOlivieri1, ThomasIsing, Alexander1, MOS, Beffara2012}. 

\enlargethispage{\baselineskip}

For random $\Delta$-regular graphs or, more generally, graphs with maximum degree $\Delta$, the uniqueness threshold on the $\Delta$-regular tree becomes relevant. For certain two-state models, such as the ferromagnetic Ising model and the hard-core model, it has been proved that  Glauber dynamics mixes rapidly when the underlying parameter is in the uniqueness regime of the regular tree, and that the dynamics mixes slowly otherwise (see  \cite{MosselSly, Mossel2009, EHSV}).  The same picture is conjectured to hold for the Potts model as well, but this remains open. For the ferromagnetic case in particular, Bordewich, Greenhill, and Patel \cite{BCP} prove rapid and slow mixing results for Glauber dynamics on random regular graphs and graphs with maximum degree $\Delta$ when the parameter $B$ is within a constant factor from the uniqueness threshold on the regular tree. More generally, there has been significant progress the last years in understanding the complexity of sampling from the Gibbs distribution in  two-state systems, but for multi-state systems progress has been slower, especially on the algorithmic side. 

In this paper, for all integers $q\geq 3$ and $\Delta\geq 3$, we design approximate sampling algorithms for the $q$-state Potts model on random $\Delta$-regular graphs (regular graphs with $n$ vertices chosen uniformly at random), when the parameter $B$ lies in the uniqueness regime of the regular tree, for both the ferromagnetic and antiferromagnetic cases. Our algorithms are not based on a Markov chain approach  but proceed by iteratively adding the edges of the graph and performing a resampling step at each stage. As such, our  algorithms can produce samples that are within  error $1/n^{\delta}$ from the Potts distribution for some fixed constant $\delta>0$ (which depends on $B,q,\Delta$). 
\begin{remark}\label{rem:4rf4rr32221}
There are certain ``bad'' $\Delta$-regular graphs where the  algorithms will fail to produce samples with the desired accuracy;  saying that the algorithms work on random $\Delta$-regular graphs means that the number of these ``bad'' graphs with $n$ vertices is a vanishing fraction of all $\Delta$-regular graphs with $n$ vertices for large $n$. Moreover, we can recognise the ``good'' graphs (where our algorithms will successfully produce samples with the desired accuracy) in polynomial time.
\end{remark} 

Our approach is inspired by Efthymiou's algorithm \cite{Efthymiou22, Efthymiou} for sampling $q$-colourings on $G(n,d/n)$; the algorithm there also proceeds by iteratively adding the edges of the graph and exploits the uniqueness on the tree to show that the sampling error is small. However, for the  antiferromagnetic Potts model,  the resampling step turns out to be significantly more involved and we need substantial amount of work to ensure that it can be carried out efficiently, as we explain in detail in Section~\ref{sec:proofapproach}. Nevertheless, for the ferromagnetic case, we manage to give a far simpler algorithm by utilising the random-cluster representation of the model (see Section~\ref{sec:rcdef}). In particular, we demonstrate that a percolation-type algorithm succeeds in sampling approximately from the random-cluster model with parameters $p,q$ on random $\Delta$-regular graphs for all values of $q\geq 1$ and $p<p_c(q,\Delta)$, where $p_c(q,\Delta)$ corresponds to a uniqueness threshold for the model on the $\Delta$-regular tree. When restricted to integer values of $q$, this yields a simple algorithm for the ferromagnetic Potts model on random $\Delta$-regular graphs.

To conclude this introductory section, we remark that, for many antiferromagnetic spin systems 
on random graphs, typical configurations in the Gibbs distribution display absence of long-range correlations 
even beyond the uniqueness threshold, up to the so-called reconstruction threshold \cite{Mezard, GM07}. Note that uniqueness guarantees the absence of long-range correlations under a ``worst-case'' boundary, while non-reconstruction only asserts the absence of long-range correlations under ``typical'' boundaries; it is widely open whether this weaker notion is in fact sufficient for sampling on random graphs. On an analogous note, for the ferromagnetic Potts model on random regular graphs, the structure of typical configurations can be fairly well understood using probabilistic arguments for all temperatures (see, e.g., \cite{Dembo2014,GSVY16}) and it would be very interesting to exploit this structure for the design of sampling algorithms beyond the uniqueness threshold.  In this direction, Jenssen, Keevash and Perkins \cite{JKP} very recently designed such an algorithm for all sufficiently large $B$ that works more generally on expander graphs (see also \cite{HPR} for similar-flavored results on the grid).

\enlargethispage{\baselineskip}

\section{Definitions and Main Results}
We first review in Section~\ref{sec:rcdef} the definition of the random-cluster model. In Section~\ref{sec:uniqueness}, we state results from the literature about uniqueness on the regular tree for the Potts and random-cluster models. Then, in Section~\ref{sec:resultsferro}, we state our algorithmic results for the ferromagnetic Potts and random-cluster models  and, in Section~\ref{sec:resultsantiferro}, our result for the antiferromagnetic Potts model.

\subsection{The random-cluster model}\label{sec:rcdef} 
The random-cluster model has two parameters $p\in [0,1]$ and $q>0$; note that $q$ in this case can take non-integer values. For a graph $G=(V,E)$, we denote the random-cluster distribution on $G$ by $\varphi_G$; this distribution is supported on the set of all edge subsets. In particular, for $S\subseteq E$, let   $k(S)$ be the number of connected components in the graph $G'=(V,S)$ (isolated vertices do count).  Then, the weight of the configuration $S$ is given by $w_G(S)=p^{|S|}(1-p)^{|E\backslash S|}q^{k(S)}$ and 
\[
\varphi_G(S) = w_G(S)/Z^{\mathrm{rc}}_G, \mbox{ where } Z^{\mathrm{rc}}_G=\sum_{S\subseteq E} w_G(S)=\sum_{S\subseteq E}p^{|S|}(1-p)^{|E\backslash S|}q^{k(S)}.
\]
Following standard terminology, each edge in $S$ will be called open, while each edge in $E\backslash S$ closed. 

For integer values of $q$, there is a well-known connection between the  random-cluster and ferromagnetic Potts model, as detailed below. 
\begin{lemma}[see, e.g., \cite{Grimmettbook}]\label{lem:rctopotts}
Let $q\geq 2$ be an integer, $B>1$, and $p=1-1/B$. Then, the following hold for any graph  $G=(V,E)$.
\begin{itemize}
\item Let $S\subseteq E$ be distributed according to the RC distribution $\varphi_G$ with parameters $p,q$. Consider the configuration $\sigma$ obtained from $S$ by assigning each component in the graph $(V,S)$ a random colour from $[q]$ independently. Then, $\sigma$ is distributed according to the Potts distribution $\mu_G$ with parameter $B$. 
\item Conversely, suppose that $\sigma:V\rightarrow [q]$ is distributed according to the Potts distribution $\mu_G$ with parameter $B$.  Consider $S\subseteq E$ obtained by adding to $S$ each monochromatic edge under $\sigma$ with probability $p$ independently. Then,  $S$ is distributed according to the RC distribution $\varphi_G$ with parameters $p,q$.
\end{itemize}
\end{lemma}
Lemma~\ref{lem:rctopotts} allows us to translate our sampling algorithm for the random-cluster model to a sampling algorithm for the ferromagnetic Potts model. The benefit of working with the random-cluster model (instead of the Potts model) is that the random-cluster distribution satisfies certain monotonicity properties (cf. Lemma~\ref{lem:monotonicity}) which simplifies significantly the analysis of the algorithm.

\subsection{Uniqueness for Potts and random-cluster models on the tree}\label{sec:uniqueness}
In this section, we review uniqueness on the tree for the Potts and random-cluster models. We start with the Potts model.  

For a configuration $\sigma$ and a set $U$, we denote by $\sigma_U$ the restriction of $\sigma$ to the set $U$; in the case of a single vertex $u$, we simply write $\sigma_u$ to denote the colour of $u$. Denote by $\TreeD$  the infinite $(\Delta-1)$-ary tree with root vertex $\rho$ and, for an integer $h\geq 0$, denote  by $T_h$ the subtree of $\TreeD$ induced by the vertices at distance $\leq h$ from $\rho$. Let $L_h$ be the set of leaves of $T_h$.
\begin{definition} \label{def:uniqueness}
Let $B>0$ and $q,\Delta\geq 3$ be integers. The $q$-state Potts model with parameter $B>0$ has  \emph{uniqueness}  on the infinite $(\Delta-1)$-ary tree if, for all colours $c\in [q]$,  it holds that
\begin{equation}\label{eq:4vvft4vtg}
\limsup_{h\to\infty} \max_{\tau: L_h \to [q]}\Big|\mu_{T_h}\big(\sigma_\rho = c\mid\sigma_{L_h}=\tau\big)-\frac{1}{q}\Big|=0.
\end{equation}
It has \emph{non-uniqueness} otherwise.
\end{definition}

For the ferromagnetic $q$-state Potts model $(B>1)$, it is known that uniqueness holds on the $(\Delta-1)$-ary tree iff $B<B_c(q,\Delta)$, where
\begin{equation}\label{eq:BcqDelta}
B_c(q,\Delta)=1+\inf_{y> 1} h(y),\mbox{ where } h(y):=\frac{(y - 1) (y^{\Delta-1} + q - 1)}{y^{\Delta-1} - y}.
\end{equation}

For the antiferromagnetic Potts model ($B<1$), the uniqueness threshold on the tree is not yet known in full generality. It is a folklore conjecture that the model has uniqueness iff $q\geq \Delta$ and $B\in (0,1)$, or  $q<\Delta$ and  $B\geq \frac{\Delta-q}{\Delta}$. It is known that the model has non-uniqueness when $B<\frac{\Delta-q}{\Delta}$ \cite{Galanis}. Establishing the uniqueness side of the conjecture is more difficult; this has been established recently in \cite{GGY} for small values of $q$ and $\Delta$. In the case $q=3$, \cite{GGY} also established the uniqueness threshold for all $\Delta$: for $\Delta\geq 4$, uniqueness holds iff  $B\in[ (\Delta-3)/\Delta,1)$ and, for $\Delta=3$ uniqueness holds iff $B\in (0,1)$.  For the $q$-colourings model $(B=0)$, Jonasson \cite{Jonasson}, building on work of Brightwell and Winkler \cite{Brightwell}, established  that the model has uniqueness iff $q>\Delta$.

\begin{remark}\label{rem:6gtf5q}
To summarise the above, a necessary condition for uniqueness on $\TreeD$ in the antiferromagnetic $q$-state Potts model with parameter $B\in (0,1)$ is that $B\geq (\Delta-q)/\Delta$. It is also conjectured that this condition is sufficient but this has only been established for $q=3$. 
\end{remark}

Uniqueness for the random-cluster model on the tree is less straightforward to define. H{\"a}ggstr{\"o}m \cite{Haggstrom} studied uniqueness of random-cluster measures on the infinite $(\Delta-1)$-ary tree where all infinite components are connected ``at infinity'' -- we review his results in more detail in Section~\ref{sec:unique}. He showed that, for all $q\geq 1$,  a sufficient condition for uniqueness is that $p<p_c(q,\Delta)$, where the critical value $p_c(q,\Delta)$ is given by\footnote{In \cite{Haggstrom}, $p_c(q, \Delta)$ is defined in a different way, but the two definitions are equivalent for all $q\geq 1$.}:
\begin{equation}\label{eq:pcq}
p_c(q,\Delta)=1-\frac{1}{1+\inf_{y> 1} h(y)},\mbox{ where } h(y):=\frac{(y - 1) (y^{\Delta-1} + q - 1)}{y^{\Delta-1} - y}.
\end{equation}
Note that the critical values in \eqref{eq:BcqDelta} and \eqref{eq:pcq} are connected for integer values of $q$ via  $p_c(q,\Delta)=1-\frac{1}{B_c(q,\Delta)}$. H{\"a}ggstr{\"o}m \cite{Haggstrom} also conjectured that uniqueness for the random-cluster model holds on $\TreeD$ when $p>\frac{q}{q+\Delta-2}$ for all $q\geq 1$; this remains open but progress has been made in \cite{JonassonRC}.
\begin{remark}\label{rem:6gtf5qb}
Using that $h(y)\ra \frac{q}{\Delta-2}$ as $y\downarrow 1$, we  obtain that $p_c(q,\Delta)\leq \frac{q}{q+\Delta-2}$ for all $q\geq 1$. It can further be shown that $p_c(q,\Delta)= \frac{q}{q+\Delta-2}$ iff $q\leq 2$. 
\end{remark}

We should note here that the bounds appearing in Remarks~\ref{rem:6gtf5q} and~\ref{rem:6gtf5qb} will be useful to simplify some of our arguments. However, the stronger assumption that the parameters are in the uniqueness region of the $\Delta$-regular tree are crucial for the analysis of our algorithm and, in particular, the proofs of the upcoming Lemmas~\ref{lem:treespath} and Lemmas~\ref{lem:treespathPotts} hinge on this assumption.

\subsection{Sampling ferro Potts and random-cluster models on random regular graphs}\label{sec:resultsferro}
We begin by stating our result for the random-cluster model on random regular graphs.
\begin{theorem}\label{thm:main2}
Let $\Delta\geq 3$, $q\geq 1$ and $p<p_c(q,\Delta)$. Then, there exists a constant $\delta>0$ such that, as $n\rightarrow \infty$, the following holds   with probability $1-o(1)$ over the choice of a random $\Delta$-regular graph $G=(V,E)$ with $n$ vertices.

There is a polynomial-time algorithm which, on input the graph $G$, outputs a random set $S\subseteq E$ whose distribution $\nu_S$ is within total variation distance $O(1/n^{\delta})$ from the RC distribution $\varphi_G$ with parameters $p,q$, i.e., $\norm{\nu_S-\varphi_{G}}_{\TV}=O(1/n^{\delta})$.
\end{theorem}
We remark that a simple implementation of the algorithm in Theorem~\ref{thm:main2} runs in time $O^*(n^{6/5})$, see Figure~\ref{alg:perc} for details. The constant $\delta>0$ that controls the error of the algorithm depends on $p,q,\Delta$ and gets smaller as $p$ approaches $p_c(q,\Delta)$.

For integer values of $q$, Theorem~\ref{thm:main2} combined with the translation between the random-cluster and Potts models (cf. Lemma~\ref{lem:rctopotts}) yields  a sampling algorithm for the ferromagnetic $q$-state Potts model on random regular graphs. Since uniqueness for the ferromagnetic Potts model holds iff $B<B_c(q,\Delta)$ and $p_c(q,\Delta)=1-\frac{1}{B_c(q,\Delta)}$, we therefore have the following corollary of Theorem~\ref{thm:main2}.
\begin{corollary}\label{thm:mainferroPotts}
Let $\Delta\geq 3$, $q\geq 3$ and $B>1$ be in the uniqueness regime of the $(\Delta-1)$-ary tree. Then, there exists a constant $\delta>0$ such that, as $n\rightarrow \infty$, the following holds   with probability $1-o(1)$ over the choice of a random $\Delta$-regular graph $G=(V,E)$ with $n$ vertices.

There is a polynomial-time algorithm which, on input the graph $G$,  outputs a random  assignment $\sigma:V\rightarrow [q]$ whose distribution $\nu_\sigma$ is within total variation distance $O(1/n^{\delta})$ from the Potts distribution $\mu_G$ with parameter $B$, i.e., $\norm{\nu_\sigma-\mu_{G}}_{\TV}=O(1/n^{\delta})$.
\end{corollary}

\subsection{Sampling antiferro Potts on random $\Delta$-regular graphs}\label{sec:resultsantiferro}
The algorithm of Corollary~\ref{thm:mainferroPotts} for the ferromagnetic Potts model does not extend to the antiferromagnetic case since there is no analogous connection with the random-cluster model in this case. Nevertheless, we are able to design a sampling algorithm on random regular graphs when the parameter $B$ is in uniqueness via a far more elaborate approach which consists of recolouring (large) bichromatic colour classes.
\begin{theorem}\label{thm:mainPotts}
Let $\Delta\geq 3$, $q\geq 3$ and $B\in(0,1)$ be in the uniqueness regime of the $(\Delta-1)$-ary tree with $B\neq (\Delta-q)/\Delta$. Then, there exists a constant $\delta>0$ such that, as $n\rightarrow \infty$, the following holds   with probability $1-o(1)$ over the choice of a random $\Delta$-regular graph $G=(V,E)$ with $n$ vertices.

There is a polynomial-time algorithm which, on input the graph $G$,  outputs a random  assignment $\sigma:V\rightarrow [q]$ whose distribution $\nu_\sigma$ is within total variation distance $O(1/n^{\delta})$ from the Potts distribution $\mu_G$ with parameter $B$, i.e.,
\[\norm{\nu_\sigma-\mu_{G}}_{\TV}=O(1/n^{\delta}).\]
\end{theorem}
Note that the sampling algorithm in the antiferromagnetic case works throughout uniqueness apart from the point $(\Delta-q)/\Delta$, where uniqueness on the tree is expected to hold but the model is conjectured to be at criticality. In particular, for all $B\neq (\Delta-q)/\Delta$ which are in the uniqueness regime of $\TreeD$, it can be shown that the decay on the tree is exponentially small in its height. In contrast,    even if uniqueness on the tree holds for $B=(\Delta-q)/\Delta$, it can be shown that the decay on the tree is only polynomial in its height.

We remark here that the algorithm for the antiferromagnetic case uses as a black-box a subroutine for sampling from the antiferromagnetic Ising model. The running time of this subroutine, which is based on correlation decay methods, is $n^c$ for some constant $c=c(q,B,\Delta)>0$; it is an open question whether there is a faster algorithm for the Ising model. Finally, the constant $\delta$ that controls the error of the algorithm depends on $B,q,\Delta$ and gets smaller as $B$ decreases.

\section{Proof Approach}\label{sec:proofapproach}
In this section, we outline the main idea behind the algorithms of Theorems~\ref{thm:main2} and~\ref{thm:mainPotts}, and the key obstacles that we have to address.  We focus on  the antiferromagnetic Potts model where the details are much more complex and discuss how we get the simplification for the ferromagnetic case via the random-cluster model later.

\begin{definition}\label{def:short}
For an $n$-vertex graph $G=(V,E)$ with maximum degree $\Delta$, a cycle is \emph{short} if its length is at most $\frac{1}{5}\log_{\Delta-1} n$, and is \emph{long} otherwise.
\end{definition}

Let $G$ be a random $\Delta$-regular graph with $n$ vertices.  Following the approach of Efthymiou \cite{Efthymiou}, our algorithm starts from the subgraph of $G$ consisting of all short cycles, which we denote by $G'$. It is fairly standard to show that, with probability $1-o(1)$ over the choice of $G$, the subgraph $G'$ is a \emph{disjoint} union of  short cycles, see Lemma~\ref{lem:disjointcycles}. It is therefore possible to sample a configuration $\sigma'$ on $G'$ which is distributed according to the Potts distribution $\mu_{G'}$ (exactly). This can be accomplished in several ways;  in fact, since the cycles are disjoint and each cycle has logarithmic length, this initial sampling step can even be done  via brute force in polynomial time (though it is not hard to come up with much faster algorithms). 

After this initial preprocessing, the algorithm then proceeds by adding sequentially the edges that do not belong to short cycles. At each step,  the current configuration is updated with the aim to  preserve its distribution close to the Potts distribution of the new graph (with the edge that we just added). Key to this update procedure is a resampling step which is performed only when the endpoints of a newly added edge $\{u,v\}$ happen to have the same colours under the current configuration; intuitively, some action is required in this case because the weight of the current configuration reduces by a factor of $B<1$ in the new graph (because of the added edge). The resampling step consists of recolouring a bichromatic class, where the latter is defined as follows. 
\begin{definition}
Let $G=(V,E)$ be a graph and $\sigma:V\rightarrow[q]$ be a configuration. For colours $c_1,c_2\in [q]$, let $\sigma^{-1}(c_1,c_2)$ be the set of vertices that have either colour $c_1$ or colour $c_2$ under $\sigma$.

For distinct colours $c_1,c_2\in [q]$, we say that $U=\sigma^{-1}(c_1,c_2)$ is the $(c_1,c_2)$-colour-class of $\sigma$ and that $U$ is a bichromatic class under $\sigma$. We refer to a connected component of $G[U]$ as a bichromatic component.
\end{definition}

In the proper colourings case ($B=0$), Efthymiou \cite{Efthymiou} demonstrated that the resampling step when adding an edge $e=\{u,v\}$ can be done  by just flipping the colours of a bichromatic component chosen uniformly at random among those containing one of the vertices $u$ and $v$ (say $u$). The rough idea there is that, when the colourings model is in uniqueness,  the bichromatic components on a random graph are typically small in size. At the same time, by the initial preprocessing step, the edge $e=\{u,v\}$ does not belong to a short cycle and therefore $u$ and $v$ are far away in the graph without $e$. Hence, $u$ and $v$ are unlikely to belong to the same bichromatic component and the flipping step will succeed in  giving $u$ and $v$ different colours with good probability. 

Unfortunately, this flipping method does not work for the antiferromagnetic Potts model. It turns out that when $q<\Delta$ and  even when the Potts model is in uniqueness, bichromatic components can be large and therefore $u$ and $v$ potentially belong to the same bichromatic component. To make matters worse, these bichromatic components can be quite complicated (with many short/long cycles). This necessitates a more elaborate approach in our setting to succeed in giving $u$ and $v$ different colours without introducing significant bias to the sampler. 

The key to overcoming these obstacles lies in the observation that  the assignment of the two colours in a bichromatic component follows the Ising distribution, see Observation~\ref{obs:wrvtvte} for the precise formulation. Hence we can hope to use an approximate sampling algorithm for the Ising model in the resampling step. The natural implementation of this idea however fails: known algorithms for the antiferromagnetic Ising model, based on correlation decay,  work as long as $B>\frac{\Delta-2}{\Delta}$, where $\Delta$ is the maximum degree of the graph \cite{Sinclair1234, ZHANG}. In general, this inequality is not satisfied for us, i.e, there exist $B$ in the uniqueness regime for Potts such that $B<\frac{\Delta-2}{\Delta}$. 

Fortunately, we can employ fairly recent technology  for two-state models \cite{MosselSly,Sinclairconn, Sinclair2017} which demonstrates that the graph parameter that  matters is not actually the maximum degree of the graph  but rather the ``average growth'' of the graph.   While we cannot apply any of the existing results in the literature directly, adapting these ideas to the antiferromagnetic Ising model is fairly straightforward, using  results from Mossel and Sly \cite{MosselSly}. The more difficult part in our setting is  proving that the average growth of the bichromatic components that we consider for resampling is indeed small for ``typical'' configurations $\sigma$ (note that in the worst case, the whole graph can be a bichromatic class which has  large average growth for our purposes, so a probability estimate over $\sigma$ is indeed due). Let us first formalise the notion of average growth that we use.
\begin{definition}\label{def:branchingfactor}
Let $M,b$ be positive constants and  $G=(V,E)$ be a graph with $n$ vertices. We say that $G$ has average growth $b$ up to depth $L=\left\lceil M\log n\right\rceil$ if for all vertices $v\in V$ the total number of paths with $L$ vertices  starting from $v$ is less than $b^L$.
\end{definition}
The notion of average growth is similar to the notion of connective constant for finite graphs used in \cite{Sinclairconn, Sinclair2017}, the reason for the slightly different definition is that we will need an explicit handle on the constant $M$ controlling the depth.  Note that, since we only consider paths with a fixed logarithmic length, this places a lower bound on the accuracy of the sampling algorithm. Nevertheless, by choosing the constant $M$ sufficiently large, this will still be sufficient to make the error of our sampler polynomially small. In particular, as long as the inequality $b \frac{1-B}{1+B}<1$ is satisfied, using results from \cite{MosselSly}, we obtain an approximate sampler for the antiferromagnetic Ising model with parameter $B$ on graphs of average growth $b$ up to depth $L=\left\lceil M\log n\right\rceil$,  see Theorem~\ref{thm:sampleising} for details.

We give a few more  technical details on how we bound the average growth of bichromatic classes. Here, we utilise the tree uniqueness and the tree-like structure of random $\Delta$-regular graphs (cf. Lemma~\ref{lem:graphneighbor}) to provide an upper bound on the number of  bichromatic paths. For paths of logarithmic length $L$, we show in Lemma~\ref{lem:treespathPotts} that the probability that a path is bichromatic is $\leq K^L$, where  $K$ is roughly $(1+B)/(B+q-1)$. Since there are at most $\Delta(\Delta-1)^{L-2}$ paths with $L$ vertices, we therefore obtain that the average growth $b$ of bichromatic components  is bounded above by $(\Delta-1)K$. When $B$ is in uniqueness, we have that $B>(\Delta-q)/\Delta$, and therefore the inequality $b  \frac{1-B}{1+B}<1$ that is required for the Ising sampler is satisfied. 

The final technical piece is to bound the error that is introduced by the resampling steps. The placement of a new edge $\{u,v\}$ reweights the probability that $u$ and $v$ have different colours and introduces an error in our sampling algorithm that is captured by the correlation between the colours of $u$ and $v$ (see Lemma~\ref{lem:samplinganti}). The main idea at this point is that, in the graph without the edge $\{u,v\}$, $u$ and $v$ are far apart (since the edge $\{u,v\}$ does not belong to a short cycle of $G$) which can be used to show  that the correlation between the colours of $u$ and $v$ is relatively small. In Lemma~\ref{lem:corrdecayIsing}, we show that  the correlation between $u$ and $v$ can in fact be upper bounded as a weighted sum over  paths connecting $u$ and $v$. This allows us to bound the aggregate sampling error of the algorithm as a weighted sum over short cycles of the random graph $G$, which can  in turn be bounded using a simple expectation argument.

The algorithm that we described for the antiferromagnetic Potts model can actually be adapted to the ferromagnetic case as well. However, as mentioned earlier, we follow a different (and surprisingly simpler) route using  the random-cluster representation of the model. At a very rough level, the reason behind the simplification is that the components in the random-cluster  model provide a much better grip on capturing the properties of the Potts distribution than the bichromatic-component proxy we used earlier.   Indeed, just as we described in the antiferromagnetic case, bichromatic components for the ferromagnetic Potts model can also be linear-sized. However, once we translate the Potts configuration to its random-cluster representation (cf. Lemma~\ref{lem:rctopotts}), the components in the latter are small in size (when the model is in the uniqueness region $p<p_c(q,\Delta)$) and therefore vertices that are far away do not belong to the same component. This allows us to perform the resampling step in the random-cluster model by a simple percolation procedure. The details can be found in Section~\ref{sec:rcalgo}.

\subsection{Organisation}
In Section~\ref{sec:randomregular}, we give the properties of random regular graphs that we are going to use in the analysis of our sampling algorithms. In Section~\ref{sec:rcalgo}, we give the algorithm for the random-cluster model and conclude Theorem~\ref{thm:main2} (assuming the upcoming Lemma~\ref{lem:treespath}). In Section~\ref{sec:tb34bb6}, we give the algorithm for the antiferromagnetic Potts model and conclude Theorem~\ref{thm:mainPotts} (assuming the upcoming Theorem~\ref{thm:sampleising} and Lemmas~\ref{lem:corrdecayIsing} and~\ref{lem:treespathPotts}). In Sections~\ref{sec:qaz234} and~\ref{sec:teb3tbb6b5b5}, we analyse the random-cluster and antiferromagnetic Potts models on ``tree-like'' graphs and give the proofs of Lemmas~\ref{lem:treespath} and~\ref{lem:treespathPotts}. Finally, in Section~\ref{sec:egbe5tb4hh56}, we prove Theorem~\ref{thm:sampleising} and Lemma~\ref{lem:corrdecayIsing} which are about correlation decay and sampling for the antiferromagnetic Ising model on graphs of small average growth.

\section{Properties of random regular graphs}\label{sec:randomregular}
In this section, we state and prove structural properties of random $\Delta$-regular graphs which ensure that our algorithms for the random-cluster and Potts models have the desired accuracy (cf. Remark~\ref{rem:4rf4rr32221}).  While the exact statements of the properties that we need do not seem to appear in the literature, their proofs follow fairly standard techniques in the area. 

We will work in the configuration model, see \cite[Chapter 9]{janson2011random} for more details. Precisely, for $\Delta n$ even, let $\Gc:=\Gc_{n,\Delta}$ denote the uniform distribution on  $\Delta$-regular graphs  which is obtained by taking a perfect matching of the set $[n]\times [\Delta]$ and collapsing for each $u\in [n]$ the elements $(u,1),\hdots, (u,\Delta)$ into a single vertex $u$; the elements of the set $[n]\times [\Delta]$ are called \emph{points}. 
Technically, the distribution  $\Gc$ is supported on multigraphs but it can be shown that the probability that $G\sim \Gc$ is  simple is asymptotically a positive constant as $n\rightarrow \infty$; conditioned on that event, $G$ is uniformly distributed over $\Delta$-regular graphs with $n$ vertices, and therefore any event that holds with probability $1-o(1)$ in $\Gc_{n,\Delta}$ also holds with probability $1-o(1)$ over the uniform distribution on $\Delta$-regular graphs with $n$ vertices. 

The following lemma guarantees that short cycles are disjoint in a random $\Delta$-regular graph. 
\begin{lemma}\label{lem:disjointcycles}
Let $\Delta\geq 3$ be an integer. Then, with probability $1-o(1)$ over the choice of a uniformly random $\Delta$-regular graph  with $n$ vertices, any two distinct cycles of length $\leq \frac{1}{5} \log_{\Delta-1} n$ are disjoint, i.e., they do not share any common vertices or edges.
\end{lemma}
\begin{proof}
For convenience, let $\ell:=\left\lfloor \frac{1}{5} \log_{\Delta-1} n\right\rfloor$ and $\Gc:=\Gc_{n,\Delta}$. Let $G\sim\Gc$ and $\mathcal{E}$ be the event that $G$ contains two distinct cycles of length $\leq \ell$ which are not disjoint. Let also $\mathcal{F}$ be the event that,  for some integer $k\in[1,2\ell]$, $G$ contains a subgraph with $k$ vertices and $k+1$ edges such that each vertex has at least two incident multiedges (in the subgraph); we call such a subgraph \emph{bad}. When $\mathcal{E}$ occurs, we obtain that $\mathcal{F}$ also occurs, so it suffices to upper bound the probability of the latter.

To bound $\mbox{$\Pr$}_{\Gc}(\mathcal{F})$, we will use a union over all possible bad subgraphs with $k$ vertices with $k\in [1,2\ell]$. Consider an arbitrary integer $k\in [1,2\ell]$. There are $\binom{n}{k}\leq (\frac{\emm n}{k})^{k}$ ways to choose the vertices in the subgraph, at most $\binom{\Delta}{2}^{k}(\Delta k)^2$ ways to choose the points that are going to be paired and then $\frac{(2(k+1))!}{(k+1)!2^{k+1}}\leq 10(2k/\emm)^{k+1}$ ways to pair the points. The probability of a particular pairing of  $2(k+1)$ points occuring is equal to $\frac{1}{(\Delta n-1)\cdots (\Delta n-(2k+1))}\leq \frac{2}{(\Delta n)^{k+1}}$ (using that $k\leq 2\ell=o(n^{1/2})$). We therefore obtain that 
\[\mbox{$\Pr$}_{\Gc}(\mathcal{F})\leq 20\sum^{2\ell}_{k=1} \binom{n}{k}\frac{(\Delta k)^2\binom{\Delta}{2}^{k}(2k/\emm)^{k+1}}{(\Delta n)^{k+1}}\leq \frac{60\Delta}{n}\sum^{2\ell}_{k=1}k^2 (\Delta-1)^k\leq \frac{500\Delta \ell^3}{n} (\Delta-1)^{2\ell}=o(1),\]
as needed. This finishes the proof.
\end{proof}

The following lemma guarantees that certain weighted sums over cycles are small; this bound will be used to show that the aggregate error of our samplers is small (cf. Section~\ref{sec:proofapproach}).
\begin{lemma}\label{lem:numpaths}
Let $\Delta\geq 3$. Then,  for any constant $W>\Delta-1$ and any constant $\ell_0>0$, there exists a constant $\delta>0$  such that the following holds with probability $1-O(1/n^\delta)$ over the choice of $G\sim\Gc_{n,\Delta}$. Let $C_\ell$ denote the number of cycles of length $\ell$. Then, 
\[\sum_{\ell\geq \ell_0 \log n}\frac{\ell C_\ell}{W^{\ell}}\leq 1/(2n^{\delta}).\]
\end{lemma}
\begin{proof}
Let $w:=(\Delta-1)/W<1$ and $\delta:=\frac{1}{2}\ell_0\log(1/w)>0$.  We will show that, for all sufficiently large $n$, it holds that 
\begin{equation}\label{eq:qwecrvet}
\Eb_{\Gc}[X]\leq \frac{1}{(1-w)n^{2\delta}}, \quad \mbox{ where } X:=\sum_{\ell\geq \ell_0 \log n}\frac{\ell C_\ell}{W^{\ell}}.
\end{equation}
Once we show \eqref{eq:qwecrvet}, the result follows from Markov's inequality.

Fix an arbitrary integer $\ell\in [n]$; we will calculate $\Eb_{\Gc}[C_\ell]$.   In the configuration model, a cycle with $\ell$ vertices corresponds to an ordered $2\ell$-tuple of points $(u_1,i_1),(u_1,i_2),\hdots, (u_\ell,i_{2\ell-1}), (u_\ell,i_{2\ell})$ where $u_1,\hdots,u_\ell$ are distinct elements of $[n]$ and $i_1,\hdots,i_{2\ell}\in [\Delta]$ such that $(u_j,i_{2j})$ is paired to $(u_{j+1},i_{2j+1})$ for all $j\in[\ell]$ (with the convention that $u_{\ell+1}=u_1$ and $i_{2l+1}=i_1$). 

There are $\ell!\binom{n}{\ell}$ ways to choose and order $u_1,\hdots,u_\ell$ and $(\Delta(\Delta-1))^{\ell}$ ways to choose $i_1,\hdots,i_{2\ell}$ for a total of $\ell!\binom{n}{\ell}(\Delta(\Delta-1))^{\ell}$ possible tuples; this overcounts the number of tuples corresponding to distinct cycles by a factor of $2\ell$ (the number of ways to root and orient the $2\ell$-tuple). Now, the pairing corresponding to a tuple occurs with probability $\frac{1}{(\Delta n-1)(\Delta n-3)\cdots(\Delta n-(2\ell-1))}$. Since 
\[2(\Delta n-1)(\Delta n-3)\cdots(\Delta n-(2\ell-1))\geq \Delta^{\ell} n(n-1)\cdots (n-\ell+1)\]
for all $\ell\in [1,n]$ and $\Delta\geq 3$, it follows that $\Eb_{\Gc}[C_\ell]\leq (\Delta-1)^{\ell}/\ell$ and hence
\[ \Eb_{\Gc}[X]=\sum_{\ell\geq \ell_0\log n}w^{\ell}\leq w^{\ell_0\log n}/(1-w),\]
which proves \eqref{eq:qwecrvet} and therefore concludes the proof.
\end{proof}

Our next lemma captures the tree-like structure of random $\Delta$-regular graphs that will be relevant for us. In particular, we give a description of the neighbourhood structure around a path. To do this accurately, we will need a few definitions. Let $G=(V,E)$ be a graph. For a vertex $v\in V$ and integer $h\geq 0$, we denote by $\Gamma_h(G,v)$ the set of vertices at distance $\leq h$ from $v$.
\begin{definition}\label{def:hpathneigh}
Let $G$ be a graph and $P$ be a path in $G$ with vertices $u_1,\hdots, u_\ell$. Let $G\backslash P$ be the graph obtained from $G$ by removing the edges of the path $P$. Then, for an integer $h\geq 0$, the \emph{$h$-graph-neighbourhood} of the path $P$ is the subgraph of $G\backslash P$ induced by the vertex set $\bigcup_{i\in [\ell]}\Gamma_h(G\backslash P, u_i)$.  

A connected component of the $h$-graph-neighbourhood will be called \emph{isolated} if it contains \emph{exactly one} of the vertices $u_1,\hdots,u_\ell$. 
\end{definition}

\begin{lemma}\label{lem:graphneighbor}
Let $\Delta\geq 3$. Then,  for any constant integer $h\geq 0$ and any $\epsilon>0$, there exists a constant $\ell_1>0$ such that the following holds. 

With probability $1-O(1/n^2)$ over the choice of $G\sim\Gc_{n,\Delta}$, every path  $P$ in $G$ with $\ell$ vertices with $\ell_1\leq \ell\leq n^{9/10}$  has an $h$-graph-neighbourhood with at least $(1-\epsilon)\ell$ isolated tree components.
\end{lemma}

We will use the following version of the well-known Chernoff/Hoeffding inequality.
\begin{lemma}[see, e.g., {\cite[Theorem 21.6 \& Corollary 21.9]{FK}}]\label{lem:chernoff}
Suppose that $S_n=X_1+\cdots+X_n$, where $\{X_i\}_{i\in[n]}$ is a collection of independent random variables such that $0\leq X_i\leq 1$ and $\E[X_i]=\mu_i$ for $i=1,\hdots,n$. Let $\mu=\mu_1+\cdots+\mu_n$. Then, for any $c>1$, 
\[\Pr(S_n\geq c\mu)\leq \exp\big(-c\mu\log(c/\emm)\big).\]
\end{lemma}

\begin{proof}[Proof of Lemma~\ref{lem:graphneighbor}]
Fix an arbitrary integer $h\geq 0$ and constant $\epsilon>0$. We will show that the lemma holds with $\ell_1:=200/\epsilon$. 

In the configuration model, a path $P$ with $\ell$ vertices corresponds to an ordered $2(\ell-1)$-tuple of points $(u_1,i_1),(u_2,i_2),(u_2,i_3),\hdots, (u_{\ell-1},i_{2(\ell-2)}),(u_{\ell-1},i_{2\ell-3}),(u_\ell,i_{2(\ell-1)})$ where $u_1,\hdots,u_\ell$ are distinct elements of $[n]$ and $i_1,\hdots,i_{2(\ell-1)}\in [\Delta]$ such that $(u_j,i_{2j-1})$ is paired to $(u_{j+1},i_{2j})$ for all $j\in[\ell-1]$.  Fix any such path $P$ with $\ell$ vertices $u_1,\hdots, u_\ell$ and condition on the event that $P$ appears in $G$. We will next reveal the $h$-graph-neighbourhood of $P$ in a breadth-first search manner, as follows. 
\begin{figure}[h]
\noindent 1. Let $U_0=\{u_1,\hdots,u_\ell\}$. Initialize $U_1,\hdots, U_h, U_{h+1}$ to be empty sets of vertices.\\
\noindent 2. For $t=0,1,\hdots,h$:\\
\noindent 3. \phantom{\ \ \ \ } Order the vertices in $U_t$ in lexicographic order.\\
\noindent 4. \phantom{\ \ \ \ }  For $k=1,\hdots, \ell \Delta^{t}$:\\
\noindent 5. \phantom{\ \ \ \ \ \ \ } Pick the $k$-th vertex in $U_t$, say $u$ (if $k>|U_t|$, set $u=0$).\\
\noindent 6. \phantom{\ \ \ \ \ \ \ }  For $i=1,\hdots,\Delta$:\\
\noindent 7. \phantom{\ \ \ \ \ \ \ \ \ \ \ \ } If $u\neq 0$ and the point $(u,i)$ is not already paired, \\
\noindent 8. \phantom{\ \ \ \ \ \ \ \ \ \ \ \ \ \ \ \ \ } Pair $(u,i)$ with a point not already paired, say $(v,i')$, selected uniformly at random.\\ 
\noindent 9. \phantom{\ \ \ \ \ \ \ \ \ \ \ \ \ \ \ \ \ } If $v\notin U_1\cup \cdots \cup U_t\cup U_{t+1}$, add $v$ to $U_{t+1}$. 
\end{figure}

By induction, we have that, for all $t\geq 0$, $U_t$ consists of the set of vertices at distance $t$ from a vertex in $\{u_1,\hdots, u_\ell\}$. Note also that  $|U_t|\leq \ell \Delta^{t}$. Now, fix arbitrary $t\in \{0,1,\hdots,h\}$, $k\in[\ell \Delta^{t}]$ and $i\in [\Delta]$. Let $\mathcal{F}_{t,k,i}$ be the pairings that we have revealed about the graph $G$ just before executing lines 7--9. Similarly, let $S_{t,k,i}$ be  the set of vertices we have encountered just before executing lines 7--9 (i.e., the union of $U_1,\hdots, U_t$ together with the current set $U_{t+1}$). Let also $Q_{t,k,i}$ be the event that in lines 7-8 all of the following happen:  (i) in line 7, $u\neq 0$ and the point $(u,i)$ is not paired, and (ii) in line 8, $(u,i)$ gets paired to a point in $S_{t,k,i}\times [\Delta]$. There are at least $\Delta (n- |S_{t,k,i}|)$ points that have not been paired, so
\begin{equation}\label{eq:qtki}
\Pr(Q_{t,k,i}\mid \mathcal{F}_{t,k,i} )\leq \frac{|S_{t,k,i}|}{n-|S_{t,k,i}|}\leq \frac{2\ell \Delta^{h+2}}{n}, 
\end{equation}
where in the last inequality we used that 
\[|S_{t,k,i}|\leq |U_0|+\cdots +|U_{t+1}|\leq \ell \Delta^{t+2}\leq \ell \Delta^{h+2}\mbox{ and } n-\ell \Delta^{h+2}\geq n/2\]
for all sufficiently large $n$ (using that $\ell_1\leq \ell\leq n^{9/10}$). Using \eqref{eq:qtki}, we obtain that the number of events $\{Q_{t,k,i}\}$ that occur is dominated above by a binomial r.v. $X_\ell\sim \mathrm{Bin}(\ell \Delta^{h+2},\frac{2\ell \Delta^{h+2}}{n})$. By Lemma~\ref{lem:chernoff},  we have that 
\[\Pr[X_\ell\geq \epsilon \ell/2 ]\leq \emm^{-\frac{1}{2}\epsilon \ell \log (c_\ell/\emm)}, \mbox{ where } c_\ell=\frac{\epsilon n}{4\ell \Delta^{2h+4}}.\]
Since $\ell\leq  n^{9/10}$, we have $c_\ell\geq  \emm n^{1/20}$ for all sufficiently large $n$ and hence 
\[\Pr[X_\ell\geq \epsilon \ell ]\leq \emm^{-\frac{1}{40}\epsilon \ell \log n}.\]
It follows that for any path $P$ with $\ell$ vertices, with probability $\geq 1-\emm^{-\frac{1}{40}\epsilon \ell \log n}$, at most $\epsilon \ell/2$ of the events $\{Q_{t,k,i}\}$ occur, i.e., the $h$-graph-neighbourhood of $P$ contains at least $(1-\epsilon)\ell$ isolated tree components (every event $Q_{t,k,i}$ that occurs decreases the number of isolated tree components by at most two - on the other hand, if $Q_{t,i,k}$ does not occur then the number of isolated tree components stays the same). Since there are at most $ n \Delta^{\ell}$ paths with $\ell$ vertices, we obtain by a union bound that the probability that there exists a path whose $h$-graph-neighbourhood contains less than $(1-\epsilon)\ell$ isolated tree components is upper bounded by
\[\sum^{n^{9/10}}_{\ell=\left\lceil 200/\epsilon\right\rceil}\emm^{\log n +\ell \log \Delta-\frac{1}{40}\epsilon \ell \log n}=O(1/n^{2}),\]
where the last bound follows by observing that, for all sufficiently large $n$, the summands are decreasing functions of $\ell$ and that for $\ell=\left\lceil 200/\epsilon\right\rceil$ we have $\emm^{\log n +\ell \log \Delta-\frac{1}{40}\epsilon \ell \log n}=O(1/n^3)$.

This concludes the proof of Lemma~\ref{lem:graphneighbor}.
\end{proof}
To conclude this section, we clarify a small point relevant to Remark~\ref{rem:4rf4rr32221}. We will only utilise Lemma~\ref{lem:graphneighbor} for paths of logarithmic length (despite that the lemma is stated for convenience for much longer paths) and therefore the property can be checked in polynomial time. Analogously, the sum in Lemma~\ref{lem:numpaths} will only be considered for cycles of logarithmic length and therefore the (restricted) inequality can also be checked in polynomial time.

\section{Algorithm for the random-cluster model}\label{sec:rcalgo}
In this section, we prove Theorem~\ref{thm:main2}. In Section~\ref{sec:bb44hg4g}, we first describe the algorithm and analyse how to update a random-cluster configuration when we add a new edge. In Section~\ref{sec:t5g4g65h77}, we show how to control the aggregate error of our sampling algorithm on random regular graphs. Finally, in Section~\ref{sec:vewe2d232}, we combine these pieces to conclude Theorem~\ref{thm:main2}.
\subsection{The Algorithm}\label{sec:bb44hg4g}
To prove Theorem~\ref{thm:main2}, we will consider a simple percolation algorithm for sampling a random-cluster configuration on a random $\Delta$-regular graph $G$. The algorithm is given in Figure~\ref{alg:perc} and in Section~\ref{sec:vewe2d232} we will detail its performance when the input is a random regular graph.

\begin{figure}[h]
\begin{mdframed}
\textbf{Algorithm} $\SampleRC(G)$ \vskip 0.05cm
\rule[0.3cm]{6cm}{0.4pt}
\vskip 0.05cm
\noindent \textbf{parameters:} reals $p\in(0,1)$ and  $q\geq 1$ \vskip 0.2cm
\noindent \textbf{Input:} Graph $G=(V,E)$ 

\noindent \textbf{Output:} Either \textsc{Fail} or a set $S\subseteq E$

\vskip 0.4cm

\noindent $E':= \{e \in E\mid \mbox{$e$ belongs to a short cycle}\}$ \vskip 0.2cm

\noindent \textbf{if} \mbox{$G'=(V,E')$ contains a component which is neither a cycle nor an isolated vertex} \vskip 0.1cm \noindent \hspace{0.5cm}\textbf{then} \textsc{Fail}

\noindent \textbf{else} \vskip 0.1cm
\noindent \hspace{0.5cm} Sample an RC configuration $S'\subseteq E'$ on $G'$ (according to $\varphi_{G'}$);  \vskip 0.1cm
\noindent\hspace{0.5cm} Add to $S'$ each edge in $E\backslash E'$ independently with probability $p/(p+(1-p)q)$;\vskip 0.1cm
\noindent\hspace{0.5cm} Output the resulting set $S\subseteq E$.
\end{mdframed}
\caption{\label{alg:perc} Algorithm for sampling a random-cluster configuration. Note that, since $G'$ is a disjoint union of short cycles, the inital configuration $S'$ in the algorithm above can be obtained quickly in various ways (e.g., even brute force takes time $O^*(n^{6/5})$ since each cycle has length $\leq \tfrac{1}{5}\log_{\Delta-1} n$ and there are at most $n$ cycles).}
\end{figure}

Prior to that, let us first motivate the algorithm $\SampleRC$, by demonstrating how to update an RC configuration when we add a single edge $\{u,v\}$. To control the effect of adding an edge, it will be relevant to consider the event that there is an open path between $u$ and $v$ (for a path $P$ in $G$ and an RC configuration $S\subseteq E$, we say that $P$ is open in $S$ if all of its edges belong to $S$); we denote this event by $u\conn v$.
\begin{lemma}\label{lem:edgewise}
Let $p\in (0,1)$ and $q\geq 1$, and consider arbitrary $\epsilon\in (0,1/q)$. 

Let $G=(V,E)$ be a graph and $u,v$ be two vertices such that $\{u,v\}\not\in E$ and $\varphi_G(u\conn v)\leq \epsilon$. Consider the graph $G'=(V,E')$ obtained from $G$ by adding the edge $\{u,v\}$. Sample a random subset of edges $Y\subseteq E'$ as follows: first, sample a subset of edges $X\subseteq E$ according to the RC distribution $\varphi_G$ and, then, set $Y=X\cup \{e\}$ with probability $p/(p+(1-p)q)$, and $Y=X$ otherwise. 

Then, the distribution of $Y$, denoted by $\nu_Y$, is within total variation distance $2q\epsilon$ from the RC distribution $\varphi_{G'}$ on $G'$ with parameters $p,q$, i.e.,
\[\norm{\nu_Y-\varphi_{G'}}_{\TV}\leq 2q \epsilon.\]
\end{lemma}
\begin{proof}
Fix an arbitrary $\epsilon\in (0,1/q)$.  

Let $\Omega_{\open}$ be the set of subsets  $S\subseteq E$ such that in the graph $(V,S)$, $u$ and $v$ are connected by an open path and let $\Omega_{\closed}=2^{E}\backslash \Omega_{\open}$. For $S\subseteq E$, denote for convenience by $S_e$ the set $S\cup \{e\}$. Observe that 
\begin{align}
\forall S\in \Omega_{\closed}:&\qquad w_{G'}(S_e)=\frac{p}{q} w_G(S), \quad  w_{G'}(S)=(1-p) w_G(S),\label{eq:Sclosed}\\
\forall S\in \Omega_{\open}:&\qquad w_{G'}(S_e)=p w_G(S), \quad  w_{G'}(S)=(1-p) w_G(S).\label{eq:Sopen}
\end{align}
Note also that for any $S\subseteq E$, we have
\[\nu_Y(S_e)=\frac{p\,\varphi_G(S)}{p+q(1-p)}=\frac{\frac{p}{q} w_G(S)}{(\frac{p}{q}+1-p)Z_G}, \qquad \nu_Y(S)=\frac{q(1-p)\varphi_G(S)}{p+q(1-p)}=\frac{(1-p)w_G(S)}{(\frac{p}{q}+1-p)Z_G}.\]
Using \eqref{eq:Sclosed}, \eqref{eq:Sopen} and the assumption $\varphi_G(u\conn v)\leq \epsilon$, we will also show the following for the partition functions $Z_G,Z_G'$:
\begin{equation}\label{eq:ZratioGGp}
|M|\leq q \epsilon, \mbox{ where } M:=\frac{(\frac{p}{q}+1-p)Z_G}{Z_{G'}}-1.
\end{equation} 
Let us conclude the proof assuming, for now, \eqref{eq:ZratioGGp}. To do this, we decompose $\norm{\nu_Y-\varphi_{G'}}_{\TV}$ as 
\begin{equation}\label{eq:decompo353}
\norm{\nu_Y-\varphi_{G'}}_{\TV}=\frac{1}{2}\sum_{S\subseteq E}|\nu_Y(S)-\varphi_{G'}(S)|+\frac{1}{2}\sum_{S\subseteq E}|\nu_Y(S_e)-\varphi_{G'}(S_e)|.
\end{equation}

For $S\subseteq E$, we have
\[|\nu_Y(S)-\varphi_{G'}(S)|=\bigg|\frac{(1-p)w_G(S)}{(\frac{p}{q}+1-p)Z_G}-\frac{(1-p)w_G(S)}{Z_{G'}}\bigg|=\frac{(1-p)w_G(S)}{(\frac{p}{q}+1-p)Z_G}|M|\leq |M|\,\varphi_G(S)\]
and therefore we can bound the first sum in \eqref{eq:decompo353} as
\begin{equation}\label{eq:v433rftg4t}
\sum_{S\subseteq E}|\nu_Y(S)-\varphi_{G'}(S)|\leq |M|\leq q\epsilon .
\end{equation}
To bound the second sum in \eqref{eq:decompo353}, we consider whether $S\in \Omega_{\closed}$ or $S\in \Omega_{\open}$. For $S\in \Omega_{\closed}$, we have
\[|\nu_Y(S_e)-\varphi_{G'}(S_e)|=\bigg|\frac{\frac{p}{q}w_G(S)}{(\frac{p}{q}+1-p)Z_G}-\frac{\frac{p}{q}w_G(S)}{Z_{G'}}\bigg|=\frac{\frac{p}{q}w_G(S)}{(\frac{p}{q}+1-p)Z_G}|M|\leq |M|\,\varphi_G(S),\]
while for $S\in \Omega_{\open}$, we have 
\[|\nu_Y(S_e)-\varphi_{G'}(S_e)|=\bigg|\frac{\frac{p}{q}w_G(S)}{(\frac{p}{q}+1-p)Z_G}-\frac{p w_G(S)}{Z_{G'}}\bigg|=\frac{\frac{p}{q}w_G(S)}{(\frac{p}{q}+1-p)Z_G}|q M+(q-1)|\leq 2q\, \varphi_G(S),\]
where the last inequality follows from  $|q M+(q-1)|\leq q(|M|+1)\leq 2q$ (using that $|M|\leq q\epsilon\leq 1$ from \eqref{eq:ZratioGGp}). Hence
\begin{equation}\label{eq:v433rftg4tb}
\sum_{S\subseteq E}|\nu_Y(S_e)-\varphi_{G'}(S_e)|\leq  |M|\,\varphi_G(\Omega_{\closed})+2q\, \varphi_G(\Omega_{\open})\leq |M|+2q \epsilon\leq 3q\epsilon.
\end{equation}
Plugging \eqref{eq:v433rftg4t} and \eqref{eq:v433rftg4tb} in \eqref{eq:decompo353}, we obtain that  $\norm{\nu_Y-\varphi_{G'}}_{\TV}\leq 2q\epsilon$ as wanted.

 To complete the proof, it only remains to show \eqref{eq:ZratioGGp}. Let
\[Z_{G,\open}=\sum_{S\in \Omega_{\open}} w_G(S),\quad  Z_{G,\closed}=\sum_{S\in \Omega_{\closed}} w_G(S),\]
so that $Z_G=Z_{G,\open}+Z_{G,\closed}$. By assumption, we have that
\[\varphi_G(u \conn v)=\frac{Z_{G,\open}}{Z_G}\leq \epsilon,\mbox{ so that } \frac{Z_{G,\open}}{Z_{G,\closed}}\leq \frac{\epsilon}{1-\epsilon}.\]
Let also
\[Z_{G',\open}=\sum_{S\in \Omega_{\open}} w_{G'}(S)+w_{G'}(S_e),\quad  Z_{G',\closed}=\sum_{S\in \Omega_{\closed}} w_{G'}(S)+w_{G'}(S_e),\]
so that  $Z_{G'}=Z_{G',\open}+Z_{G',\closed}$. Using \eqref{eq:Sclosed} and \eqref{eq:Sopen}, we obtain that
\[Z_{G',\open}=Z_{G,\open}, \qquad Z_{G',\closed}=\Big(\frac{p}{q}+1-p\Big)Z_{G,\closed} \] 
and therefore
\[\frac{Z_G}{Z_{G'}}=\frac{Z_{G,\open}+Z_{G,\closed}}{Z_{G',\open}+Z_{G',\closed}}=\frac{\frac{Z_{G,\open}}{Z_{G,\closed}}+1}{\frac{Z_{G,\open}}{Z_{G,\closed}}+\big(\frac{p}{q}+1-p\big)}.\]
Since $\tfrac{Z_{G,\open}}{Z_{G,\closed}}\in(0, \tfrac{\epsilon}{1-\epsilon}]$ and the function $f(x):=\frac{x+1}{x+(\frac{p}{q}+1-p)}$ is decreasing in $x$, we obtain that $f(\tfrac{\epsilon}{1-\epsilon})\leq\tfrac{Z_G}{Z_{G'}}\leq f(0)$ and hence the bounds
\[\frac{\frac{p}{q}+1-p}{\epsilon+(1-\epsilon)(\frac{p}{q}+1-p)}\leq \frac{(\frac{p}{q}+1-p)Z_G}{Z_{G'}}\leq 1.\]
The l.h.s. is $\geq 1-q \epsilon$ for all $p\in (0,1)$, $q\geq 1$ and $\epsilon\in (0,1)$, and hence \eqref{eq:ZratioGGp} follows. 

This concludes the proof of Lemma~\ref{lem:edgewise}.
\end{proof}
\subsection{Aggregating the error}\label{sec:t5g4g65h77}
To utilise Lemma~\ref{lem:edgewise}, we need to upper bound the probability that two vertices belong to the same component in a RC configuration. In turn, it suffices to bound the probability that there is an open path between the vertices. To this end, we utilise the fact that the parameters $p,q$ are in the uniqueness region of the $(\Delta-1)$-ary tree and the tree-like structure around paths (cf. Definition~\ref{def:hpathneigh}) to show the following. The proof of the lemma  is given in Section~\ref{sec:qaz234}. 

\newcommand{\statelemtreespath}{
Let $\Delta\geq 3$ be an integer, $q\geq 1$ and $p<p_c(q,\Delta)$. There exist constants $K<1/(\Delta-1)$ and $\epsilon>0$ such that the following holds for all sufficiently large integers $\ell$ and $h$.    

Let $G$ be a $\Delta$-regular graph and $P$ be a path with $\ell$ vertices whose $h$-graph-neighbourhood contains $(1-\epsilon)\ell$ isolated tree components.  Let $\varphi_G$ be the RC distribution on $G$ with parameters $p,q$. Then,   
\[\varphi_G(\mbox{path $P$ is open})\leq K^{\ell}.\]
}
\begin{lemma}\label{lem:treespath}
\statelemtreespath
\end{lemma}

Using monotonicity properties of the RC distribution, we can extend Lemma~\ref{lem:treespath} to arbitrary subgraphs of a target graph $G$. In particular, suppose that $G,P$ are as in Lemma~\ref{lem:treespath} and that $G'$ is a subgraph of $G$ which contains the path $P$. Then it also holds that $\varphi_{G'}(\mbox{$P$ is open})\leq K^{\ell}$.  We will not define the notion of monotonic distributions in its full generality, but instead we will just state the following property of RC distributions  which will be sufficient for our purposes, see \cite[Sections 2.1 \& 2.2]{Grimmettbook} for a detailed exposition. 
\begin{lemma}[see, e.g., {\cite[Chapter 2]{Grimmettbook}}]\label{lem:monotonicity}
Let $G=(V,E)$ be a graph and consider the RC distribution on $G$ with parameters $p\in (0,1)$ and $q\geq 1$. Then for any subsets $S,S'\subseteq E$ such that $S\subseteq S'$, it holds that
\[\varphi_G(\mathcal{F}\mid\mbox{$S$ open})\leq \varphi_G(\mathcal{F}\mid\mbox{$S'$ open})\quad \mbox{ and }\quad \varphi_G(\mathcal{F}\mid\mbox{$S$ closed})\geq \varphi_G(\mathcal{F}\mid\mbox{$S'$ closed})\]
for any increasing event $\mathcal{F}$.\footnote{An event $\mathcal{F}\subseteq 2^E$ is increasing if, for all $T\subseteq T'\subseteq E$, $T\in \mathcal{F}$ implies that $T'\in \mathcal{F}$ as well. Note also that $S'$ open means that every edge in $S'$ is open, and $S'$ closed means that every edge in $S'$ is closed.}
\end{lemma}

Combining Lemmas~\ref{lem:numpaths},~\ref{lem:graphneighbor} and~\ref{lem:treespath}, we can now conclude the following.
\begin{lemma}\label{lem:error}
Let $\Delta\geq 3$ be an integer, $q\geq 1$ and $p<p_c(q,\Delta)$. Then, there exists a constant $\delta>0$ such that, as $n\rightarrow \infty$, the following holds with probability $1-o(1)$ over the choice of a uniformly random $\Delta$-regular graph $G=(V,E)$ with $n$ vertices.

 Let $e_1,\hdots, e_t$ be the edges of $G$ that do not belong to short cycles. For $j\in [t]$, let $e_j=\{u_j,v_j\}$ and $G_j$ be the subgraph $G\backslash\{e_1,\hdots,e_j\}$. Then, it holds that
\begin{equation}\label{eq:additererror}
\sum^t_{j=1}\varphi_{G_j}(u_j\conn v_j)\leq 1/n^{\delta}.
\end{equation}
\end{lemma}
\begin{proof}
Let $K<1/(\Delta-1)$ and $\epsilon>0$ be the constants in Lemma~\ref{lem:treespath}, and let $\ell',h'>0$ be constants so that Lemma~\ref{lem:treespath} applies for all $\ell\geq \ell'$ and $h\geq h'$. Fix $h$ to  be any integer greater than $h'$. Let $\delta>0$ be the constant in Lemma~\ref{lem:numpaths} corresponding to $\ell_0:=1/(5\log (\Delta-1))$ and $W=1/K$ (note that $W>\Delta-1$). Let $\ell_1>0$ be the constant in Lemma~\ref{lem:graphneighbor} corresponding to our choice of $\epsilon$. Finally, let $\ell_2:=\max\{4\log(W/(\Delta-1)),2\ell_0\}$.

Taking a union bound over Lemmas~\ref{lem:numpaths} and~\ref{lem:graphneighbor}, we have that a uniformly random $\Delta$-regular graph $G=(V,E)$ with $n$ vertices satisfies the following with probability $1-o(1)$ over the choice of the graph:
\begin{enumerate}
\item \label{it:numpaths} $\displaystyle\sum_{\ell=L_0}^{L_2}\frac{\ell C_\ell}{W^{\ell}}\leq 1/(2n^\delta)$, where $L_0:=\lceil \ell_0 \log n\rceil=\lceil\tfrac{1}{5}\log_{\Delta-1}n\rceil$, $L_2:=\lfloor \ell_2 \log n\rfloor$ and $C_\ell$ is the number of cycles of length $\ell$ in $G$.
\item \label{it:lemgraphneighbor} every path  $P$ in $G$ with $\ell$ vertices where $\ell_1\leq \ell\leq L_2+1$  has an $h$-graph-neighbourhood with at least $(1-\epsilon)\ell$ isolated tree components.
\end{enumerate}
We will show that for any $\Delta$-regular graph $G$ which satisfies Items~\ref{it:numpaths} and~\ref{it:lemgraphneighbor}, it holds that
\begin{equation*}\tag{\ref{eq:additererror}}
\sum^t_{j=1}\varphi_{G_j}(u_j\conn v_j)\leq 1/n^{\delta},
\end{equation*}
where $e_1=\{u_1,v_1\},\hdots, e_t=\{u_t,v_t\}$ are  the edges of $G$ that do not belong to short cycles (i.e., cycles of length $\leq \ell_0\log n$),  $G_j$ is the subgraph $G\backslash\{e_1,\hdots,e_j\}$ and $\varphi_{G_j}$ is the RC distribution on $G_j$ with parameters $p,q$. Decreasing the value of $\delta$ does not affect the validity of Item~\ref{it:numpaths}, and hence we will assume that $\delta\in (0,1)$. 

For $j\in [t]$, consider the edge $e_j=\{u_j,v_j\}$ and let $P_{\ell,j}$ denote the number of paths with $\ell$ vertices in $G$ whose endpoints are $u_j$ and $v_j$. Using the fact that $G$ satisfies Item~\ref{it:lemgraphneighbor}, we will show shortly that, for all $j\in [t]$, it holds that
\begin{equation}\label{eq:amortise}
\varphi_{G_j}(u_j\conn v_j)\leq \frac{1}{n^3}+\sum^{L_2}_{\ell=L_0}\frac{P_{\ell,j}}{W^{\ell}}.
\end{equation}
Let us assume \eqref{eq:amortise} for now, and conclude the proof of \eqref{eq:additererror}. Summing \eqref{eq:amortise} over $j\in [t]$ (and using the trivial bound $t\leq |E|\leq \Delta n/2$), we obtain that
\begin{equation}\label{eq:v4tyyhu43737}
\sum^{t}_{j=1}\varphi_{G_j}(u_j\conn v_j)\leq \frac{\Delta}{n^2}+\sum^{L_2}_{\ell=L_0}\frac{\sum^t_{j=1}P_{\ell,j}}{W^{\ell}}\leq\frac{\Delta}{n^2}+\sum^{L_2}_{\ell=L_0}\frac{\ell C_{\ell}}{W^{\ell}}, 
\end{equation}
where in the last inequality we used that $\sum^t_{j=1}P_{\ell,j}\leq \ell C_\ell$ which follows from the observation that every path with $\ell$ vertices connecting the endpoints of an edge $\{u_j,v_j\}$ maps to a cycle of length $\ell$ (by adding the edge $\{u_j,v_j\}$) and each cycle of length $\ell$ can potentially arise at most $\ell$ times under this mapping. Using \eqref{eq:v4tyyhu43737} and the fact that $G$ satisfies Item~\ref{it:numpaths}, we obtain \eqref{eq:additererror}, as wanted.

 To finish the proof, it only remains to prove \eqref{eq:amortise}. Since $G$ satisfies Item~\ref{it:lemgraphneighbor} and $W=1/K$, by Lemma~\ref{lem:treespath}, we have that for any path $P$ of length $\ell\in [L_0,L_2+1]$ connecting $u_j,v_j$, it holds that
\begin{equation}\label{eq:GpathP}
\varphi_G(\mbox{$P$ is open})\leq 1/W^{\ell} \mbox{ for any path $P$ in $G$ with $\ell$ vertices, $\ell\in [L_0,L_2+1]$}.
\end{equation}
Since $G_j$ is a subgraph of $G$,  any path in $G_j$ that connects $u_j$ and $v_j$ is also present in $G$. Moreover, we have that $\varphi_{G_j}$ is obtained by conditioning some edges of $G$ to be closed (namely, $e_1, \hdots,e_j$). Therefore, by Lemma~\ref{lem:monotonicity}, we conclude from \eqref{eq:GpathP} that 
\begin{equation}\label{eq:GjpathP}
\varphi_{G_j}(\mbox{path $P$ is open})\leq 1/W^{\ell}, \mbox{ for any path $P$ in $G_j$ with $\ell$ vertices, $\ell\in [L_0,L_2+1]$}.
\end{equation}
Since the edge $\{u_j,v_j\}$ does not belong to a short cycle, we have that any path $P$ in $G_j$ connecting $u_j$ and $v_j$ has length at least $\ell_0\log n$. We can therefore bound the probability of an open path between $u_j$ and $v_j$ by 
\begin{equation}\label{eq:wwferfrgyy3}
\varphi_{G_j}(u_j\conn v_j)\leq \varphi_{G_j}(\mathcal{E}_j)+\varphi_{G_j}(\mathcal{F}_j),
\end{equation}
where $\mathcal{E}_j$ is the event that there exists an open path $P$ with  $\ell$ vertices with $L_0\leq\ell\leq L_2$  connecting $u_j$ and $v_j$, whereas $\mathcal{F}_j$ is the event that there exists an open path $P$ with $\ell=L_2+1$ vertices starting from  $u_j$ (the other endpoint can be $v_j$ or any other vertex of the graph). Using \eqref{eq:GjpathP}, we have by a union bound over paths that
\begin{equation}\label{eq:wdrw4f45}
\varphi_{G_j}(\mathcal{E}_j)\leq \sum^{L_2}_{\ell=L_0}\frac{P_{\ell,j}}{W^{\ell}}, \quad \varphi_{G_j}(\mathcal{F}_j)\leq \frac{\Delta(\Delta-1)^{\ell_2\log n}}{W^{\ell_2\log n}}\leq 1/n^{3},
\end{equation}
where in the bound for $\varphi_{G_j}(\mathcal{E}_j)$ we used that there are $P_{\ell,j}$ paths with  $\ell$ vertices connecting $u_j$ and $v_j$, while in the bound for  $\varphi_{G_j}(\mathcal{F}_j)$  we used that there are at most $\Delta(\Delta-1)^{L_2-1}$ paths in $G$ with $L_2+1$  vertices starting with $u_j$ (since $G$ has max degree $\Delta$), the trivial inequalities $\ell_2\log n-1\leq L_2\leq \ell_2\log n$,  and the choice of $\ell_2$ which guarantees that $\ell_2\geq 4\log(W/(\Delta-1))$. 

Combining \eqref{eq:wwferfrgyy3} and \eqref{eq:wdrw4f45} yields \eqref{eq:amortise} (since $\delta\in(0,1)$), thus completing the proof of Lemma~\ref{lem:error}.
\end{proof}

\subsection{Combining the pieces --- Proof of Theorem~\ref{thm:main2}}\label{sec:vewe2d232}
We are now able to prove the following theorem, which details the performance of the algorithm $\SampleRC$ on random $\Delta$-regular graphs and yields as an immediate corollary Theorem~\ref{thm:main2}.
\begin{theorem}\label{thm:main}
Let $\Delta\geq 3$, $q\geq 1$ and $p<p_c(q,\Delta)$. Then, there exists a constant $\delta>0$ such that, as $n\rightarrow \infty$, the following holds   with probability $1-o(1)$ over the choice of a random $\Delta$-regular graph $G=(V,E)$ with $n$ vertices.

The output of Algorithm $\SampleRC(G)\,$ (cf. Figure~\ref{alg:perc}) is a set $S\subseteq E$ whose distribution $\nu_S$ is within total variation distance $O(1/n^{\delta})$ from the RC distribution $\varphi_G$ with parameters $p,q$, i.e.,
\[\norm{\nu_S-\varphi_{G}}_{\TV}=O(1/n^{\delta}).\]
\end{theorem} 
\begin{proof}
By Lemmas~\ref{lem:disjointcycles} and~\ref{lem:error}, we have by a union bound that a uniformly random $\Delta$-regular graph $G$ with $n$ vertices satisfies the following with probability $1-o(1)$:
\begin{enumerate}
\item \label{it:disjoint} any two distinct cycles of length $\leq \frac{1}{5} \log_{\Delta-1} n$ are disjoint,
\item \label{it:error} $\sum^t_{j=1}\varphi_{G_j}(u_j\conn v_j)\leq 1/n^{\delta}$, where $e_1=\{u_1,v_1\},\hdots, e_t=\{u_t,v_t\}$ are the edges of $G$ that do \emph{not} belong to short cycles, and $G_j$ is the subgraph $G\backslash\{e_1,\hdots,e_j\}$ (for $j\in [t]$).
\end{enumerate}
We will show that for any graph $G=(V,E)$ that satisfies Items~\ref{it:disjoint} and~\ref{it:error}, the output of the algorithm \mbox{\SampleRC} is a random set $S\subseteq E$ whose distribution $\nu_S$ is within total variation distance $1/n^{\delta}$ from the RC distribution $\varphi_G$ with parameters $p,q$, therefore proving the result.

Let $G_t=(V,E_t),\hdots, G_1=(V,E_1)$ be the sequence of subgraphs as in Item~\ref{it:error} above and, for convenience, set $E_0=E$ and let  $G_0=(V,E_0)$ (so that $G_0=G$). Note that $G_t=G'$ (where $G'$ is the graph considered in the algorithm \mbox{\SampleRC}), i.e., $G_t$ is the subgraph of $G$ where only the edges that belong to short cycles appear. By Item~\ref{it:disjoint}, we have that $G_t$ consists of isolated vertices and disjoint cycles and hence we can conclude that the output of the algorithm \mbox{\SampleRC} is not \textsc{Fail}, i.e., on input a graph $G$ satifying Items~\ref{it:disjoint} and~\ref{it:error}, \mbox{\SampleRC} outputs a random set $S\subseteq E$. It therefore remains to show that the distribution $\nu_S$ of $S$ satisfies
\begin{equation}\label{eq:rv4ttbtbt45}
\norm{\nu_S-\varphi_G}_{\TV}=O( 1/n^{\delta}).
\end{equation}
For $j=t,t-1,\hdots, 0$, let $S_j=S\cap E_j$ and let $\nu_{S_j}$ denote the distribution of $S_j$. Note that $S_0=S$ and $S_{t}=S'$ (where $S'$ is the subset of edges considered in the algorithm \mbox{\SampleRC}). We have that
\begin{equation}\label{eq:startingpoint}
\norm{\nu_S-\varphi_G}_{\TV}=\norm{\nu_{S_0}-\varphi_{G_0}}_{\TV}, \quad \norm{\nu_{S_t}-\varphi_{G_t}}_{\TV}=\norm{\nu_{S'}-\varphi_{G'}}_{\TV}=0.
\end{equation}  
 
For $j\in [t]$, we have that $S_{j-1}$ is obtained from $S_j$ by adding the edge $e_j$ with probability $p/(q+(1-p)q)$. Let $\hat{S}_{j-1}\subseteq E_{j-1}$ be a subset of edges obtained by sampling an RC configuration from $G_j$ (according to $\varphi_{G_j}$) and  adding the edge $e_j$ with probability $p/(q+(1-p)q)$; denote by $\nu_{\hat{S}_{j-1}}$ the distribution of $\hat{S}_{j-1}$.  By Lemma~\ref{lem:edgewise} we have that\footnote{Note that, to apply Lemma~\ref{lem:edgewise}, we need to ensure that $\varphi_{G_j}(u_j\conn v_j)<1/q$, which is guaranteed by Item~\ref{it:error}.}
\[\big\lVert\nu_{\hat{S}_{j-1}}-\varphi_{G_{j-1}}\big\rVert_{\TV}\leq 2q\,\varphi_{G_j}(u_j\conn v_j).\]
Moreover, since in each of $S_{j-1}$ and $\hat{S}_{j-1}$ the edge $e_j$ appears independently with the same probability $p/(q+(1-p)q)$, we have that  
\[\big\lVert\nu_{S_{j-1}}-\nu_{\hat{S}_{j-1}}\big\rVert_{\TV}=\big\lVert\nu_{S_{j}}-\varphi_{G_{j}}\big\rVert_{\TV}.\] 
Using the triangle inequality and induction, we obtain that for all $j=0,1,\hdots,t$ it holds that
\[\norm{\nu_{S_j}-\varphi_{G_j}}_{\TV}\leq \norm{\nu_{S_t}-\varphi_{G_t}}_{\TV}+2q\sum^{t-1}_{j'=j} \varphi_{G_{j'}}(u_{j'}\conn v_{j'})\]
Writing this out for $j=0$ and using \eqref{eq:startingpoint}, we obtain that
\[\norm{\nu_S-\varphi_G}_{\TV}\leq 2q\sum^{t-1}_{j=0} \varphi_{G_j}(u_j\conn v_j)\leq 2q/n^{\delta},\]
where in the last inequality we used that the graph $G$ satisfies Item~\ref{it:error}. This finishes the proof of \eqref{eq:rv4ttbtbt45} and therefore the proof of Theorem~\ref{thm:main} as well.
\end{proof}

\section{Algorithm for the antiferromagnetic Potts model}\label{sec:tb34bb6}

In this section, we give the details of our sampling algorithm for the antiferromagnetic Potts model (outlined in Section~\ref{sec:proofapproach}). The section is organised as follows. First, in Section~\ref{sec:vvverv3242}, we formalise the connection between the Potts model on bichromatic classes and the Ising model.  Then, in Section~\ref{sec:erv3v3444}, we state the sampling algorithm for the Ising model on graphs with small average growth that we are going to use for resampling bichromatic classes in the Potts model; moreover, we state certain correlation decay properties for the Ising model that will be relevant for analysing the error of our Potts  sampler. In Section~\ref{sec:vr55g45g45g45}, we state the key lemma that allows us to bound the average growth of bichromatic classes in the Potts model on random regular  graphs. In Section~\ref{sec:t5g45g45g5}, we show an ``idealised'' subroutine that  updates a Potts configuration when we add a new edge $\{u,v\}$; the subroutine works by resampling an appropriately chosen bichromatic class and it is ``idealised'' in the sense that it assumes that certain steps can be carried out efficiently. In Section~\ref{sec:resamplereal}, we modify the subroutine to make it computationally efficient by considering the average growth of bichromatic classes that get resampled; there, we give the complete description of the actual resampling subroutine used in our Potts sampler.  With these pieces in place, we are in position to complete the description and analysis of the Potts sampler in Section~\ref{sec:thyh55h6}.

\subsection{Connection between Potts on bichromatic classes and the Ising model}\label{sec:vvverv3242}

In this section, we describe  the connection between the Potts model on bichromatic classes and the Ising model. Recall that the Ising model is the special case $q=2$ of the Potts model; to distinguish between the models, we will use $\pi_G$ to denote the Ising distribution on $G$ with parameter $B$. Sometimes, we will need to replace the binary set of states $\{1,2\}$ in the Ising model by other binary sets to facilitate the arguments; we use $\pi^{c_1,c_2}_G$ to denote the Ising distribution with binary set of states $\{c_1,c_2\}$ (we will have that $c_1,c_2\in [q]$).

For a configuration $\sigma:V\rightarrow [q]$, we will denote by $\sigma_U$ the restriction of $\sigma$ to the set $U$. Our sampling algorithm is based on the following simple observation.

\begin{observation}\label{obs:wrvtvte}
Let $q\geq 3$ and $B>0$. Let $G=(V,E)$ be a graph, $U$ be a subset of $V$and $c,c'$ be distinct colours in $[q]$. Then, for any configuration $\eta: U\rightarrow \{c,c'\}$, it holds that
\[\mu_G\big(\sigma_U=\eta\mid \sigma^{-1}(c,c')=U\big)=\pi^{c,c'}_{G[U]}(\eta),\]
i.e., conditioned on $U$ being the $(c,c')$-colour-class in the Potts distribution $\mu_G$, the marginal distribution on $U$ is the Ising distribution $\pi_{G[U]}$ (with set of states $\{c,c'\}$). 
\end{observation}

The following definition will be notationally convenient.
\begin{definition}\label{def:piuv}
Let $G$ be a graph, $u,v$ be vertices in $G$ and $c,c'$ be distinct colours in $[q]$. We write $\pi^{c,c'}_{G,u,v}$ to denote the Ising distribution on $G$ (with set of states $\{c,c'\}$)  conditioned on $u$ taking the state $c$ and $v$ the state $c'$.
\end{definition}

\subsection{Sampling the Ising model on graphs with small average growth}\label{sec:erv3v3444}
Recall from Section~\ref{sec:proofapproach} that our algorithm for sampling the antiferromagnetic Potts  model with parameter $B$ will use as a subroutine a sampling algorithm for the Ising model with parameter $B$ to recolour bichromatic classes.  In general, these classes may consist of large bichromatic components (with a linear number of vertices), so to carry out this subroutine efficiently, we need to use an approximate sampling algorithm; our leverage point will be that we can bound the average growth of  bichromatic components, cf. Definition~\ref{def:branchingfactor}. Adapting results of \cite{MosselSly}, we show the following in Section~\ref{sec:sampleising}.

\newcommand{\statethmsampleising}{Let $B\in (0,1)$ and $b>0$ be constants such that $b\frac{1-B}{1+B}<1$,  and let $\Delta\geq 3$ be an integer. Then, there exists  $M_0>0$ such that the following holds for all $M>M_0$.

There is a polynomial-time algorithm that, on input an $n$-vertex graph $G$ with maximum degree at most $\Delta$ and  average growth $b$ up to depth $L=\left\lceil M\log n\right\rceil$, outputs  a configuration $\tau:V\rightarrow\{1,2\}$ whose distribution $\nu_\tau$ is within total variation distance $1/n^\cons$ from the Ising distribution on $G$ with parameter $B$, i.e., 
\[\big\lVert\nu_\tau-\pi_G\big\rVert_{\TV}\leq 1/n^{\cons}.\]
Moreover, the algorithm, when given as additional input two vertices $u$ and $v$ in $G$, outputs a configuration $\tau:V\rightarrow\{1,2\}$  such that $\tau_u=1$ and $\tau_v=2$, and whose distribution $\nu_\tau$ satisfies
\[\big\lVert\nu_\tau-\pi^{1,2}_{G,u,v}(\cdot)\big\rVert_{\TV}\leq 1/n^{\cons},\]
where $\pi^{1,2}_{G,u,v}$ is the Ising distribution on $G$  conditioned on $u$ having state $1$ and $v$ having state $2$.}
\begin{theorem}\label{thm:sampleising}
\statethmsampleising
\end{theorem}

In addition, we will use the following spatial mixing result to analyse the accuracy of our algorithm for the antiferromagnetic Potts model. The proof is given in Section~\ref{sec:corrdecayIsing}.
\newcommand{\statelemcorrdecayIsing}{Let $B\in (0,1)$ and $b>0$ be constants such that $b\frac{1-B}{1+B}<1$. Then, there exists $M_0'>0$ such that the following holds for all $M>M_0'$. 

Let $G$ be an $n$-vertex graph with average growth $b$ up to depth $L=\left\lceil M\log n\right\rceil$,  and let $u,v$ be distinct vertices in $G$. Then
\[\Big|\pi_G(\sigma_u=1\mid \sigma_v=1) -\pi_G(\sigma_u=1\mid \sigma_v=2)\Big|\leq \frac{1}{n^{\cons}}+\sum^{L}_{\ell=1}P_\ell(G,u,v)\Big(\frac{1-B}{1+B}\Big)^{\ell}\]
where $P_\ell(G,u,v)$ is the number of paths with $\ell$ vertices in $G$ that connect $u$ and $v$.}
\begin{lemma}\label{lem:corrdecayIsing}
\statelemcorrdecayIsing
\end{lemma}

 We will also need the following crude bound.
\begin{lemma}\label{lem:f4f35f63g}
Let $B\in (0,1)$ and $\Delta\geq 3$ be an integer. Suppose that $G$ is a graph of maximum degree at most $\Delta$ and let $u$ be a vertex and $\Lambda$ be a set of vertices in $G$ such that $u\not\in \Lambda$. Then, for every configuration $\tau:\Lambda\to\{1,2\}$ and $s\in\{1,2\}$
\[\frac{B^{\Delta}}{1+B^\Delta}\leq \pi_G(\sigma_u=s\mid \sigma_\Lambda=\tau).\]
\end{lemma}
\begin{proof}
Without loss of generality, we assume that $s=1$.
Let $D$ be the number of neighbours of $u$ and $u_1,\hdots, u_d$ be the neighbours of $u$ in $G\setminus \Lambda$ and note that $d\leq D\leq \Delta$. Let $d_0$ be the number of $v\in\Lambda$ such that $v$ is $u$'s neighbour and $\tau_v=1$.  Let $s_1,\hdots,s_d\in \{1,2\}$ be arbitrary and let $d_1$ be the number of the $s_i$'s that are equal to $1$. Then we have that 
\[\pi_G(\sigma_u=1\mid \sigma_{u_1}=s_1,\hdots, \sigma_{u_d}=s_d, \sigma_\Lambda=\tau)=\frac{B^{d_2}}{B^{d_2}+B^{D-d_2}},\]
where $d_2 = d_0+d_1 \leq d_0+d \leq D$.
Since  $u_1,\hdots,u_d\not\in \Lambda$ we have that $\pi_G(\sigma_{u_1}=s_1,\hdots, \sigma_{u_d}=s_d\mid \sigma_\Lambda=\tau)\neq 0$ for any choice of $s_1,\hdots s_d$ and therefore by the law of total probability we obtain
\[\min_{d_2\in \{d_0,\hdots,d_0+d\}}\Big\{\frac{B^{d_2}}{B^{d_2}+B^{D-d_2}}\Big\}\leq \pi_G(\sigma_u=1\mid \sigma_\Lambda=\tau).\]
The function $\frac{B^x}{B^x+B^{D-x}}$ is decreasing for $x\in [0,D]$ (since $B\in (0,1)$) and hence
\[\min_{d_2\in \{d_0,\hdots,d_0+d\}}\Big\{\frac{B^{d_2}}{B^{d_2}+B^{D-d_2}}\Big\}=\frac{B^{D}}{1+B^D}.\]
Using that $B\in (0,1)$ and $D\leq \Delta$, we obtain the inequalities in the statement of the lemma.
\end{proof}

\subsection{Average growth of bichromatic components in the Potts distribution}\label{sec:vr55g45g45g45}
To utilise Theorem~\ref{thm:sampleising} and Lemma~\ref{lem:corrdecayIsing} for our sampling algorithms, we will need to bound the average growth of bichromatic components in a typical Potts configuration on a random regular graph. Our key lemma to achieve this will bound the probability that a path  is bichromatic\footnote{Let $G=(V,E)$ be a graph and $\sigma:V\rightarrow [q]$. We call a path $P$ bichromatic under $\sigma$ if there exist colours $c_1,c_2\in[q]$ such that every vertex $u$ of $P$ satisfies $\sigma_u\in\{c_1,c_2\}$.} in uniqueness, provided that the local neighbourhood around the path (in the sense of Definition~\ref{def:hpathneigh}) has a tree-like structure. The following lemma quantifies this probability bound and  is proved  in Section~\ref{sec:treespathPotts}. The proof uses the fact that the parameter $B$ lies in the uniqueness regime of the $(\Delta-1)$-ary tree.
\newcommand{\statelemtreespathPotts}{
Let $\Delta,q\geq 3$ be integers, and $B\in (0,1)$ be in the uniqueness regime of the $(\Delta-1)$-ary tree with $B\neq (\Delta-q)/\Delta$. Then, for any $\epsilon'>0$, there exists a positive constant $K<\frac{1+B}{B+q-1}+\epsilon'$ and $\epsilon>0$ such that the following holds for all sufficiently large integers $\ell$ and $h$.    

Let $G$ be a graph of maximum degree $\Delta$ and $P$ be a path with $\ell$ vertices whose $h$-graph-neighbourhood contains $(1-\epsilon)\ell$ isolated tree components.  Let $\mu_G$ be the Potts measure on $G$ with parameter $B$. Then,   
\[\mu_G(\mbox{path $P$ is bichromatic})\leq K^\ell.\]
}

\begin{lemma}\label{lem:treespathPotts}
\statelemtreespathPotts
\end{lemma}

For a random $\Delta$-regular graph $G$,  paths do have the tree-like structure of Lemma~\ref{lem:treespathPotts} (cf. Lemma~\ref{lem:graphneighbor}), and hence we can aggregate over all paths emanating from an arbitrary vertex (roughly $(\Delta-1)^{\ell}$ of them) and get a bound of roughly $(\Delta-1)K<\tfrac{(\Delta-1)(1+B)}{B+q-1}$ for the average growth of bichromatic components in a typical configuration $\sigma$.  This will allow us to use the upcoming $\ReSample$ subroutine for updating a bichromatic class using the Ising sampler of Section~\ref{sec:erv3v3444}.

\subsection{Analysing an Ideal ReSample subroutine}\label{sec:t5g45g45g5}

In this section, we give a preliminary description and analysis  of  the $\ReSample$ subroutine that updates a Potts configuration when we add a new edge $\{u,v\}$. The $\ReSample$ subroutine is inspired by the approach of Efthymiou \cite{Efthymiou} for colourings.

  In fact, for the moment, we will only study an ``idealised'' version of $\ReSample$ which we call $\IdealReSample$; the subroutine is idealised in the sense that it assumes that certain steps can be carried out efficiently.  Later, we will modify the subroutine to obtain the actual $\ReSample$ subroutine whose running time will be polynomial with respect to the size of the input graph. 

The point of analysing first $\IdealReSample$ is to give the key ideas behind the  underlying resampling step without bothering for the moment to make the subroutine computationally efficient. Moreover, the detour is going to be smaller than it might appear since the analysis of the  actual $\ReSample$ subroutine will follow from the analysis of $\IdealReSample$.

The $\IdealReSample$ subroutine takes as inputs a graph $G$, two vertices $u$ and $v$ of $G$ and a configuration $\sigma$ on $G$ such that $\sigma_u=\sigma_v$; it outputs a configuration $\sigma'$ on $G$ by updating the configuration on an appropriately chosen bichromatic class containing the vertices $u$ and $v$. The details of the subroutine can be found in Figure~\ref{alg:idealresample}. The subroutine will be used to update a Potts configuration when we add a new edge $\{u,v\}$, see the upcoming Lemma~\ref{lem:samplinganti}. 
 
\begin{figure}[h]
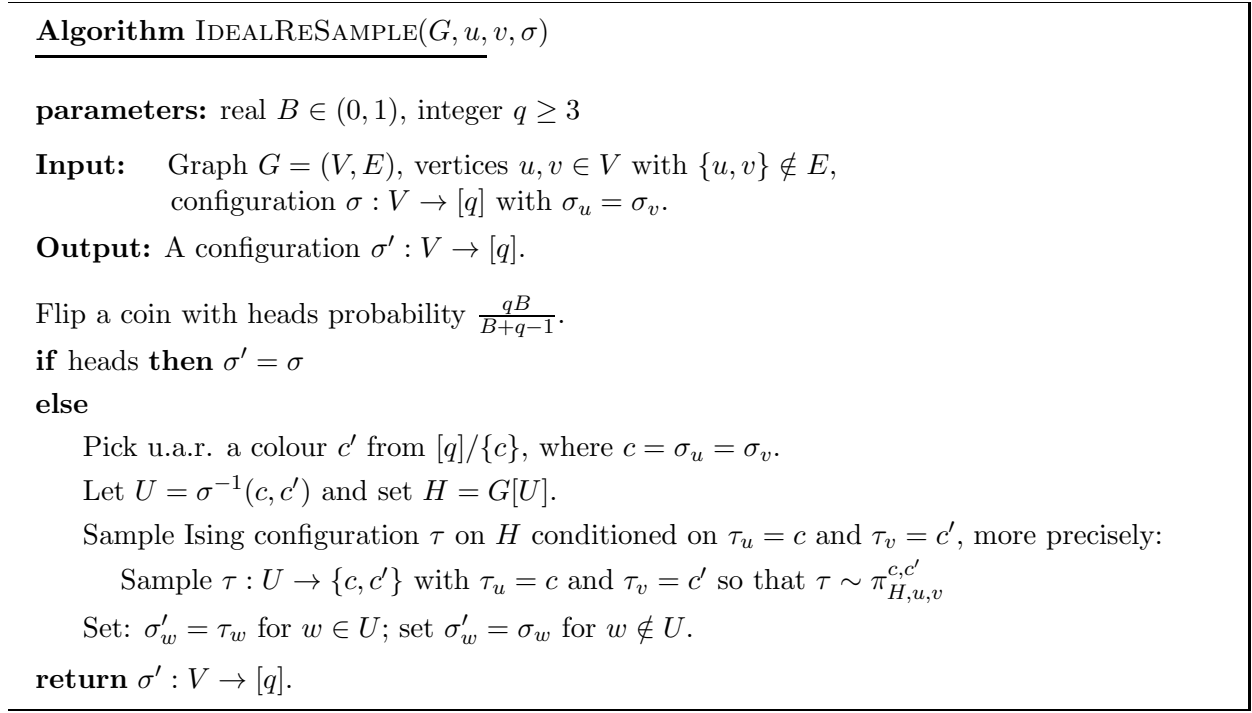


\begin{mdframed}
\textbf{Algorithm} $\IdealReSample(G,u,v,\sigma)$ \vskip 0.05cm
\rule[0.3cm]{6cm}{0.4pt}
\vskip 0.05cm

\noindent \textbf{parameters:} real $B\in(0,1)$, integer $q\geq 3$ \vskip 0.25cm

\noindent \textbf{Input:} \phantom{ \ \ }Graph $G=(V,E)$, vertices $u,v\in V$ with $\{u,v\}\notin E$,\\
\phantom{ \ \ \ \ \ \ \ \ \ \ \ \ } configuration $\sigma:V\rightarrow[q]$ with $\sigma_u=\sigma_v$.\vskip 0.15cm

\noindent \textbf{Output:} A configuration $\sigma':V\rightarrow[q]$.
\vskip 0.4cm

\noindent  Flip a coin with heads probability $\frac{qB}{B+q-1}$.\vskip 0.1cm
\noindent \textbf{if} heads \textbf{then} $\sigma'=\sigma$\vskip 0.1cm
\noindent \textbf{else} \vskip 0.1cm
\noindent \hspace{0.5cm} Pick u.a.r. a colour $c'$ from $[q]/\{c\}$, where $c=\sigma_u=\sigma_v$.  \vskip 0.1cm
\noindent \hspace{0.5cm} Let $U=\sigma^{-1}(c,c')$ and set $H=G[U]$. \vskip 0.1cm
\noindent \hspace{0.5cm} Sample Ising configuration $\tau$ on $H$ conditioned on $\tau_u=c$ and $\tau_v=c'$, more precisely:\vskip 0.1cm 
\noindent \hspace{1cm} Sample $\tau:U\rightarrow \{c,c'\}$ with $\tau_u=c$ and $\tau_v=c'$ so that $\tau\sim\pi^{c,c'}_{H,u,v}$ \vskip 0.15cm
\noindent \hspace{0.5cm}  Set: $\sigma_w'=\tau_w$ for $w\in U$;   set $\sigma_w'=\sigma_w$ for $w\notin U$. \vskip 0.2cm

\noindent \textbf{return} $\sigma':V\rightarrow[q]$.
\end{mdframed}
\caption{\label{alg:idealresample} The $\IdealReSample$ subroutine; we will later modify this to obtain the actual $\ReSample$ subroutine used in Algorithm $\SamplePotts$ (cf. Figure~\ref{alg:Potts}).}
\end{figure}

To control the output distribution  of the $\IdealReSample$ subroutine, the following definition will be crucial. 
\begin{definition}
Let $B>0$.  Suppose that $G=(V,E)$ is a graph and that $u,v$ are vertices in $G$. For a set $U\subseteq V$ such that $u,v\in U$, let
\[\Corr_G(U,u,v)=\bigg|\frac{\pi_{G[U]}\big(\eta_u=1,\eta_v=1\big)}{\pi_{G[U]}\big(\eta_u=1,\eta_v=2\big)}-1\bigg|,\]
where $\eta$ denotes a configuration $U\rightarrow\{1,2\}$ sampled according to $\pi_{G[U]}$. Note that $\Corr_G(U,u,v)$ measures the correlation between $u$ and $v$ in the Ising distribution with parameter $B$ on the subgraph $G[U]$.
\end{definition}

To state the main lemma of this section, we will also need the following definition of a ``random bichromatic class'' containing two specific vertices $u$ and $v$ under a configuration $\sigma$.
\begin{definition}\label{def:Usigma}
Let $G=(V,E)$ be a graph, $u,v$ be vertices in $G$  and $\sigma:V\rightarrow [q]$ be a configuration on $G$. We let $U_\sigma\subseteq V$ be a bichromatic class under $\sigma$ which contains $u$ and $v$, chosen uniformly at random among the set of all such classes if there is more than one. 

More precisely, if $\sigma_u\neq \sigma_v$, then $U_\sigma$ is the $(c_1,c_2)$-colour-class in $\sigma$ where  $c_1,c_2$ are the colours of $u$ and $v$ under $\sigma$. If $\sigma_u=\sigma_v$, then $U_\sigma$ is the $(c,c')$-colour-class in $\sigma$, where $c$ is the common colour of $u$ and $v$ under $\sigma$ and $c'$ is a uniformly random colour from $[q]\backslash\{c\}$.
\end{definition}

The following lemma will be critical for our Potts sampler. It shows how to update a Potts configuration when we add a new edge $\{u,v\}$, based on the $\IdealReSample$ subroutine. It also controls the error introduced based on the ``average correlation'' between $u$ and $v$ in a random bichromatic class that contains them.
\begin{lemma}\label{lem:samplinganti}
Let $B\in (0,1)$ and $q\geq 3$ be an integer, and consider arbitrary   $\epsilon\in (0,1)$.

Let $G=(V,E)$ be a graph and $\mu_G$ be the Potts distribution on $G$ with parameter $B$. Suppose that $u,v$ are vertices in $G$ such that $\{u,v\}\not\in E$ and
\[\mathbf{E}\big[\Corr_G(U_\sigma,u,v)\big]\leq \epsilon,\]
where the expectation is over the choice of a random configuration $\sigma\sim \mu_G$ and the choice of a random bichromatic class $U_\sigma\subseteq V$ containing $u$ and $v$ under $\sigma$  (cf. Definition~\ref{def:Usigma}). 

Consider the graph $G'=(V,E')$ obtained from $G$ by adding the edge $\{u,v\}$. Sample a configuration $\sigma': V\rightarrow [q]$ as follows. First, sample $\sigma:V\rightarrow[q]$ according to $\mu_G$. Then, if $\sigma_u\neq \sigma_v$, set $\sigma'=\sigma$; otherwise, set $\sigma'=\IdealReSample(G,u,v,\sigma)$. Then, the distribution of $\sigma'$, denoted by $\nu_{\sigma'}$, is within total variation distance $2\epsilon/B$ from the Potts distribution $\mu_{G'}$ on $G'$ with parameter $B$, i.e.,
\[\norm{\nu_{\sigma'}-\mu_{G'}}_{\TV}\leq 2 \epsilon/B.\]
\end{lemma}
\begin{proof}
We begin with a few definitions that will be used throughout the proof. Fix distinct colours $c_1,c_2\in [q]$ and a set $U\subseteq V$ such that $u,v\in U$. Let $\Omega(U,c_1,c_2)$ be the set of configurations such that $U$ is the $(c_1,c_2)$-colour-class, i.e., 
\[\Omega(U,c_1,c_2)=\big\{\eta:V\rightarrow[q]\mid  U=\eta^{-1}(c_1,c_2)\big\}.\]
We will be interested in two particular types of configurations in $\Omega(U,c_1,c_2)$, those where $u,v$ take the colours $c_1,c_2$ and those where $u,v$ take the colours $c_1,c_1$. Namely, let
\begin{align*}
\Omega_{\noteq}(U,c_1,c_2)&=\big\{\eta\in \Omega(U, c_1,c_2)\mid \eta_u=c_1,\, \eta_v=c_2 \big\},\\ 
\Omega_{\eq}(U,c_1,c_2)&=\big\{\eta\in \Omega(U, c_1,c_2)\mid \eta_u=c_1,\, \eta_v=c_1 \big\}.
\end{align*}
Note the asymmetry in the above definitions with respect to $c_1,c_2$, e.g., for $\eta\in\Omega_{\eq}(U,c_1,c_2)$, we have that $\eta_u,\eta_v\neq c_2$. For a configuration $\eta\in\Omega(U,c_1,c_2)$, we also denote by 
\[\Omega^{\eta}(U,c_1,c_2)=\big\{\tau\in \Omega(U,c_1,c_2)\mid  \tau_{V\backslash U}=\eta_{V\backslash U}\big\}\]
the set of configurations in $\Omega(U,c_1,c_2)$ that agree with $\eta$ on $V\backslash U$; we define analogously the sets $\Omega^{\eta}_{\noteq}(U,c_1,c_2),\Omega^{\eta}_{\eq}(U,c_1,c_2)$. Using Observation~\ref{obs:wrvtvte} and the Ising distribution $\pi^{c_1,c_2}_{G[U]}$ (cf.  Section~\ref{sec:vvverv3242}), we have that for any $\eta\in\Omega(U,c_1,c_2)$ it holds that
\begin{equation}\label{eq:wdcewwc123}
\mu_{G}(\eta)=\mu_G\big(\Omega^{\eta}(U,c_1,c_2)\big)\, \pi^{c_1,c_2}_{G[U]}(\eta_U),
\end{equation}
and
\begin{equation}\label{eq:12ewvete6343}
\frac{\mu_G\big(\Omega^{\eta}_{\eq}(U,c_1,c_2)\big)}{\mu_G\big(\Omega^{\eta}(U,c_1,c_2)\big)}=\pi^{c_1,c_2}_{G[U]}(\tau_u=\tau_v=c_1), \quad
\frac{\mu_G\big(\Omega^{\eta}_{\noteq}(U,c_1,c_2)\big)}{\mu_G\big(\Omega^{\eta}(U,c_1,c_2)\big)}=\pi^{c_1,c_2}_{G[U]}(\tau_u=c_1,\tau_v=c_2).
\end{equation}
Also, the assumption $\mathbf{E}\big[\Corr_G(U_\sigma,u,v)\big]\leq \epsilon$ translates into
\begin{equation}\label{eq:ECorrGuv}
\sum_{\substack{c_1,c_2\in [q];\\ c_1\neq c_2}}\sum_{\substack{U\subseteq V;\\ u,v\in U}}\bigg(\mu_G\big(\Omega_{\noteq}(U,c_1,c_2)\big)+\frac{1}{q-1}\mu_G\big(\Omega_{\eq}(U,c_1,c_2)\big)\bigg)\Corr_G(U,u,v)\leq \epsilon.
\end{equation}
To see this, for a set $U\subseteq V$ such that $u,v\in U$, we find how much $\Corr_G(U,u,v)$ contributes to $\mathbf{E}\big[\Corr_G(U_\sigma,u,v)]$. Note, for a configuration $\sigma$ to have $U_\sigma=U$ (cf. Definition~\ref{def:Usigma}), there must exist distinct colours $c_1,c_2\in [q]$ such that $U=\sigma^{-1}(c_1,c_2)$ and either (i) $\sigma(u)=c_1$, $\sigma(v)=c_2$, or (ii) $\sigma(u)=c_1$, $\sigma(v)=c_1$. In case (i), we have that $\sigma\in \Omega_{\noteq}(U,c_1,c_2)$ and $U_\sigma=U$ with probability 1. In case (ii), we have that $\sigma\in \Omega_{\eq}(U,c_1,c_2)$ and $U_\sigma=U$ with probability $1/(q-1)$. By  aggregating over the relevant $\sigma$, we therefore obtain \eqref{eq:ECorrGuv}.\vskip 0.2cm

We next proceed to the proof. Recalling that $\nu_{\sigma'}$ is the distribution of $\sigma'$, our goal is to show that
\begin{equation}\label{eq:v3tg6gg1s2}
\norm{\nu_{\sigma'}-\mu_{G'}}_{\TV}\leq 2\epsilon/B.
\end{equation} 
\vskip 0.2cm
\noindent \textbf{Weight of configurations in $\nu_{\sigma'}$.} 
For distinct colours $c_1,c_2\in [q]$ and a set $U\subseteq V$ such that $u,v\in U$, we first show that
\begin{align}
\forall\eta\in \Omega_{\eq}(U,c_1,c_2):&\quad \nu_{\sigma'}(\eta)=\frac{qB}{B+q-1}\mu_{G}(\eta),\label{eq:monoweight}\\
\forall\eta\in \Omega_{\noteq}(U,c_1,c_2):&\quad \nu_{\sigma'}(\eta)=\mu_G(\eta)+\frac{(1-B)\pi^{c_1,c_2}_{G[U],u,v}(\eta_U)}{B+q-1}\mu_G\big(\Omega^{\eta}_{\eq}(U,c_1,c_2)\big).\label{eq:eteroweight}
\end{align}  
(Recall from Definition~\ref{def:piuv} that $\pi^{c_1,c_2}_{G[U],u,v}$ is the conditional Ising distribution on $G[U]$ where $u$,$v$ take the colours $c_1,c_2$, respectively.) To see the expression for $\nu_{\sigma'}(\eta)$ in \eqref{eq:monoweight}, note that we can obtain $\eta\in\Omega_{\eq}(U,c_1,c_2)$ via the subroutine $\IdealReSample$ only if $\sigma=\eta$ and the coin flip came up heads (because $\eta_u=\eta_v=c_1$). Analogously, to see the  expression for $\nu_{\sigma'}(\eta)$ in \eqref{eq:eteroweight}, note that we can obtain $\eta\in \Omega_{\noteq}(U,c_1,c_2)$ if one of the following happens (using that  $\eta_u=c_1$ and $\eta_v=c_2$).
\begin{itemize}
\item We started with the configuration $\sigma=\eta$; this happens with probability $\mu_G(\eta)$. Then, we obtain $\eta$ with probability $1$.
\item We started with a configuration $\sigma$ such that $\sigma_u=\sigma_v=c_1$,   $\sigma^{-1}(c_1,c_2)=U$ and $\sigma_{V\backslash U}=\eta_{V\backslash U}$, i.e., $\sigma\in\Omega^{\eta}_{\eq}(U,c_1,c_2)$; this happens with probability $\mu_G\big(\Omega^{\eta}_{\eq}(U,c_1,c_2)\big)$.  Then the coin flip came tails, the colour $c'$ selected was $c'=c_2$, and the Ising configuration $\tau$  sampled was $\eta_U$. All that happens with probability $\frac{(1-B)\pi^{c_1,c_2}_{G[U],u,v}(\eta_U)}{B+q-1}$.
\end{itemize} 
Using the assumption $\mathbf{E}\big[\Corr_G(U_\sigma,u,v)\big]\leq \epsilon$ (cf. \eqref{eq:ECorrGuv}), we will later show the following bound for the ratio of the partition functions $Z_G,Z_{G'}$:
\begin{equation}\label{eq:ZratioGGpanti}
|M|\leq \epsilon/B, \mbox{ where } M:=\frac{(B+q-1)Z_G}{q Z_{G'}}-1.
\end{equation} 
Let us conclude the proof assuming, for now, \eqref{eq:ZratioGGpanti}. Then, we will prove \eqref{eq:ZratioGGpanti} in Step II below. \vskip 0.2cm

\noindent \textbf{Step I: Proof of \eqref{eq:v3tg6gg1s2} (assuming \eqref{eq:ZratioGGpanti}).} We will decompose $\norm{\nu_{\sigma'}-\mu_{G'}}_{\TV}$ as 
\begin{equation}\label{eq:tvduc}
\norm{\nu_{\sigma'}-\mu_{G'}}_{\TV}=\frac{1}{2}\sum_{\eta:V\rightarrow [q]}|\nu_{\sigma'}(\eta)-\mu_{G'}(\eta)|=\frac{1}{2}\sum_{U;\, u,v\in U}\sum_{c_1,c_2\in [q];\, c_1\neq c_2}D(U,c_1,c_2),
\end{equation}
where 
\begin{equation}\label{eq:duc12}
D(U,c_1,c_2):=\sum_{\eta\in \Omega_{\noteq}(U,c_1,c_2)}|\nu_{\sigma'}(\eta)-\mu_{G'}(\eta)|+\frac{1}{q-1}\sum_{\eta\in \Omega_{\eq}(U,c_1,c_2)}|\nu_{\sigma'}(\eta)-\mu_{G'}(\eta)|.
\end{equation}
To see the second equality in  \eqref{eq:tvduc},  fix any configuration $\eta: V\rightarrow [q]$. If $\eta_u=c_1$ and $\eta_v=c_2$ for distinct colours $c_1,c_2\in [q]$, then $|\nu_{\sigma'}(\eta)-\mu_{G'}(\eta)|$ appears only in one term $D(U,c_1,c_2)$ in the summation of \eqref{eq:tvduc}, namely for $U=\eta^{-1}(c_1,c_2)$. If $\eta_u=\eta_v=c_1$ for some colour $c_1\in [q]$, then $|\nu_{\sigma'}(\eta)-\mu_{G'}(\eta)|$ appears in exactly $q-1$ terms $D(U,c_1,c_2)$  in the summation of \eqref{eq:tvduc}, once for each colour $c_2\neq c_1$ and $U=\eta^{-1}(c_1,c_2)$. 

For $\eta\in \Omega_{\eq}(U,c_1,c_2)$, we have $\mu_{G'}(\eta)=\frac{w_{G'}(\eta)}{Z_{G'}}=\frac{Bw_{G}(\eta)}{Z_{G'}}=\frac{qB(M+1)\mu_G(\eta)}{B+q-1}$ where the last equality follows from the definition of $M$ in \eqref{eq:ZratioGGpanti}. Hence, by \eqref{eq:monoweight},
\begin{equation}\label{eq:v43t34v4t3v145}
|\nu_{\sigma'}(\eta)-\mu_{G'}(\eta)|=\frac{qB}{B+q-1}\Big|\mu_{G}(\eta)-(M+1)\mu_G(\eta)\Big|\leq |M|\,\mu_G(\eta),
\end{equation}
where the last inequality follows from $\frac{qB}{B+q-1}\leq 1$ which holds for all $B\in(0,1)$.

For $\eta\in \Omega_{\noteq}(U,c_1,c_2)$, we have $\mu_{G'}(\eta)=\frac{w_{G'}(\eta)}{Z_{G'}}=\frac{w_{G}(\eta)}{Z_{G'}}=\frac{q(M+1)}{B+q-1}\mu_G(\eta)$ using again the definition of $M$ in \eqref{eq:ZratioGGpanti}.  Hence, \mbox{by \eqref{eq:eteroweight},}
\begin{align}
|\nu_{\sigma'}(\eta)-\mu_{G'}(\eta)|&=\bigg|\frac{1-B}{B+q-1}\Big(\mu_{G}(\eta)-\pi^{c_1,c_2}_{G[U],u,v}(\eta_U)\mu_G\big(\Omega^{\eta}_{\eq}(U,c_1,c_2)\big)\Big)+\frac{qM}{B+q-1}\mu_G(\eta)\bigg|\notag\\
&\leq\Big|\mu_{G}(\eta)-\pi^{c_1,c_2}_{G[U],u,v}(\eta_U)\mu_G\big(\Omega^{\eta}_{\eq}(U,c_1,c_2)\big)\Big|+2|M|\,\mu_G(\eta),\label{eq:weqcrecr}
\end{align}
where the last inequality follows from the triangle inequality and the inequalities $\frac{1-B}{B+q-1}\leq 1$, $\frac{q}{B+q-1}\leq 2$. Note that, by the definition of $\pi^{c_1,c_2}_{G[U],u,v}(\cdot)$, for $\eta\in \Omega_{\noteq}(U,c_1,c_2)$ it holds that
\[\pi^{c_1,c_2}_{G[U],u,v}(\eta_U)=\frac{\pi^{c_1,c_2}_{G[U]}(\eta_U)}{\pi^{c_1,c_2}_{G[U]}(\tau_u=c_1,\tau_v=c_2)}\]
and therefore, using \eqref{eq:wdcewwc123} and \eqref{eq:12ewvete6343}, we have that
\begin{equation}\label{eq:322wxsxwf4tv2}
\Big|\mu_{G}(\eta)-\pi^{c_1,c_2}_{G[U],u,v}(\eta_U)\mu_G\big(\Omega^{\eta}_{\eq}(U,c_1,c_2)\big)\Big|=\mu_G(\eta)\Corr_G(U,u,v),
\end{equation}
and hence \eqref{eq:weqcrecr} gives that, for all $\eta\in \Omega_{\noteq}(U,c_1,c_2)$, it holds that
\begin{equation}\label{eq:weqcrecr12}
|\nu_{\sigma'}(\eta)-\mu_{G'}(\eta)|\leq \big(2|M|+\Corr_G(U,u,v)\big)\mu_G(\eta).
\end{equation}

Summing \eqref{eq:v43t34v4t3v145} and \eqref{eq:weqcrecr12} over the relevant configurations $\eta$, we obtain that, for all $U\subseteq V$ with $u,v\in U$ and distinct colours $c_1,c_2\in [q]$, it holds that
\begin{equation}\label{eq:Duc1c2a1313} 
D(U,c_1,c_2)\leq  \big(2|M|+\Corr_G(U,u,v)\big)\bigg(\mu_G\big(\Omega_{\noteq}(U,c_1,c_2)\big)+\frac{\mu_G\big(\Omega_{\eq}(U,c_1,c_2)\big)}{q-1}\bigg), 
\end{equation}
To conclude the proof of \eqref{eq:v3tg6gg1s2}, note that analogously to \eqref{eq:tvduc} we have that
\begin{equation}\label{eq:edca2s134}
\sum_{U;\, u,v\in U}\sum_{c_1,c_2\in [q];\, c_1\neq c_2}\Big(\mu_G\big(\Omega_{\noteq}(U,c_1,c_2)\big)+\frac{\mu_G\big(\Omega_{\eq}(U,c_1,c_2)\big)}{q-1}\Big)=1.
\end{equation}
Now consider the expression for $\norm{\nu_{\sigma'}-\mu_{G'}}_{\TV}$ from \eqref{eq:tvduc}. We bound each term $D(U,c_1,c_2)$ using \eqref{eq:Duc1c2a1313} and then apply \eqref{eq:ECorrGuv} and \eqref{eq:edca2s134} to obtain that
\begin{equation*}
\norm{\nu_{\sigma'}-\mu_{G'}}_{\TV}\leq \frac{1}{2}(2|M|+\epsilon)\leq \tfrac{3}{2B}\epsilon\leq 2\epsilon/B,
\end{equation*}
where the last inequality follows from $|M|\leq \epsilon/B$ (cf. \eqref{eq:ZratioGGpanti}). This finishes the proof of \eqref{eq:v3tg6gg1s2}, modulo the proof of \eqref{eq:ZratioGGpanti} which is given below.\vskip 0.2cm

\noindent \textbf{Step II: Proof of \eqref{eq:ZratioGGpanti}.} To prove \eqref{eq:ZratioGGpanti}, it will be useful to define
\begin{align}
X(U,c_1,c_2)&:=\sum_{\eta\in \Omega_{\noteq}(U,c_1,c_2)}w_G(\eta),\notag\\
Y(U,c_1,c_2)&:=\sum_{\eta\in \Omega_{\eq}(U,c_1,c_2)}w_G(\eta),\notag\\
Z(U,c_1,c_2)&:=X(U,c_1,c_2)+\frac{1}{q-1}Y(U,c_1,c_2).\label{eq:4byb335f34v}
\end{align}
Analogously to \eqref{eq:tvduc} (and \eqref{eq:edca2s134}), we have that
\begin{equation}\label{eq:vvcrwe42vwfr}
Z_G=\sum_{c_1\neq c_2}\sum_{U; u,v\in U}Z(U,c_1,c_2).
\end{equation}
Using that $w_{G'}(\eta)=w_G(\eta)$ for $\eta\in \Omega_{\noteq}(U,c_1,c_2)$ and $w_{G'}(\eta)=B w_G(\eta)$ for $\eta\in \Omega_{\eq}(U,c_1,c_2)$, we obtain that 
\begin{equation}\label{eq:4byb335f34vb}
Z_{G'}=\sum_{c_1\neq c_2}\sum_{U; u,v\in U}Z'(U,c_1,c_2), \mbox{ where } Z'(U,c_1,c_2):=X(U,c_1,c_2)+\frac{B}{q-1}Y(U,c_1,c_2).
\end{equation}
We first show that, for all $U\subseteq V$ with $u,v\in U$ and distinct colours $c_1,c_2\in [q]$, it holds that
\begin{equation}\label{eq:4vt564g345}
\Big|\frac{(B+q-1)Z(U,c_1,c_2)}{q Z'(U,c_1,c_2)}-1\Big|\leq \frac{1}{q}\Corr_G(U,u,v). 
\end{equation}
Observe  that 
\begin{equation}\label{eq:ev436x54}
\frac{(B+q-1)Z(U,c_1,c_2)}{q Z'(U,c_1,c_2)}-1=\frac{1-B}{q}\frac{Y(U,c_1,c_2)-X(U,c_1,c_2)}{X(U,c_1,c_2)+\frac{B}{q-1}Y(U,c_1,c_2)}.
\end{equation}
Observe also that (cf. \eqref{eq:12ewvete6343}) $\Big|\frac{Y(U,c_1,c_2)}{X(U,c_1,c_2)} -1\Big|=\Corr_G(U,u,v)$,  which combined with  \eqref{eq:ev436x54} (and ignoring the positive term $\frac{B}{q-1}Y(U,c_1,c_2)$ in the denominator of the r.h.s. in the latter)  gives \eqref{eq:4vt564g345}, using also that $B\in (0,1)$.  

Using \eqref{eq:vvcrwe42vwfr} and \eqref{eq:4byb335f34vb}, we have that
\begin{equation}\label{eq:r2rf25t5h67}
\frac{(B+q-1)Z_G}{qZ_{G'}}-1=\sum_{c_1\neq c_2}\sum_{U;\, u,v\in U}\frac{qZ'(U,c_1,c_2)}{Z_{G'}}\bigg(\frac{(B+q-1)Z(U,c_1,c_2)}{qZ'(U,c_1,c_2)}-1\bigg).
\end{equation}
Since $B\in (0,1)$, we have for all $U\subseteq V$ with $u,v\in U$ and distinct colours $c_1,c_2\in [q]$ that $Z'(U,c_1,c_2)\leq Z(U,c_1,c_2)$ and $Z_{G'}\geq BZ_G$, therefore 
\begin{equation}\label{eq:t3tbt3vf5f38}
\frac{Z'(U,c_1,c_2)}{Z_{G'}}\leq \frac{Z(U,c_1,c_2)}{BZ_{G}}=\frac{1}{B}\bigg(\mu_G\big(\Omega_{\noteq}(U,c_1,c_2)\big)+\frac{\mu_G\big(\Omega_{\eq}(U,c_1,c_2)\big)}{q-1}\bigg).
\end{equation}
Combining \eqref{eq:4vt564g345}, \eqref{eq:r2rf25t5h67} and \eqref{eq:t3tbt3vf5f38}, we obtain by the triangle inequality that (cf. \eqref{eq:ECorrGuv})
\[\bigg|\frac{(B+q-1)Z_G}{qZ_{G'}}-1\bigg|\leq \mathbf{E}\big[\Corr_G(U_\sigma,u,v)\big]/B\leq \epsilon/B,\]   
thus completing the proof of \eqref{eq:ZratioGGpanti}, and therefore the proof of Lemma~\ref{lem:samplinganti}.
\end{proof}

\subsection{The ReSample subroutine}\label{sec:resamplereal}
In this section, we modify the $\IdealReSample$ subroutine of Section~\ref{sec:t5g45g45g5} to make it computationally efficient; this will give us the actual $\ReSample$ subroutine that we will use in our sampling algorithm for the Potts model. To describe the $\ReSample$ subroutine, we will need the following definition.
\begin{definition}
Let $b,M>0$ be constants. Let $G=(V,E)$ be an $n$-vertex graph and let $\sigma$ be a configuration on $G$. We say that a bichromatic component in $\sigma$ is \emph{$(b,M)$-good} if it has average growth $b$ up to depth $L=\left\lceil M\log n\right\rceil$; we say that it is \emph{$(b,M)$-bad} otherwise. Analogously, we say that $\sigma$ is \emph{$(b,M)$-good} if \emph{all} bichromatic components in $\sigma$ are good; otherwise, we say that  $\sigma$ is $(b,M)$-bad. 
\end{definition}
We are now able to describe
 the $\ReSample$ subroutine, which takes as inputs a graph $G$, two vertices $u$ and $v$ of $G$ and a configuration $\sigma$ on $G$ such that $\sigma(u)=\sigma(v)$. The subroutine is given in Figure~\ref{alg:resample}.
\begin{figure}[h]
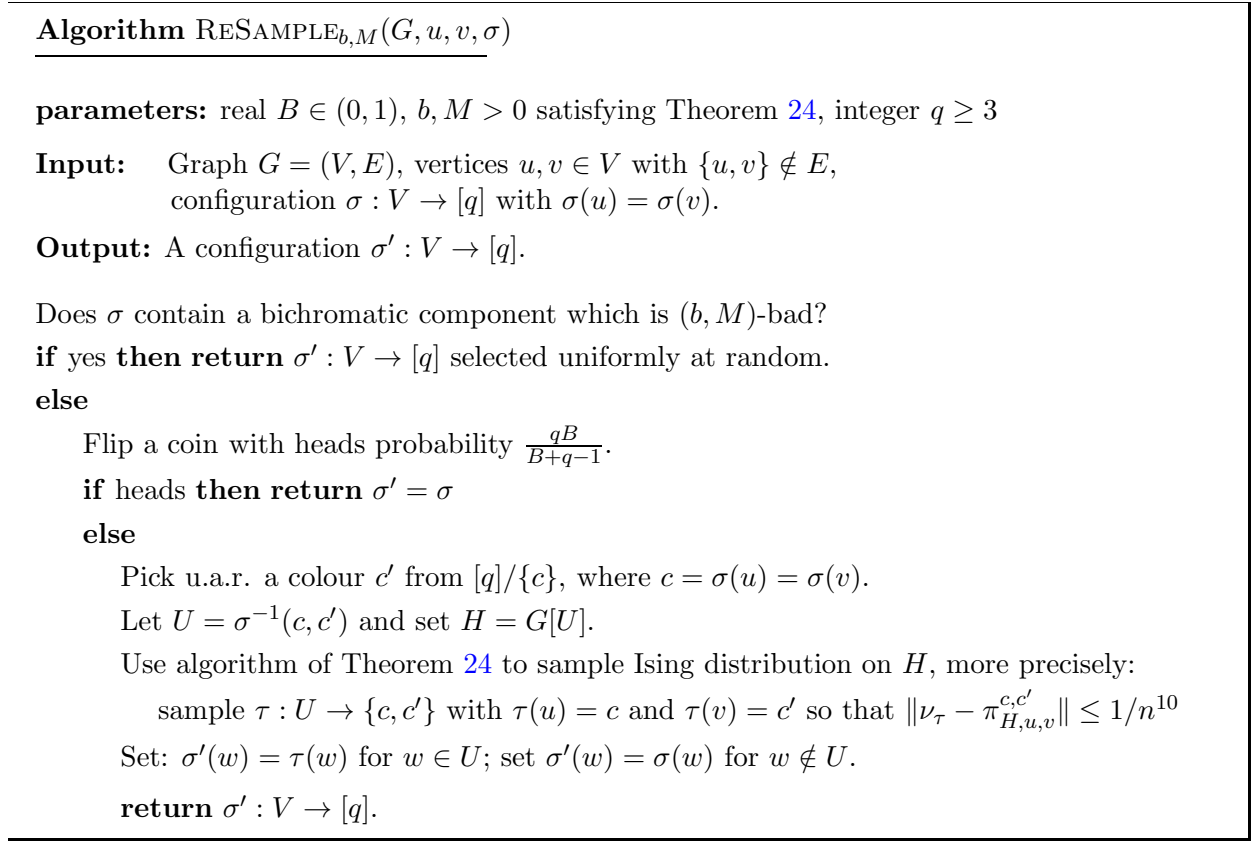

\begin{mdframed}
\textbf{Algorithm} $\ReSample_{b,M}(G,u,v,\sigma)$ \vskip 0.05cm
\rule[0.3cm]{6cm}{0.4pt}
\vskip 0.05cm

\noindent \textbf{parameters:} real $B\in (0,1)$, $b,M>0$ satisfying Theorem~\ref{thm:sampleising}, integer $q\geq 3$ \vskip 0.25cm

\noindent \textbf{Input:} \phantom{ \ \ }Graph $G=(V,E)$, vertices $u,v\in V$ with $\{u,v\}\notin E$,\\
\phantom{ \ \ \ \ \ \ \ \ \ \ \ \ } configuration $\sigma:V\rightarrow[q]$ with $\sigma(u)=\sigma(v)$.\vskip 0.15cm

\noindent \textbf{Output:} A configuration $\sigma':V\rightarrow[q]$.
\vskip 0.4cm

Does $\sigma$ contain a bichromatic component which is $(b,M)$-bad? \vskip 0.1cm
\noindent \textbf{if} yes \textbf{then} \textbf{return} $\sigma':V\rightarrow[q]$ selected uniformly at random.\vskip 0.1cm 
\noindent \textbf{else} \vskip 0.1cm 
\noindent \hspace{0.5cm} Flip a coin with heads probability $\frac{qB}{B+q-1}$.\vskip 0.1cm
\noindent \hspace{0.5cm} \textbf{if} heads \textbf{then} \textbf{return} $\sigma'=\sigma$\vskip 0.1cm
\noindent \hspace{0.5cm} \textbf{else} \vskip 0.1cm
\noindent \hspace{1cm} Pick u.a.r. a colour $c'$ from $[q]/\{c\}$, where $c=\sigma(u)=\sigma(v)$.  \vskip 0.1cm
\noindent \hspace{1cm} Let $U=\sigma^{-1}(c,c')$ and set $H=G[U]$. \vskip 0.1cm
\noindent \hspace{1cm} Use algorithm of Theorem~\ref{thm:sampleising} to sample Ising distribution on $H$, more precisely:\vskip 0.1cm 
\noindent \hspace{1.5cm} sample $\tau:U\rightarrow \{c,c'\}$ with $\tau(u)=c$ and $\tau(v)=c'$ so that $\lVert{\nu_\tau-\pi^{c,c'}_{H,u,v}\rVert}\leq 1/n^{\cons}$  \vskip 0.15cm
\noindent \hspace{1cm}  Set: $\sigma'(w)=\tau(w)$ for $w\in U$;   set $\sigma'(w)=\sigma(w)$ for $w\notin U$. \vskip 0.2cm

\noindent\hspace{1cm} \textbf{return} $\sigma':V\rightarrow[q]$.
\end{mdframed}
\caption{\label{alg:resample} The $\ReSample$ subroutine used in Algorithm $\SamplePotts$ (cf. Figure~\ref{alg:Potts}).}
\end{figure}

\subsection{Analysis of the Potts algorithm}\label{sec:thyh55h6}
We now have all the pieces to give our algorithm for sampling from the antiferromagnetic Potts model, see the algorithm $\SamplePotts$ in Figure~\ref{alg:Potts}. We next prove the following theorem, which details the performance of the algorithm $\SamplePotts$ on random $\Delta$-regular graphs and yields as an immediate corollary Theorem~\ref{thm:mainPotts}.
\begin{figure}[h]
\begin{mdframed}
\textbf{Algorithm} $\SamplePotts(G)$ \vskip 0.05cm
\rule[0.3cm]{6cm}{0.4pt}
\vskip 0.05cm
\noindent \textbf{parameters:} real $B\in(0,1)$, $b,M>0$ satisfying Theorem~\ref{thm:sampleising}, integers $q,\Delta\geq 3$
\vskip 0.2cm
\noindent \textbf{Input:} Graph $G=(V,E)$ 

\noindent \textbf{Output:} Either \textsc{Fail} or an assignment $\sigma:V\rightarrow[q]$

\vskip 0.4cm

\noindent $E':= \{e \in E\mid \mbox{$e$ belongs to a short cycle}\}$ 

\vskip 0.2cm

\noindent \textbf{if} \mbox{$G'=(V,E')$ contains a component which is neither a cycle nor an isolated vertex} \vskip 0.1cm \noindent \hspace{0.5cm}\textbf{then} \textbf{return} \textsc{Fail}

\noindent \textbf{else} \vskip 0.1cm
\noindent \hspace{0.5cm} Sample a configuration $\sigma':V\rightarrow [q]$ on $G'$ (according to $\mu_{G'}$); \vskip 0.1cm
\noindent \hspace{0.5cm} Let $e_1,e_2,\hdots, e_t$ be the edges in $E\backslash E'$; set $G_t=G'$ and $\sigma_t=\sigma'$.\vskip 0.1cm
\noindent\hspace{0.5cm} \textbf{for} $j=t$ \textbf{downto} 1:\vskip 0.1cm

\noindent\hspace{1cm} Suppose that $e_j=(u_j,v_j)$;\vskip 0.1cm
\noindent\hspace{1cm} \textbf{if} $\sigma_j(u_j)=\sigma_j(v_j)$ \textbf{ then } $\sigma_{j-1}=\ReSample_{b,M}(G_j,u_j,v_j,\sigma_j)$ \textbf{ else } $\sigma_{j-1}=\sigma_j$;\vskip 0.1cm
\noindent\hspace{1cm} Obtain the graph $G_{j-1}$ by adding the edge $\{u_j,v_j\}$ in $G_j$\vskip 0.1cm
\noindent\hspace{0.5cm} \textbf{end}
\vskip0.2cm
\noindent\hspace{0.5cm} \textbf{return} $\sigma=\sigma_0$.
\end{mdframed}
\caption{\label{alg:Potts} Algorithm for sampling a Potts configuration in the antiferromagnetic case $B\in (0,1)$. The details of the $\ReSample$ subroutine are given in Figure~\ref{alg:resample}. While the algorithm can also be modified to work for the ferromagnetic case $B>1$, we will instead use a simpler percolation algorithm via the random-cluster representation.}
\end{figure}

\begin{theorem}\label{thm:mainPotts2}
Let $\Delta\geq 3$, $q\geq 3$ and $B\in (0,1)$ be in the uniqueness regime of the $(\Delta-1)$-ary tree. Then, there exists constants $b,M,\delta>0$ such that, as $n\rightarrow \infty$, the following holds   with probability $1-o(1)$ over the choice of a random $\Delta$-regular graph $G=(V,E)$ with $n$ vertices.

The output of the algorithm $\SamplePotts(G)$ (cf. Figure~\ref{alg:Potts}) is an assignment $\sigma:V\rightarrow [q]$ whose distribution $\nu_\sigma$ is within total variation distance $O(1/n^{\delta})$ from the Potts distribution $\mu_G$ with parameter $B$, i.e.,
\[\norm{\nu_\sigma-\mu_{G}}_{\TV}=O(1/n^{\delta}).\]
\end{theorem}
\begin{proof}
Since $B$ is in the uniqueness regime of the $(\Delta-1)$-ary tree and $B\neq \frac{\Delta-q}{\Delta}$,  we have that $B>\frac{\Delta-q}{\Delta}$ (cf. Remark~\ref{rem:6gtf5q}). It follows that $\frac{1-B}{B+q-1}<\frac{1}{\Delta-1}$ and therefore there exists $\epsilon'>0$ such that
\begin{equation}\label{eq:choiceepsp}
\Big(\frac{1+B}{B+q-1}+3\epsilon'\Big)\frac{1-B}{1+B}<\frac{1}{\Delta-1}.
\end{equation} 
Let $K<\frac{1+B}{B+q-1}+\epsilon'$ and $\epsilon>0$ be the constants in Lemma~\ref{lem:treespathPotts} corresponding to $\epsilon'$, and let $h',\ell'$ be positive constants such that Lemma~\ref{lem:treespathPotts} applies  for all integers $h\geq h'$ and $\ell\geq \ell'$. Fix $h$ to be any integer greater than $h'$. Let $\ell_1>0$ be the constant in Lemma~\ref{lem:graphneighbor} corresponding to the values of $\epsilon$ and $h$. Let $\delta>0$ be the constant in Lemma~\ref{lem:numpaths} corresponding to $\ell_0:=1/(5\log (\Delta-1))$ and $W:=1/(K \frac{1-B}{1+B})$ (note that \eqref{eq:choiceepsp} guarantees that $W>\Delta-1$). Let also 
\begin{equation*}
b':=(\Delta-1)\Big(\frac{1+B}{B+q-1}+2\epsilon'\Big), \mbox{ and }
b:=(\Delta-1)\Big(\frac{1+B}{B+q-1}+3\epsilon'\Big).
\end{equation*}  
Let $M_0,M_0'$ be the constants in Theorem~\ref{thm:sampleising} and Lemma~\ref{lem:corrdecayIsing}, respectively. Let $M$ be sufficiently large so that $M>\max\{2\ell_0,2M_0,2M_0'\}$ and the following inequalities hold (for all sufficiently large $n$):
\begin{equation}\label{eq:v4t4447778}
\Delta\big((\Delta-1)K\big)^{\left\lceil M\log n\right\rceil}\leq (b')^{\left\lceil M \log n\right\rceil} \mbox{ and } (b'/b)^{\left\lceil M \log n\right\rceil}\leq 1/n^{11}.
\end{equation}
Note that such an $M$ exists since $(\Delta-1)K<b'<b$. Finally, set 
\[L_0:=\lceil \ell_0 \log n\rceil=\lceil\tfrac{1}{5}\log_{\Delta-1}n\rceil, \quad L:=\left\lceil M \log n\right\rceil.\]

Taking a union bound over Lemmas~\ref{lem:numpaths} and~\ref{lem:graphneighbor}, we have that, for all sufficiently large $n$, a uniformly random $\Delta$-regular graph $G=(V,E)$ with $n$ vertices satisfies the following with probability $1-o(1)$ over the choice of the graph:
\begin{enumerate}
\item \label{it:numpaths1qa} $\displaystyle\sum_{\ell=L_0}^{L}\ell C_\ell\Big(K\frac{1-B}{1+B}\Big)^{\ell}\leq 1/(2n^\delta)$, where  $C_\ell$ is the number of cycles of length $\ell$ in $G$.
\item \label{it:lemgraphneighbor1qa} every path  $P$ in $G$ with $\ell$ vertices where $\ell_1\leq \ell\leq L$  has an $h$-graph-neighbourhood with at least $(1-\epsilon)\ell$ isolated tree components.
\end{enumerate}

Fix any $\Delta$-regular graph $G$ which satisfies Items~\ref{it:numpaths} and~\ref{it:lemgraphneighbor}. The theorem will follow by showing that the output of $\SamplePotts(G)$ (cf. Figure~\ref{alg:Potts}) is an assignment $\sigma:V\rightarrow [q]$ whose distribution $\nu_\sigma$ is within total variation distance $O(1/n^{\delta})$ from the Potts measure $\mu_G$ with parameter $B$, i.e.,
\[\norm{\nu_\sigma-\mu_{G}}_{\TV}=O(1/n^{\delta}).\]

To do this, as in the algorithm $\SamplePotts(G)$, let $e_1=\{u_1,v_1\},\hdots, e_t=\{u_t,v_t\}$ be  the edges of $G$ that do not belong to short cycles (i.e., cycles of length $\leq \ell_0\log n$). For $j\in \{0,1,\hdots,t\}$, let  $G_j$ be the subgraph $G\backslash\{e_1,\hdots,e_j\}$ and $\mu_{j}$ be the Potts distribution on $G_j$ with parameter $B$; note that the graphs $G_j$ are defined exactly as in the algorithm $\SamplePotts(G)$. We will use $\hat{\sigma}_j$ to denote a random configuration distributed according to $\mu_j$; note that $\sigma_j$ is used to denote the configuration considered by the algorithm at the beginning of the step $j$ of the algorithm and, as we shall see soon, its distribution is close to that of $\hat{\sigma}_j$ on $G_j$.  

For an integer $\ell\geq 1$, denote by $\mathcal{P}_{\ell,j}$  the set of paths of length $\ell$ that connect $u_j$ and $v_j$ in $G_j$ and by $P_{\ell,j}=|\mathcal{P}_{\ell,j}|$ the number of all such paths. Note that 
\begin{equation}\label{eq:4b5u67n969}
\mbox{$\sum^{t}_{j=1}$}\, P_{\ell,j}\leq \ell C_\ell,
\end{equation}
since every path with $\ell$ vertices connecting the endpoints of an edge $\{u_j,v_j\}$ in $G_j$ maps to a cycle with $\ell$ vertices in the initial graph $G$ (by adding the edge $\{u_j,v_j\}$), and each cycle with  $\ell$ vertices in $G$ can potentially arise at most $\ell$ times under this mapping.  Let also
\begin{equation}\label{eq:15v3v47g7h7}
\epsilon_j:=\frac{6}{n^{\cons}}+\sum^{L}_{\ell=L_0}P_{\ell,j}\Big(K\frac{1-B}{1+B}\Big)^{\ell}.
\end{equation}
Let also $\Omega_j(b,M)$ be the set of all $(b,M)$-bad configurations on $G_j$. Using the fact that $G$ satisfies Item~\ref{it:lemgraphneighbor1qa}, we will show that for all $j=1,\hdots,t$ it holds that
\begin{gather}
\mu_{G_j}\big(\Omega_j(b,M))\leq 1/n^{\cons},\label{eq:c34vtty}\\
\mathbf{E}_j\big[\Corr_{G_j}(U_{\hat{\sigma}_j},u_j,v_j)\big]\leq 2\epsilon_j/B^{\Delta},\label{eq:4bv3c3cf5}
\end{gather}
where the expectation in \eqref{eq:4bv3c3cf5} is over the choice of a random configuration $\hat{\sigma}_j$ distributed according to $\mu_j$ and over the choice of the random bichromatic class $U_{\hat{\sigma}_j}$ containing $u_j$ and $v_j$ under $\hat{\sigma}_j$.

We will prove \eqref{eq:c34vtty} and \eqref{eq:4bv3c3cf5} shortly, but let us assume them for now and conclude the proof of the theorem.  We will prove by induction that for all $j\in\{0,1,\hdots,t\}$ it holds that
\begin{equation}\label{eq:89v3tvt4tv}
\big\|\nu_{\sigma_j}-\mu_{G_j}\big\|_{\TV}\leq \frac{6(t-j)}{n^{\cons}}+\frac{4}{B^{\Delta+1}}\sum^{t}_{j'=j+1}\epsilon_{j'}.
\end{equation}
For $j=t$, the result holds trivially since $\sigma_t$ is distributed as $\mu_{G_t}$ (exactly). Assume that \eqref{eq:89v3tvt4tv} holds for $j$ where $j\in\{1,\hdots, t\}$, we will also show that it holds for $j-1$. To do this, let us consider the configuration $\sigma_{j-1}'$ defined as follows: 
\[\mbox{ if } \hat{\sigma}_{j}(u_j)\neq\hat{\sigma}_{j}(v_j), \mbox{ then } \sigma_{j-1}'=\hat{\sigma}_{j} \mbox{ else }\sigma_{j-1}'=\IdealReSample(G_j,u_j,v_j,\hat{\sigma}_{j}),\]
i.e., $\sigma_{j-1}'$ is obtained using at the $j$-th step of the algorithm the  random configuration $\hat{\sigma}_{j}$ distributed according to $\mu_{G_j}$ (exactly). In contrast, note that 
\[\mbox{ if } \sigma_{j}(u_j)\neq\sigma_{j}(v_j), \mbox{ then } \sigma_{j-1}=\sigma_{j} \mbox{ else }\sigma_{j-1}=\ReSample_{b,M}(G_j,u_j,v_j,\sigma_{j}).\] Let $\nu_{\sigma_{j-1}'}$ denote the distribution of $\sigma_{j-1}'$. Then, by Lemma~\ref{lem:samplinganti} applied to the graph $G_j$ and the inequality in \eqref{eq:4bv3c3cf5} we have that
\begin{equation}\label{eq:tvt4vv41111}
\big\|\nu_{\sigma_{j-1}'}-\mu_{G_{j-1}}\big\|_{\TV}\leq \frac{2}{B}\mathbf{E}_j[\Corr_{G_j}(U_{\hat{\sigma}_j},u_j,v_j)]\leq \frac{4}{B^{\Delta+1}}\epsilon_j.
\end{equation}
Since $\hat{\sigma}_j$ is distributed according to $\mu_{G_j}$, by the Coupling Lemma, there exists a coupling $\Pr(\cdot)$ of $\sigma_{j}$ and $\hat{\sigma}_j$ such  that 
\begin{equation}\label{eq:53g5t4gg}
\Pr(\sigma_j\neq \hat{\sigma}_j)=\big\|\nu_{\sigma_j}-\mu_{G_j}\big\|_{\TV}.
\end{equation}
Note also that, for a configuration $\eta\notin \Omega_j(b,M)$, conditioned on $\sigma_{j}=\eta$ and $\hat{\sigma}_{j}=\eta$, we can couple $\sigma_{j-1}$ and $\sigma_{j-1}'$ so that $\sigma_{j-1}\neq \sigma_{j-1}'$ with probability at most $1/n^{\cons}$. To see this,  if $\eta(u_j)=\eta(v_j)$ then we trivially have $\sigma_{j-1}'=\sigma_{j-1}=\eta$. Otherwise, $\sigma_{j-1}$ and $\sigma_{j-1}'$ are produced by first choosing a random bichromatic class containing $u_j$ and $v_j$ under $\eta$, which we can couple so that it is the same in both $\ReSample_{b,M}(G_j,u_j,v_j,\eta)$ and $\IdealReSample(G_j,u_j,v_j,\eta)$. Denote this class by $U$ and let $c_1=\eta(u_j)$, $c_2=\eta(v_j)$. Then, from the definition of $\IdealReSample(G_j,u_j,v_j,\eta)$ we have that the distribution of $\sigma'_{j-1}(U)$ is given by $\pi^{c_1,c_2}_{G_j[U], u_j,v_j}$, while from $\ReSample_{b,M}(G_j,u_j,v_j,\eta)$ we have that the distribution of $\sigma_{j-1}(U)$ is $1/n^{\cons}$-close to $\pi^{c_1,c_2}_{G_j[U], u_j,v_j}$. We therefore have that 
\begin{equation}\label{eq:4gvt45v45h6}
\Pr\big(\sigma_{j-1}\neq \sigma_{j-1}'\mid \sigma_{j}=\hat{\sigma}_j\notin \Omega_j(b,M)\big)\leq 1/n^{\cons}.
\end{equation}
Invoking the Coupling Lemma again, we therefore obtain that 
\begin{align}
\big\|\nu_{\sigma_{j-1}}-\nu_{\sigma_{j-1}'}\big\|_{\TV}&\leq \Pr(\sigma_j\neq \hat{\sigma}_j)+\Pr\big(\sigma_j=\hat{\sigma}_j\in \Omega_j(b,M)\big)\notag\\
&\hskip 5cm+\Pr\big(\sigma_{j-1}\neq \sigma_{j-1}'\mid \sigma_j=\hat{\sigma}_j\notin \Omega_j(b,M)\big)\notag\\
&\leq \big\|\nu_{\sigma_j}-\mu_{G_j}\big\|_{\TV}+\mu_{G_j}(\Omega_j(b,M))+ 1/n^{\cons},\label{eq:45gt4222qa}
\end{align}
where in the last inequality we used \eqref{eq:53g5t4gg}, \eqref{eq:4gvt45v45h6} and 
\[\Pr\big(\sigma_j=\hat{\sigma}_j\in \Omega_j(b,M)\big)\leq \Pr\big(\hat{\sigma}_j\in \Omega_j(b,M)\big)=\mu_{G_j}(\Omega_j(b,M)).\]
Combining \eqref{eq:45gt4222qa} with the inductive hypothesis \eqref{eq:89v3tvt4tv} and \eqref{eq:c34vtty}, we obtain 
\[\big\|\nu_{\sigma_{j-1}}-\nu_{\sigma_{j-1}'}\big\|_{\TV}\leq \frac{6(t-j+1)}{n^{\cons}}+\frac{4}{B^{\Delta+1}}\sum^{t}_{j'=j+1}\epsilon_{j'}.\]
Using this and \eqref{eq:tvt4vv41111} completes (via the triangle inequality) the inductive step, therefore completing the proof of  \eqref{eq:89v3tvt4tv} for all $j\in \{0,1,\hdots,t\}$ as wanted.

We have $G_0=G$ and $\sigma_0=\sigma$, so \eqref{eq:89v3tvt4tv} for $j=0$ gives
\[\big\|\nu_{\sigma}-\mu_{G}\big\|_{\TV}\leq \frac{6t}{n^{\cons}}+\frac{4}{B^{\Delta+1}}\sum^{t}_{j=1}\epsilon_j\leq \frac{1}{n^{8}}+\frac{4}{B^{\Delta+1}}\sum^{t}_{j=1}\sum^{L}_{\ell=L_0}P_{\ell,j}\Big(K\frac{1-B}{1+B}\Big)^{\ell},\]
where the last inequality holds for all sufficiently large $n$ using the values of $\epsilon_j$ from \eqref{eq:15v3v47g7h7} (and the crude bound $t\leq |E|\leq \Delta n$). Using \eqref{eq:4b5u67n969}, we therefore have that 
\[\big\|\nu_{\sigma}-\mu_{G}\big\|_{\TV}\leq \frac{1}{n^{8}}+\frac{4}{B^{\Delta+1}}\sum^{L}_{\ell=L_0}\ell C_\ell\Big(K\frac{1-B}{1+B}\Big)^{\ell}\leq \frac{10}{B^{\Delta+1}n^{\delta}},\]
where the last inequality follows from our assumption that $G$ satisfies Item~\ref{it:numpaths1qa} (and by assuming that $n$ is sufficiently large). This completes the proof of Theorem~\ref{thm:mainPotts2}, modulo the proofs of \eqref{eq:c34vtty} and \eqref{eq:4bv3c3cf5} which are given below.

To prove \eqref{eq:c34vtty} and \eqref{eq:4bv3c3cf5}, fix any value $j\in\{1,\hdots,t\}$. Since $G_j$ is a subgraph of $G$ and $G$ satisfies Item~\ref{it:lemgraphneighbor1qa} every path  $P$ in $G_j$ with $\ell$ vertices where $\ell\in [\ell_1,L]$  has an $h$-graph-neighbourhood with at least $(1-\epsilon)\ell$ isolated tree components.  By Lemma~\ref{lem:treespathPotts} which, recall, applies for all $\ell\geq \ell'$ (where $\ell'$ is the constant specified in the beginning of the proof), we therefore have that 
\begin{equation}\label{eq:v3t4tv4y12}
\mu_{G_j}(\mbox{$P$ is bichromatic})\leq K^{\ell} \mbox{\ \  for any path $P$ in $G_j$ with $\ell$ vertices, $\ell\in [\max\{\ell_1,\ell'\},L]$.}
\end{equation}
We are now ready to prove \eqref{eq:c34vtty}. For a vertex $v$, let 
$X_v$ be the r.v. that counts the number of bichromatic paths with $L$ vertices that start from $v$ in a random configuration $\sigma\sim \mu_{G_j}$. Since $G_j$ has maximum degree $\leq \Delta$,  there are at most $\Delta(\Delta-1)^{L-2}$ paths starting from $v$ with $L$ vertices and each of them is bichromatic with probability at most $K^{L}$ from \eqref{eq:v3t4tv4y12}. Hence, the expectation of $X_v$ is at most 
 \[\Delta(\Delta-1)^{L-2}K^{L}\leq (b')^{L},\]
where the last inequality follows from the choice of $M$ (cf. \eqref{eq:v4t4447778}).
By Markov's inequality we therefore have that for all $v\in V$ it holds that
\[\mu_{G_j}(X_v> b^{L})\leq (b'/b)^{L}\leq 1/n^{11},\]
where the last inequality also follows from the choice of $M$ (cf. \eqref{eq:v4t4447778}). Note that for $(b,M)$-bad configurations, i.e., configurations in $\Omega_j(b,M)$,  at least one of the events $X_v>b^{L}$ occurs for some $v\in V$, and therefore  by a union bound over $v \in V$ we have that 
\[\mu_{G_j}\big(\Omega_j(b,M)\big)\leq \sum_{v\in V} \mu_{G_j}\big(X_v> b^{L}\big)\leq 1/n^{10}.\]
This finishes the proof of \eqref{eq:c34vtty}.

To prove \eqref{eq:4bv3c3cf5}, for a set $U\subseteq V$ such that $u_j,v_j\in U$, let 
\[\widehat{\Corr}_{G_j}(U,u_j,v_j):=\big|\pi_{G_j[U]}(\eta_{u_j}=1\mid \eta_{v_j}=1)-\pi_{G_j[U]}(\eta_{u_j}=1\mid \eta_{v_j}=2)\big|\]
Note, by the symmetry of the states in the Ising model, we have that $\pi_{G[U]}\big(\tau_{v_j}=1 \big)=\pi_{G[U]}\big(\tau_{v_j}=2 \big)=1/2$ and therefore
\[\Corr_{G_j}(U,u_j,v_j)=\bigg|\frac{\pi_{G[U]}\big(\tau_{u_j}=1,\tau_{v_j}=1 \big)}{\pi_{G[U]}\big(\tau_{u_j}=1,\tau_{v_j}=2\big)}
-1\bigg|=\bigg|\frac{\pi_{G[U]}\big(\tau_{u_j}=1\mid \tau_{v_j}=1\big)}{\pi_{G[U]}\big(\tau_{u_j}=1\mid \tau_{v_j}=2\big)}
-1\bigg|.\]
By Lemma~\ref{lem:f4f35f63g} and since $G_j$ has maximum degree  $\leq \Delta$, we obtain that 
\[\Corr_{G_j}(U,u_j,v_j)\leq \frac{2}{B^{\Delta}}\widehat{\Corr}_{G_j}(U,u_j,v_j),\]
and therefore, to prove \eqref{eq:4bv3c3cf5} we only need to show that
\begin{equation}\label{eq:v4tvy4444412}
\mathbf{E}_j\big[\widehat{\Corr}_{G_j}(U_{\hat{\sigma}_j},u_j,v_j)\big]\leq \epsilon_j.
\end{equation}
For convenience, let $\mathcal{F}_j$ be the event that the Potts configuration $\hat{\sigma}_j$ is $(b,M)$-bad, i.e., $\mathcal{F}_j=\{\hat{\sigma}_j\in \Omega_j(b,M)\}$, and denote by $\overline{\mathcal{F}_j}$  the event that $\hat{\sigma}_j\notin \Omega_j(b,M)$.
\begin{align}
\mathbf{E}_j\big[\widehat{\Corr}_{G_j}(U_{\hat{\sigma}_j},u_j,v_j)\big]&=\mathbf{E}_j\big[\widehat{\Corr}_{G_j}(U_{\hat{\sigma}_j},u_j,v_j) \mid \mathcal{F}_j\big]\mu_j\big(\mathcal{F}_j\big)+\mathbf{E}_j\big[\widehat{\Corr}_{G_j}(U_{\hat{\sigma}_j},u_j,v_j) \mid \overline{\mathcal{F}_j}\big]\mu_j\big(\overline{\mathcal{F}_j}\big)\notag\\
& \leq \frac{1}{n^{\cons}}+\mathbf{E}_j\big[\widehat{\Corr}_{G_j}(U_{\hat{\sigma}_j},u_j,v_j) \mid \overline{\mathcal{F}_j}\big],\label{eq:t4v4vvy4bb121}
\end{align}
where the inequality follows from the bounds $\widehat{\Corr}_{G_j}(U_{\hat{\sigma}_j},u_j,v_j)\leq 1$ and $\mu_j\big(\mathcal{F}_j\big)\leq 1/n^{\cons}$ from \eqref{eq:c34vtty}. 
Recall that $\mathcal{P}_{\ell,j}$ denotes the set of paths of length $\ell$ that connect $u_j$ and $v_j$ in $G_j$. Consider also the r.v.
\[Y_j=\frac{1}{n^{\cons}}+\sum^{L}_{\ell=L_0}\sum_{P\in \mathcal{P}_{\ell,j}}\mathbf{1}\{\mbox{$P$ is bichromatic}\}\times \Big(\frac{1-B}{1+B}\Big)^{\ell}.\]
Recall that $U_{\hat{\sigma}_j}$ is a bichromatic class under $\hat{\sigma}_j$ which contains $u_j$ and $v_j$ (chosen uniformly at random among the set of all such classes if there is more than one).  Conditioned on the event $\hat{\sigma}_j\notin \Omega_j(b,M)$, we have that every bichromatic component under $\hat{\sigma}_j$ has average growth $b$ up to depth $L=\left\lceil M\log n\right\rceil$ and therefore, irrespective of  the random choice of  $U_{\hat{\sigma}_j}$, we obtain by Lemma~\ref{lem:corrdecayIsing} that $\widehat{\Corr}_{G_j}(U_{\hat{\sigma}_j},u_j,v_j)\leq Y_j$. Note, in applying Lemma~\ref{lem:corrdecayIsing}, we used that $P_{\ell,j}=0$ for all $\ell\in [1,L_0]$; this holds because  the edge $\{u_j,v_j\}$ does not belong to a short cycle in $G$ (and therefore in $G_{j-1}$ as well). We thus have that
\begin{equation}\label{eq:3r3rfdeeetyy}
\mathbf{E}_j\big[\widehat{\Corr}_{G_j}(U_{\hat{\sigma}_j},u_j,v_j) \mid \overline{\mathcal{F}_j}\big]\leq \mathbf{E}_j\big[Y_j\mid \overline{\mathcal{F}_j}\big]\leq \frac{\mathbf{E}_j[Y_j]}{\mu_{G_j}\big(\overline{\mathcal{F}_j}\big)}\leq \Big(1+\frac{2}{n^{\cons}}\Big)\mathbf{E}_j[Y_j],
\end{equation}
where the last inequality follows from $1/\mu_{G_j}\big(\overline{\mathcal{F}_j}\big)\leq 1/(1-1/n^{\cons})\leq 1+2/n^{\cons}$ (cf. \eqref{eq:c34vtty}).  Using \eqref{eq:v3t4tv4y12} and the fact that $L_0\geq \max\{\ell_1,\ell'\}$ for all sufficiently large $n$, we have that
\begin{equation}\label{eq:byyu345uunn}
\mathbf{E}_j[Y_j]\leq \frac{1}{n^{\cons}}+ \sum^{L}_{\ell=L_0}P_{\ell,j}\Big(K\frac{|B-1|}{B+1}\Big)^{\ell}.
\end{equation}
Since $G$ satisfies Item~\ref{it:numpaths1qa}, we have that the sum in \eqref{eq:byyu345uunn} is less than $1/n^{\delta}$ and hence \eqref{eq:3r3rfdeeetyy} and \eqref{eq:byyu345uunn} give that for all sufficiently large $n$ it holds that
\begin{equation*}
\mathbf{E}_j\big[\widehat{\Corr}_{G_j}(U_{\hat{\sigma}_j},u_j,v_j) \mid \overline{\mathcal{F}_j}\big]\leq \frac{1}{n^{\cons}}+\sum^{L}_{\ell=L_0}P_{\ell,j}\Big(K\frac{|B-1|}{B+1}\Big)^{\ell}+\frac{2}{n^{\cons}}\Big(\frac{1}{n^{\cons}}+\frac{1}{n^{\delta}}\Big).
\end{equation*}
Combining this with \eqref{eq:t4v4vvy4bb121} yields \eqref{eq:v4tvy4444412}, therefore concluding the proof of \eqref{eq:4bv3c3cf5}. This finishes the proof of Theorem~\ref{thm:mainPotts2}.
\end{proof}

\section{Analyzing RC on graphs with tree-like structure}\label{sec:qaz234}
In this section, we give the proof of  Lemma~\ref{lem:treespath}. We begin by revisiting the uniqueness results of H{\"a}ggstr{\"o}m \cite{Haggstrom} on the $(\Delta-1)$-ary tree; then, we use these results to obtain Lemma~\ref{lem:treespath} in Section~\ref{sec:errr3ded}.

\subsection{Uniqueness on the $(\Delta-1)$-ary tree}\label{sec:unique}

Fix an integer $\Delta\geq 3$. In this section, we review the results of H{\"a}ggstr{\"o}m \cite{Haggstrom} about uniqueness of random-cluster measures on the infinite $(\Delta-1)$-ary tree. In fact, it will be more relevant for us to consider the case of finite trees of large depth (rather than the infinite tree itself); this approach has also been followed in \cite[Chapter 10]{Grimmettbook} for the case $\Delta=3$. 

As in Section~\ref{sec:uniqueness}, let $\TreeD$ denote the infinite $(\Delta-1)$-ary tree with root vertex $\rho$. For integer $h\geq 0$, let $T_h=(V_h,E_h)$ denote the subtree of $\TreeD$ induced by the vertices at distance $\leq h$ from $\rho$ and let $L_h$ denote the leaves of $T_h$. We will consider the random-cluster distribution on $T_h$ with the so-called wired boundary condition where all the leaves are identified into a single vertex or, equivalently, all the leaves are connected to a vertex ``at infinity''.\footnote{To motivate the wired boundary condition, consider the case where we have a connected graph $G$ and we want to upper bound the probability that a vertex $v$ belongs to a large cluster in an RC configuration. Using the monotonicity of RC (cf. Lemma~\ref{lem:monotonicity}), we can restrict our attention to the graph induced by the ball of radius $h$  around the vertex $v$, by conditioning all  edges outside the ball to be open. This conditioning has exactly the same effect as the wiring we describe here.} In particular, for $S\subseteq E_h$, let $k^*(S)$ denote the number of connected components in the graph with vertex set $V_h \cup \{\infty\}$ and edge set $S\cup (L_h\times \{\infty\})$; the purpose of the extra vertex $\infty$  and the edges $L_h\times \{\infty\}$ connecting the leaves to $\infty$ is to capture the wired boundary condition that all leaves are in the same cluster. The ``wired'' RC distribution on $T_h$ is given by
\begin{equation}\label{eq:Zstarhtree}
\varphi^*_h(S) = \frac{p^{|S|}(1-p)^{|E_h\backslash S|}q^{k^*(S)}}{Z^*_h}, \mbox{ where } Z^*_h=\sum_{S\subseteq E_h} p^{|S|}(1-p)^{|E_h\backslash S|}q^{k^*(S)}.
\end{equation} 
Denote by $\rho \conn \infty$ the event that there exists an open path connecting the root $\rho$ to infinity (or, equivalently, that the root is connected via an open path to some leaf).  The following lemma is implicitly proved in \cite{Haggstrom} in the context of the infinite tree, we give an alternative proof following the approach in \cite[Chapter 10]{Grimmettbook} which is carried out for the case $\Delta=3$. (The proof is for the sake of completeness, and the reader might want to skip this.)
\begin{lemma}[{\cite[Theorems 1.5 \& 1.6]{Haggstrom}}]\label{lem:hagguni}
Let $\Delta\geq 3$, $q\geq 1$ and $p\in[0,1]$. Then, in the wired RC distribution on $T_h$, the probability of the event $\rho \conn \infty$ converges as $h$ grows, i.e.,
\[\lim_{h\rightarrow \infty} \varphi_h^*(\rho \conn \infty)=\varphi^*, \mbox{ where } \varphi^*\in [0,1]. \]
For $p_c(q,\Delta)$ as in \eqref{eq:pcq}, it holds that $\varphi^*=0$ if $p< p_c(q,\Delta)$ and $\varphi^*>0$ if $p> p_c(q,\Delta)$.  
\end{lemma}
\begin{remark}
For the sake of completeness, we remark that at criticality, i.e., when $p=p_c(q,\Delta)$, it holds  that $\varphi^*=0$ iff $1\leq q\leq 2$.
\end{remark}
\begin{proof}[Proof of Lemma~\ref{lem:hagguni}]
For convenience, let $d:=\Delta-1$. Let $Z^*_{h,\infty}$ denote the contribution to $Z^*_h$ from $S\subseteq E_h$ such that $\rho$ is connected to infinity and $Z^*_{h,\neg\infty}$ from $S\subseteq E_h$ such that $\rho$ is not connected to infinity. Note that $\varphi^*_h(\rho \conn \infty)=\frac{Z^*_{h,\infty}}{Z^*_{h,\infty}+Z^*_{h,\neg\infty}}$ for all $h\geq 0$. Moreover, with $t:=\frac{p}{q}+1-p$, we have that 
\begin{align*}
Z^*_{h+1,\infty}&=q\sum^{d}_{k=1}\binom{d}{k}\Big(\frac{Z^*_{h,\infty}}{q}\Big)^k\Big(t\frac{Z^*_{h,\neg\infty}}{q}\Big)^{d-k}\big(1-(1-p)^k\big)\\
&=q\Big(\frac{Z^*_{h,\infty}}{q}+t\frac{Z^*_{h,\neg\infty}}{q}\Big)^d-q\Big(\frac{(1-p)Z^*_{h,\infty}}{q}+t\frac{Z^*_{h,\neg\infty}}{q}\Big)^d\\
Z^*_{h+1,\neg\infty}&=q^2\sum^d_{k=0}\binom{d}{k}\Big(\frac{(1-p)Z^*_{h,\infty}}{q}\Big)^k\Big(t\frac{Z^*_{h,\neg\infty}}{q}\Big)^{d-k}=q^2\Big(\frac{(1-p)Z^*_{h,\infty}}{q}+t\frac{Z^*_{h,\neg\infty}}{q}\Big)^d.
\end{align*}
It follows that
\[\varphi^*_{h+1}(\rho \conn \infty)=f\big(\varphi^*_{h}(\rho \conn \infty)\big),\mbox{ where }f(x):=\frac{\Big(t+p(1-\frac{1}{q})x\Big)^d-\Big(t-\frac{p}{q}x\Big)^d}{\Big(t+p(1-\frac{1}{q})x\Big)^d+(q-1)\Big(t-\frac{p}{q}x\Big)^d}.\]
Note that\footnote{This follows from Lemma~\ref{lem:monotonicity}:  the event $\rho\conn\infty$ is increasing and $\varphi^*_{h}$ can be obtained from $\varphi^*_{h+1}$ by conditioning the edges incident to the leaves of $T_{h+1}$ to be open.} $\varphi^*_{h+1}(\rho \conn \infty)\leq \varphi^*_{h}(\rho \conn \infty)$ for all $h\geq 0$ and therefore, as $h$ goes to infinity, $\varphi^*_{h}(\rho \conn \infty)$ converges to a limit $\varphi^*$ which is the largest root in the interval $[0,1]$ of the equation 
\[\varphi^*=f(\varphi^*).\]
Consider the transformation $u=\frac{t + p (1 - 1/q) \varphi^*}{t - \frac{p}{q} \varphi^*}$, whose inverse transformation is given by $\varphi^*=\frac{(p+q(1-p)) (u-1)}{p (u+q-1)}$.  It follows that $u$ is the largest root in the interval $[1,\frac{1}{1-p}]$ of the equation 
\begin{equation}\label{eq:errbthw4}
\frac{\big(p+q(1-p)\big) (u-1)}{p (u+q-1)}=\frac{u^d-1}{u^d+(q-1)} \mbox{ or equivalently } p=1-\frac{1}{1+h(u)},
\end{equation}
where $h(y)=\frac{(y - 1) (y^{d} + q - 1)}{(y^{d} - y)}$ is the same function as in \eqref{eq:pcq}.

We next examine for which values of $p$ it holds that the root $u$ of \eqref{eq:errbthw4} is strictly larger than 1  (note that $\varphi^*>0$ iff $u>1$). Recall from \eqref{eq:pcq} that the value of $p_c(q,\Delta)$ is given by
\begin{equation*}\tag{\ref{eq:pcq}}
p_c(q,\Delta)=1-\frac{1}{1+\inf_{y> 1} h(y)},\mbox{ where } h(y):=\frac{(y - 1) (y^{d} + q - 1)}{(y^{d} - y)}.
\end{equation*}
First, consider the case $p<p_c(q,\Delta)$. For the sake of contradiction, assume that $u>1$. Then,  we obtain from \eqref{eq:errbthw4} that 
\[p=1-\frac{1}{1+h(u)}>1-\frac{1}{1+\inf_{y>1}h(y)}=p_c(q,\Delta),\] 
contradiction. Thus, $\varphi^*=0$ for all $p<p_c(q,\Delta)$. 

Next, consider the case $p>p_c(q,\Delta)$. Using the continuity of the function $h$ in the interval $(1,\infty)$, we obtain that there exists $y>1$ such that $p=1-\frac{1}{1+h(y)}$. Note that 
\[1+h(y)=\frac{1}{1-p}\] 
and we have that $y\leq 1+h(y)$ for all $y>1$, so in fact $y\in (1,\frac{1}{1-p}]$. It follows that $u$,  the largest root in the interval $[1,\frac{1}{1-p}]$ of the equation \eqref{eq:errbthw4}, satisfies $u\geq y>1$, and therefore $\varphi^*>0$.
\end{proof}
We will need the following corollary of Lemma~\ref{lem:hagguni} for a slightly modified tree where the root has degree $\Delta-2$. In particular, for an integer $h\geq 1$, let $\hat{T}_{h}=\big(\hat{V}_h,\hat{E}_h\big)$ be the tree obtained by taking $\Delta-2$ disjoint copies of $T_{h-1}$ and joining their root vertices into a new vertex $\rho$ (for $h=0$, we let $\hat{T}_{h}$ to be the single-vertex graph).  Analogously to \eqref{eq:Zstarhtree}, we use $\hat{L}_h$ to denote the leaves of the tree and $\hat{\varphi}^*_{h}$ to denote the wired measure on $\hat{T}_{h}$ where all the leaves are wired to infinity, i.e.,
\begin{equation}\label{eq:Zstarhtreemod}
\hat{\varphi}^*_h(S) = \frac{p^{|S|}(1-p)^{|\hat{E}_h\backslash S|}q^{k^*(S)}}{\hat{Z}^*_h}, \mbox{ where } \hat{Z}^*_h=\sum_{S\subseteq \hat{E}_h} p^{|S|}(1-p)^{|\hat{E}_h\backslash S|}q^{k^*(S)},
\end{equation} 
and $k^*(S)$ denotes the number of connected components in the graph $\big(\hat{V}_h \cup \{\infty\}, S\cup (\hat{L}_h\times \infty)\big)$.

\begin{corollary}\label{cor:modified}
Let $\Delta\geq 3$, $q\geq 1$ and $p<p_c(q,\Delta)$. Then, in the wired RC distribution on $\hat{T}_h$, the probability of the event $\rho \conn \infty$ converges to 0 as $h$ grows, i.e.,
\[\lim_{h\rightarrow \infty} \hat{\varphi}_h^*(\rho \conn \infty)=0. \]
\end{corollary}  
\begin{proof}
Let $T_h=(V_h,E_h)$ and $\hat{T}_h=(\hat{V}_h,\hat{E}_h)$. Since $\hat{T}_h$ is a subgraph of $T_h$, we may assume that $\hat{V}_h\subseteq V_h$ and  $\hat{E}_h\subseteq E_h$. Observe that 
\[\hat{\varphi}_h^*(\cdot)= \varphi_h^*(\,\cdot\mid E_h\backslash \hat{E}_h \mbox{ closed})\]
Since the event $\rho \conn \infty$ is increasing, we have that
\[ \varphi_h^*(\rho \conn \infty\mid E_h\backslash \hat{E}_h \mbox{ closed})\leq \varphi_h^*(\rho \conn \infty).\]
It follows that $\hat{\varphi}_h^*(\rho \conn \infty)\leq \varphi_h^*(\rho \conn \infty)$. Using Lemma~\ref{lem:hagguni}  for $p<p_c(q,\Delta)$, we obtain the corollary.
\end{proof}
 
\subsection{Analysing RC on disjoint trees whose roots are connected via a path}\label{sec:errr3ded}

We are now in position to prove Lemma~\ref{lem:treespath} which we restate here for convenience. 
\begin{lem:treespath}
\statelemtreespath
\end{lem:treespath}
Let us first give the rough idea of the proof. For simplicity, we will consider the somewhat special case where $\epsilon=0$, but the argument can be easily adapted to account for small positive $\epsilon>0$. In particular, let $H$ be the graph induced by the $h$-graph-neighbourhood of the path $P$, together with the edges of $P$. For $\epsilon=0$, the assumptions of the lemma imply that $H$ is a union of $\ell$ disjoint trees, each\footnote{A technical detail, which is not important for this rough outline, is that we have to be a bit careful with the endpoints of the path, since the degree regularity of $G$ implies that the trees ``hanging'' from the endpoints of the $P$ have to be copies of $(\Delta-1)$-ary tree $T_h$ (where the root has degree $\Delta-1$ instead of $\Delta-2$ which is the case for $\hat{T}_h$).} isomorphic to $\hat{T}_h$,  whose roots are connected by a path.  Using the monotonicity of the RC distribution (cf. Lemma~\ref{lem:monotonicity}), to upper bound the probability that $P$ is open in $\varphi_G$, we can condition all the edges outside $H$ to be open; let $\varphi_H^*$ be the conditioned probability distribution (note the analogy with the wired measure we considered in the previous section).  

Consider first the graph $F$ which is the disjoint union of the $\ell$ trees (i.e., $F$ is obtained from $H$ by removing the edges of the path $P$) and let $\varphi_F^*$ be the analogue of $\varphi_H^*$ (i.e., the RC distribution on $G\backslash P$ where all edges outside $F$ are assumed to be open).    Crucially, since $p<p_c(q,\Delta)$, we have by Lemma~\ref{lem:hagguni} that only $\epsilon '\ell$  roots are connected via an open path to infinity, where $\epsilon'>0$ is a constant that can be made arbitrarily small by taking the depth $h$ of the trees sufficiently large.  Denote by $R$ this (random) set of root vertices, so that $|R|\leq \epsilon' \ell$ with high probability. 

Now, we add the edges of the path $P$ and consider how this reweights the random-cluster configuration. In particular, we focus on edges  of the path which are not incident to a root in $R$, let $e$ be such an edge. If $e$ is open, we get a factor of $p/q$ in the weight of the random-cluster configuration ($p$ because $e$ is open and $1/q$ because the total number of connected components decreases by 1); if $e$ is closed, we get a factor of $1-p$ in the weight of the random-cluster configuration (because $e$ is closed -- note that the number of connected components stays the same in this case). Since $|R|\leq \epsilon' \ell$, there are at least $\ell-2\epsilon' \ell$ edges of the path $P$ that are not incident to a vertex in $R$. Therefore the probability that all of them are open is roughly $\tau^{\ell(1-2\epsilon)}$ where $\tau:=\frac{p/q}{p/q+(1-p)}$. For $p<p_c(q,\Delta)$ it holds that $\tau<1/(\Delta-1)$ and hence the lemma  follows by taking $\epsilon'>0$ to be a sufficiently small constant (to account also for the roughly $2^{\epsilon' \ell}$ choices of the set $R$).

\begin{proof}[Proof of Lemma~\ref{lem:treespath}]
For convenience, let $d:=\Delta-1$. We begin first by specifying the constant $K$ and how large $\ell$ and $h$ need to be.

Let $\tau:=\frac{p}{p+q(1-p)}$ and note that for all $p<p_c(q,\Delta)$ we have that $\tau<1/d$ (since $p_c(q,\Delta)<\frac{q}{q+d-1}$). Let $K$ be any constant satisfying $\tau<K<1/d$ and $\epsilon>0$ be a small constant such that for $\epsilon':=10q^3\epsilon$ it holds that $\tau^{1-\epsilon'}q^{10\epsilon'}<K$ (note that such an $\epsilon$ exists by considering the limit $\epsilon\downarrow 0$). Let $\ell$ be sufficiently large so that $\epsilon' \ell\geq \epsilon \ell\geq  2$.

As in Section~\ref{sec:unique}, for an integer $h\geq 0$, let $\hat{T}_{h}=(\hat{V}_h,\hat{E}_h)$ denote the subtree of the $(\Delta-1)$-ary tree with height $h$ where the root $\rho$ has degree $\Delta-2$ (and every other non-leaf vertex has degree $\Delta$). Let $\hat{Z}^*_h$ denote the partition function of the random-cluster model where all the leaves are connected to infinity (cf. \eqref{eq:Zstarhtreemod}). Let $\hat{Z}^*_{h,\infty}$ be the contribution to $\hat{Z}^*_h$ from $S\subseteq \hat{E}_h$ such that $\rho$ is connected to infinity and $\hat{Z}^*_{h,\neg\infty}$ from $S\subseteq E_h$ such that $\rho$ is not connected to infinity. Since $p<p_c(q,\Delta)$, we have by Corollary~\ref{cor:modified} that for all sufficiently large $h$ it holds that 
\begin{equation}\label{eq:Zhtreeee12}
\hat{Z}^*_{h,\infty}/\hat{Z}^*_h\leq \epsilon.
\end{equation}

We are now ready to proceed to the proof of the lemma. Consider a path $P$ with vertices $u_1,\hdots,u_\ell$ whose $h$-graph neighbourhood contains at least $(1-\epsilon)\ell$ isolated tree components.  By definition, an isolated tree component contains exactly one of the vertices $u_1,\hdots,u_\ell$ and therefore
the set
\[U=\big\{u_i \,\big|\, \mbox{$u_i$ belongs to an isolated tree component and $u_i\neq u_1,u_\ell$}\big\},\]
satisfies $|U|\geq (1-\epsilon)\ell-2\geq (1-2\epsilon)\ell$. Observe also that 
\[\mbox{for $u\in U$, $\Gamma_h(G\backslash P, u)$ induces a subgraph in $G$ which is isomorphic to $\hat{T}_h$}.\]
Indeed, for $u\in U$, denote by $\mathcal{C}_u$ the component of the $h$-graph-neighbourhood of the path $P$ that the vertex $u$ belongs to. Since $u\in U$, $\mathcal{C}_u$ is isolated and therefore the vertex set of $\mathcal{C}_u$ is precisely $\Gamma_h(G\backslash P, u)$ and $\mathcal{C}_u$ is an induced subgraph of $G$. It remains to observe that $\mathcal{C}_u$ is a tree (since $u\in U$),  every non-leaf vertex in $\mathcal{C}_u$ different than $u$ has degree $\Delta$ (since the graph $G$ is $\Delta$-regular) and $u$ itself has degree $\Delta-2$ (since $u\neq u_1,u_\ell$). 

Let $F$ be the subgraph of $G$  induced by the vertex set  $\bigcup_{u\in U}\Gamma_h(G\backslash P,u)$. Note that $F$ is a disjoint union of copies of $\hat{T}_h$. For convenience, denote by $V_F,E_F$ the vertex and edge set of $F$ and similarly denote $V_P, E_P$ for the corresponding sets of the path $P$. Note that $E_F$ is disjoint from $E_P$ (though $V_P$ and $V_F$ intersect at the roots of the trees). By Lemma~\ref{lem:monotonicity}, we have that
\[\varphi_G(\mbox{$P$ is open})\leq \varphi_G(\mbox{$P$ is open}\mid \mbox{ all edges in $E\backslash (E_P\cup E_F)$ are open}).\]
The conditioning in the r.h.s. has the same effect as ``wiring''. In particular,  let $H=(V_H,E_H)$ be the subgraph  of $G$ with vertex set $V_P\cup V_F$ and edge set $E_P\cup E_F$. Note that $H$ is a path with (disjoint) copies of the tree $\hat{T}_h$ ``hanging'' from most vertices of the path (in particular, the vertices in $U$). Let $L$ be the set of leaves of all the  trees and let $W=(V_P\backslash U)\cup L$. Let $\varphi_H^*$ be the RC distribution on $H$, where we wire all vertices in $W$ to infinity, i.e., $\varphi^*_H$ is given by
\begin{equation}\label{eq:ZstarhtreeHH}
\varphi^*_H(S) = \frac{p^{|S|}(1-p)^{|E_H\backslash S|}q^{k^*(S)}}{Z^*_H}, \mbox{ where } Z^*_H:=\sum_{S\subseteq E_H} p^{|S|}(1-p)^{|E_H\backslash S|}q^{k^*(S)},
\end{equation}
where $k^*(S)$ is the number of connected components in the graph with vertex set $V_H\cup \{\infty\}$ and edge set $S\cup (W\times\infty)$. With this definition, we have that
\[\varphi_G\big(\mbox{$P$ is open}\mid \mbox{ all edges in $E\backslash (E_F\cup E_P)$ are open}\big)=\varphi^*_H(\mbox{$P$ is open}).\]
Let $Q$ be the contribution to $Z_H^*$ from configurations $S\subseteq E_H$ such that the path $P$ is open, i.e.,
\[Q:=\sum_{S\subseteq E_H;\, E_P\subseteq S}p^{|S|}(1-p)^{|E_H\backslash S|}q^{k^*(S)}.\]
We will show that for all sufficiently large $h$ (so that \eqref{eq:Zhtreeee12} holds) and all sufficiently large $\ell$ it holds that
\begin{equation}\label{eq:boundspart123}
Z_H^*\geq \Big(\frac{p}{q}+(1-p)\Big)^{\ell-1} \big(\hat{Z}_h/q\big)^{|U|}, \qquad Q\leq \Big(\frac{p}{q}\Big)^{\ell-1}q^{10\epsilon' \ell}\big(\hat{Z}_h/q\big)^{|U|}.
\end{equation}
Assuming \eqref{eq:boundspart123} for the moment, note that
\[\varphi^*(\mbox{$P$ is open})=\frac{Q}{Z_H^*}\leq \tau^{\ell-1}q^{10\epsilon'\ell}\leq K^\ell,\]
where in the last inequality we used that for all sufficiently large $\ell$, we have $\tau^{\ell-1}q^{10\epsilon'\ell}\leq (\tau^{1-\epsilon'}q^{10\epsilon'})^{\ell}\leq K^{\ell}$. This proves the lemma. We therefore focus on proving \eqref{eq:boundspart123}.

To prove the first inequality in \eqref{eq:boundspart123}, note that for any $S\subseteq E_H$,  we have 
\begin{equation}\label{eq:efver21}
\begin{gathered}
|S|=|S\cap E_F|+|S\cap E_P|, \quad |E_H\backslash S|=|E_F\backslash S|+(\ell-1-|S\cap E_P|)\\
k^*(S)+|S\cap E_P|\geq k^*(S\cap E_F).
\end{gathered} 
\end{equation}
The equalities follow from the fact that $\{E_F,E_P\}$ is a partition of $E_H$, while the inequality follows by noting that deleting an edge in $S\cap E_P$ increases the number of connected components by at most 1. Using \eqref{eq:efver21} and $q\geq 1$, we therefore obtain that (with $m=|S\cap E_P|$ and $S'=S\cap E_F$) 
\begin{align}
Z_H^*&\geq \sum_{S\subseteq E_H}p^{|S|}(1-p)^{|E_H\backslash S|}q^{k^*(S\cap E_F)-|S\cap E_P|}\notag\\
&=\sum^{\ell-1}_{m=0}\binom{\ell-1}{m}\Big(\frac{p}{q}\Big)^m(1-p)^{\ell-1-m}\sum_{S'\subseteq E_F}p^{|S'|}(1-p)^{|E_F\backslash S'|}q^{k^*(S')}.\label{eq:v5g3v3rv}
\end{align}
In the following, we focus on lower bounding the inner sum in \eqref{eq:v5g3v3rv}. Recall that the graph $F$ consists of $|U|$ disjoint copies of the tree $\hat{T}_h$, each rooted at a vertex in $U$. For an integer $r=0,\hdots,|U|$, let
\[\Omega_F(r)=\Big\{S'\subseteq E_F\, \Big|\, \begin{array}{c}\mbox{exactly $r$ vertices in $U$ are connected to $\infty$ in the graph }\\ \mbox{ with vertex set $V_F\cup \{\infty\}$ and edge set $S'\cup (L\times \infty)$}\end{array}\Big\}.\] 
Then, we have that
\begin{equation}\label{eq:SpSEp2432}
q^{|U|-r}\sum_{S'\in \Omega_F(r)}p^{|S'|}(1-p)^{|E_F\backslash S'|}q^{k^*(S')}=\binom{|U|}{r}\big(\hat{Z}_{h,\infty}\big)^r\big(\hat{Z}_{h,\neg\infty}\big)^{|U|-r}.
\end{equation}
(Note that the r.h.s. in \eqref{eq:SpSEp2432} counts $|U|-r+1$ times the component containing $\infty$, while $k^*(S')$ only once  for all $S'\subseteq E_F$, which explains the need for the factor $q^{|U|-r}$ on the l.h.s.). We therefore have that (using $q\geq 1$)
\begin{equation}\label{eq:vtvtgew21}
\sum_{S'\subseteq E_F}p^{|S'|}(1-p)^{|E_F\backslash S'|}q^{k^*(S')}=\frac{1}{q^{|U|}}\sum^{|U|}_{r=0}q^{r}\binom{|U|}{r}\big(\hat{Z}_{h,\infty}\big)^r\big(\hat{Z}_{h,\neg\infty}\big)^{|U|-r}\geq \Big(\frac{\hat{Z}_h}{q}\Big)^{|U|}.
\end{equation}
Plugging \eqref{eq:vtvtgew21} into \eqref{eq:v5g3v3rv} and using the binomial expansion yields the first inequality in \eqref{eq:boundspart123} as wanted.

To prove the second inequality in \eqref{eq:boundspart123}, consider as before the set of configurations $\Omega_F(r)$ where exactly $r$ roots of the trees are connected to infinity and let $\Omega_H(r)=\{S'\cup E_P \mid S'\in \Omega_F(r)\}$. For $S\in \Omega_H(r)$, we will show
\begin{equation}\label{eq:f34f32}
\begin{gathered}
|S|=|S\cap E_F|+(\ell-1), \quad |E_H\backslash S|=|E_F\backslash S|,\\
k^*(S)+2(|U|-r)-(\ell-1)\leq k^*(S\cap E_F),
\end{gathered}
\end{equation}
The equalities are an immediate consequence of the equalities in \eqref{eq:efver21} and the fact that, by the definition of $\Omega_H(r)$, we have that $E_P\subseteq S$ for all $S\in \Omega_H(r)$. To justify the inequality, note that there are exactly $M=r+\ell-|U|$ vertices of the path $P$ which are connected to infinity. It follows that there are at least $\ell-1-2M=2(|U|-r)-(\ell+1)$ edges in $S\cap E_P=E_P$ whose endpoints are not connected to infinity; deleting any of these edges causes the number of components to increase by one. Using  \eqref{eq:f34f32} and the fact that $q\geq 1$,  we can bound $Q$ by 
\begin{equation}\label{eq:Wlowerbound241}
\begin{aligned}
Q&\leq \sum^{|U|}_{r=0}\sum_{S\in \Omega_H(r)}p^{|S|}(1-p)^{|E_H\backslash S|}q^{k^*(S\cap E_F)+(\ell+1)-2(|U|-r)}= \Big(\frac{p}{q}\Big)^{\ell-1}\sum^{|U|}_{r=0}A_r
\end{aligned}
\end{equation}
where 
\[A_r:=\sum_{S'\in \Omega_F(r)}p^{|S'|}(1-p)^{|E_F\backslash S'|}q^{k^*(S')+2(\ell-|U|-r)}.\]
Using \eqref{eq:SpSEp2432}  again, we obtain that
\[\frac{A_r}{(\hat{Z}_h/q)^{|U|}}=q^{2(l-|U|)+3r}\binom{|U|}{r}\Big(\frac{\hat{Z}_{h,\infty}}{\hat{Z}_h}\Big)^r\Big(\frac{\hat{Z}_{h,\neg\infty}}{\hat{Z}_h}\Big)^{|U|-r}.\]
Recall that $\epsilon'=10q^3\epsilon$; using \eqref{eq:Zhtreeee12}, note that for all $r\geq \epsilon'\ell$ we have
\[\binom{|U|}{r}\Big(\frac{\hat{Z}_{h,\infty}}{\hat{Z}_h}\Big)^r\Big(\frac{\hat{Z}_{h,\neg\infty}}{\hat{Z}_h}\Big)^{\ell-r}\leq \Big(\frac{\emm|U|}{r}\Big)^{r}\epsilon^r\leq \Big(\frac{\emm\ell}{r}\Big)^{r}\epsilon^r\leq \frac{1}{(4q^3)^r}.\]
We also have that $|U|\geq (1-2\epsilon)\ell$, so $\ell-|U|\leq 2\epsilon \ell$. Therefore, for all sufficiently large $\ell$ we have the bound
\[\frac{1}{(\hat{Z}_h/q)^{|U|}}\sum^{|U|}_{r=0}A_r\leq q^{4\epsilon \ell}\sum^{|U|}_{r=0}q^{3r}\binom{|U|}{r}\Big(\frac{\hat{Z}_{h,\infty}}{\hat{Z}_h}\Big)^r\Big(\frac{\hat{Z}_{h,\neg\infty}}{\hat{Z}_h}\Big)^{\ell-r}\leq q^{4\epsilon \ell}\Big( 1+\sum^{\left\lfloor \epsilon'\ell\right\rfloor}_{r=0}q^{3\epsilon'\ell}\Big)\leq q^{10\epsilon' \ell}.\]
Plugging this into \eqref{eq:Wlowerbound241} yields the second inequality in \eqref{eq:boundspart123}, as needed.
\end{proof}

\section{Analysing the Potts model on graphs with tree-like structure}\label{sec:teb3tbb6b5b5}
The goal of this section is to prove Lemma~\ref{lem:treespathPotts}.
\subsection{Analysing Potts on trees}
 Fix an integer $\Delta\geq 3$. As in Section~\ref{sec:unique}, we use $\TreeD$ denote the infinite $(\Delta-1)$-ary tree with root vertex $\rho$. For integer $h\geq 0$, let $T_h=(V_h,E_h)$ denote the subtree of $\TreeD$ induced by the vertices at distance $\leq h$ from $\rho$ and let $L_h$ denote the leaves of $T_h$. 

Recall that uniqueness on $\TreeD$ implies that root-to-leaves correlations on $T_h$ tend to 0 as $h\rightarrow \infty$, cf. Definition~\ref{def:uniqueness}.  The following lemma extends these decay properties to arbitrary subtrees of $T_h$; this was proved in \cite{Efthymiou} in the case of the colourings model and the proof for the Potts model is analogous.

\begin{lemma}\label{lem:subtreefree}
Let $\Delta,q\geq 3$ be integers, and $B>0$ be in the uniqueness regime of the $(\Delta-1)$-ary tree. There exists a function $\vartheta:\mathbb{N}\rightarrow \mathbb{R}_{\geq 0}$ with $\vartheta(h)\rightarrow 0$ as $h\rightarrow \infty$ such that the following holds for any  integer $h\geq 0$.

Let $T_h'$ be an arbitrary subtree of $T_h$ containing the root $\rho$ and let $L_h'$ be the set of vertices in $T_h'$ at distance exactly $h$ from $\rho$. Then, for any colour $c\in [q]$ and any configuration $\tau:L_h'\rightarrow [q]$, it holds that 
\[\Big|\mu_{T_h'}(\sigma_\rho=c\mid \sigma_{L_h'}=\tau)-\frac{1}{q}\Big|\leq \vartheta(h).\]
\end{lemma}
\begin{proof}
We will show that the statement of the lemma holds with the function $\vartheta(\cdot)$ given by $\vartheta(0)=1$ and 
\[\vartheta(h):=\max_{c\in[q],\ \tau: L_h \to [q]}\Big|\mu_{T_h}[\sigma(\rho) = c\mid\sigma_{L_h}=\tau]-\frac{1}{q}\Big|.\]
Note that, since $B$ is assumed to be in the uniqueness regime of the $(\Delta-1)$-ary tree, we have by definition that $\vartheta(h)\rightarrow 0$ as $h\rightarrow \infty$.  

We first show by induction on $h$ that, for all $h\geq 0$, for all subtrees $T_h'$ of $T_h$ containing the root $\rho$, for all configurations $\tau:L_h'\rightarrow [q]$ and any colour $c\in[q]$, it holds that
\begin{equation}\label{eq:trunc2}
\mu_{T_h'}(\sigma_\rho=c\mid \sigma_{L_h'}=\tau)=\mu_{T_h}(\sigma_\rho=c\mid \sigma_{L_h'}=\tau).
\end{equation}
For $h=0$ the result is trivial. Suppose that $h\geq 1$ and that the result holds for all integers less than $h$, we prove the result for $h$ as well. Let $U=\{u_1,\hdots, u_{\Delta-1}\}$ be the children of $\rho$ in $T_h$ and let $U'\subseteq U$ be the children of $\rho$ in $T_h'$. For a vertex $u\in U$, denote by $T_h(u)$ the subtree of $T_h$ rooted at $u$ and by $L_h(u)$ the vertices in $L_h$ that belong to  the tree $T_h(u)$. Similarly, for a vertex $u\in U'$, denote by $T_h'(u)$ the subtree of $T_h'$ rooted at $u$ and by $L_h'(u)$ the vertices in $L_h'$ that belong to  the tree $T_h'(u)$. Using standard tree recursions (see for example \cite[Lemma 19]{GGY}), we have that 
\begin{equation}\label{eq:3f3tf3tf112}
\mu_{T_h}(\sigma_\rho = c\mid \sigma_{L_h'}=\tau) = 
    \frac{\prod_{u\in U} \big(1-(1-B)\,
        \mu_{T_h(u)}(\sigma_u=c\mid \sigma_{L_h'(u)}=\tau_{L_h'(u)})\big)}
    {\sum_{c'=1}^q\prod_{u\in U} \big(1-(1-B)\,
        \mu_{T_h(u)}(\sigma_u=c'\mid \sigma_{L_h'(u)}=\tau_{L_h'(u)})\big)}
\end{equation}
    and
\begin{equation}\label{eq:3f3tf3tf112a}
 \mu_{T_h'}(\sigma_\rho = c\mid \sigma_{L_h'}=\tau) = 
    \frac{\prod_{u\in U'} \big(1-(1-B)\,
        \mu_{T_h'(u)}(\sigma_u=c\mid \sigma_{L_h'(u)}=\tau_{L_h'(u)})\big)}
    {\sum_{c'=1}^q\prod_{u\in U'} \big(1-(1-B)\,
        \mu_{T_h'(u)}(\sigma_u=c'\mid \sigma_{L_h'(u)}=\tau_{L_h'(u)})\big)}
\end{equation}
For every $u\in U\backslash U'$ we have that  $L_h'(u') = \emptyset$ and therefore, using the symmetry among the colours, we have that for every $c'\in [q]$ it holds that
\begin{equation}\label{eq:3f3tf3tf112b}
\mu_{T_h(u)}(\sigma_u=c'\mid \sigma_{L_h'(u)}=\tau_{L_h'(u)})=\mu_{T_h(u)}(\sigma_u=c')=1/q
\end{equation}
For $u\in U'$, we have that $T_h(u)$ is isomorphic to $T_{h-1}$ and $T_h'(u)$ is a subtree of $T_h(u)$, so by the induction hypothesis we have that for any colour $c\in [q]$ it holds that
\begin{equation}\label{eq:3f3tf3tf112c}
\mu_{T_h'(u)}(\sigma_u=c'\mid \sigma_{L_h'(u)}=\tau_{L_h'(u)})=\mu_{T_h(u)}(\sigma_u=c'\mid \sigma_{L_h(u)}=\tau_{L_h(u)}).
\end{equation}
Combining \eqref{eq:3f3tf3tf112}, \eqref{eq:3f3tf3tf112a}, \eqref{eq:3f3tf3tf112b} and \eqref{eq:3f3tf3tf112c} yields \eqref{eq:trunc2}, thus completing the induction.

To complete the proof it remains to observe that for any colour $c\in[q]$, it holds that   
    \begin{equation*}
\mu_{T_h}(\sigma_\rho=c\mid \sigma_{L_h'}=\tau)=\sum_{\eta: L_h\backslash L_h'\rightarrow [q]}\mu_{T_h}(\sigma_\rho=c\mid \sigma_{L_h'}=\tau, \sigma_{L_h\backslash L_h'}=\eta ) \times  \mu_{T_h}(\sigma_{L_h\backslash L_h'}=\eta\mid \sigma_{L_h'}=\tau)
    \end{equation*}
Note that, by the definition of the function $\vartheta(\cdot)$, for any $\eta:L_h\backslash L_h'\rightarrow [q]$ we have that 
\[\Big|\mu_{T_h}(\sigma_\rho=c\mid \sigma_{L_h'}=\tau, \sigma_{L_h\backslash L_h'}=\eta )-\frac{1}{q}\Big|\leq \vartheta(h)\]
and $\sum_{\eta: L_h\backslash L_h'\rightarrow [q]}\mu_{T_h}(\sigma_{L_h\backslash L_h'}=\eta\mid \sigma_{L_h'}=\tau)=1$, which gives that
\[\Big|\mu_{T_h}(\sigma_\rho=c\mid \sigma_{L_h'}=\tau)-\frac{1}{q}\Big|\leq \vartheta(h),\]
Combining this with \eqref{eq:trunc2} yields the lemma.
\end{proof}

\subsection{Analysing Potts on disjoint trees whose roots are connected via a path}\label{sec:treespathPotts}
To prove Lemma~\ref{lem:treespathPotts}, we will need the following lemmas.
\begin{lemma}\label{lem:path2endscond}
Let $q\geq 3$ and $B\in (0,1)$, and set $\chi:=\frac{1+B}{B+q-1}$. Let $P$ be a path with $\ell\geq 2$ vertices and endpoints $u,v$. Then, for arbitrary colours  $c,c'\in [q]$, it holds that
\[\mu_P(\mbox{$P$ is bichromatic} \mid \sigma_u=c,\sigma_v=c')\leq (4q/B) \chi^{\ell-2}.\] 
\end{lemma}
\begin{proof}
Let $z$ be the neighbour of $v$ (note that if $\ell=2$ then $z=u$). We have that
\begin{equation}\label{eq:v443sfs24223fwed}
\mu_P(\sigma_v=c'\mid \sigma_u=c)=\sum_{c''\in [q]}\mu_P(\sigma_v=c'\mid \sigma_z=c'',\sigma_u=c)\times \mu_P(\sigma_z=c''\mid \sigma_u=c).
\end{equation}
Note also that for arbitrary colours $c,c',c''\in [q]$ we have that
\[\mu_P(\sigma_v=c'\mid \sigma_z=c'', \sigma_u=c)=\mu_P(\sigma_v=c'\mid \sigma_z=c'')\geq\frac{B^{\mathbf{1}\{c'=c''\}}}{B+q-1}\geq B/(B+q-1).\] 
Combining this with \eqref{eq:v443sfs24223fwed}, we obtain that $\mu_P(\sigma_v=c'\mid \sigma_u=c)\geq B/(B+q-1)$ and hence
\begin{equation}\label{eq:brvt453vgr}
\mu_P(\sigma_u=c, \sigma_v=c')=\frac{1}{q} \mu_P(\sigma_v=c'\mid \sigma_u=c)\geq \frac{B}{q(B+q-1)}.
\end{equation}

 Let $c_1,c_2\in [q]$ be distinct colours in $[q]$ and $\mathcal{E}(c_1,c_2)$ be the event that every vertex in $P$ is coloured with $c_1,c_2$. The lemma will follow by showing that 
\begin{equation}\label{eq:vteve42gvtv24vt}
\mu_P(\mathcal{E}(c_1,c_2))=2\chi^{\ell-1}/q.
\end{equation}
Let us briefly conclude the lemma assuming \eqref{eq:vteve42gvtv24vt}. Indeed, if $c\neq c'$, then applying \eqref{eq:vteve42gvtv24vt} for $c_1=c$ and $c_2=c'$ and  using the lower bound in \eqref{eq:brvt453vgr} we obtain that
\[\mu_P(\mbox{$P$ is bichromatic} \mid \sigma_u=c,\sigma_v=c')\leq (4/B)\chi^{\ell-2},\]
while, if $c=c'$, we obtain by summing \eqref{eq:vteve42gvtv24vt} for $c_1=c$ and the $q-1$ possible values of $c_2$  that 
\[\mu_P(\mbox{$P$ is bichromatic} \mid \sigma_u=c,\sigma_v=c')\leq (4q/B)\chi^{\ell-2}.\]
It remains to prove \eqref{eq:vteve42gvtv24vt}. For convenience, denote by $w_1,\hdots,w_\ell$ the vertices of $P$ in order so that $u=w_1$ and $v=w_\ell$ and let $C=\{c_1,c_2\}$. Note that 
\[\mu_P(\mathcal{E}(c_1,c_2))=\mu_P(\sigma_{w_1}\in C)\prod^\ell_{i=2}\mu_{P}\big(\sigma_{w_i}\in C \mid \sigma_{w_{i-1}}\in C, \hdots, \sigma_{w_1}\in C\big).\]
We have $\mu_P(\sigma_{w_1}\in C)=2/q$. For $i=2,\hdots, n$,   let $P_i$ be the path induced by the vertices $w_{i-1},\hdots,w_n$. We have that
\begin{equation*}
\mu_{P}\big(\sigma_{w_i}\in C \mid \sigma_{w_{i-1}}\in C, \hdots, \sigma_{w_1}\in C\big)=\mu_{P_i}\big(\sigma_{w_i}\in C \mid \sigma_{w_{i-1}}\in C\big).
\end{equation*}
Since $\mu_{P_i}(\sigma_{w_i}\in C)=\mu_{P_i}(\sigma_{w_{i-1}}\in C)=2/q$, we have by the Bayes' rule that
\[\mu_{P_i}(\sigma_{w_i}\in C\mid \sigma_{w_{i-1}}\in C)=\mu_{P_i}(\sigma_{w_{i-1}}\in C \mid \sigma_{w_{i}}\in C)=\frac{1+B}{B+q-1}=\chi.\]
Combining these, we obtain \eqref{eq:vteve42gvtv24vt}, therefore concluding the proof of Lemma~\ref{lem:path2endscond}.
\end{proof}
We will use the following corollary of Lemma~\ref{lem:path2endscond}.

\begin{corollary}\label{cor:path2endscond}
Let $q\geq 3$ and $B\in (0,1)$, and set $\chi:=\frac{1+B}{B+q-1}$. Let $P$ be a path with $\ell\geq 2$ vertices and let $\Lambda$ be a subset of the vertices which includes the endpoints of the path. Then, for any configuration $\tau:\Lambda\rightarrow [q]$,  it holds that
\[\mu_P(\mbox{$P$ is bichromatic} \mid \sigma_\Lambda=\tau)\leq (4q/B)^{|\Lambda|}\chi^{\ell-2|\Lambda|}.\] 
\end{corollary}
\begin{proof}
Let $t:=|\Lambda|$ and denote the set of vertices in $\Lambda$ by $u_1, \hdots, u_t$ in the order that they appear in the path. Note that $u_1$ and $u_t$ are the endpoints of $P$ (since by assumption $\Lambda$ includes the endpoints of $P$). For $i=1,2,\hdots,t-1$, let $P_{i}$ be the path induced by the vertices between $u_i$ and $u_{i+1}$ and let $\ell_i$ be the number of vertices in $P_i$. Then, we have
\begin{align*}
\mu_P(\mbox{$P$ is bichromatic} \mid \sigma_\Lambda=\tau)&\leq \mu_P(\mbox{$P_i$ is bichromatic for $i=1,\hdots, t-1$} \mid \sigma_\Lambda=\tau)\\
&= \prod^{t-1}_{i=1}\mu_{P_i}(\mbox{$P_i$ is bichromatic} \mid \sigma_{u_i}=\tau_{u_i},\sigma_{u_{i+1}}=\tau_{u_{i+1}}).
\end{align*}
By Lemma~\ref{lem:path2endscond}, we have that 
\[\mu_{P_i}(\mbox{$P_i$ is bichromatic} \mid \sigma_{u_i}=\tau_{u_i},\sigma_{u_{i+1}}=\tau_{u_{i+1}})\leq (4q/B)^{}\chi^{\ell_i-2}\]
and therefore, since $\sum^{t}_{i} \ell_i=\ell$ and $\chi,B\in (0,1)$, we obtain that 
\[\mu_P(\mbox{$P$ is bichromatic} \mid \sigma_\Lambda=\tau)\leq (4q/B)^{|\Lambda|-1}\chi^{\ell-2(|\Lambda|-1)}\leq (4q/B)^{|\Lambda|}\chi^{\ell-2|\Lambda|}.\]
This finishes the proof.
\end{proof}

We are now ready to prove Lemma~\ref{lem:treespathPotts}, which we restate here for convenience.
\begin{lem:treespathPotts}
\statelemtreespathPotts
\end{lem:treespathPotts}
\begin{proof}
Let $\chi:=\frac{1+B}{B+q-1}$ and consider arbitrary $\epsilon'>0$. We begin  by specifying the constants $K,\epsilon$ and how large $\ell$ and $h$ need to be. In particular, let $K$ be any constant satisfying $\chi<K<\chi+\epsilon'$ and $\epsilon>0$ be a small constant such that
\begin{equation}\label{eq:eps1selection}
(4q/\chi B)^{2\epsilon} \Big(\frac{1/q+\epsilon}{1/q-\epsilon}\chi\Big)^{1-2\epsilon}\leq K/(1+\epsilon).
\end{equation}
Note that such an $\epsilon$ exists by considering the limit $\epsilon\downarrow 0$. Let $\ell$ be sufficiently large so that 
\begin{equation}\label{eq:ell2selection}
\epsilon \ell\geq  2 \quad \mbox{and}\quad (1+\epsilon)^{\ell}\geq q^2.
\end{equation}
Finally, let $\vartheta(\cdot)$ be the function in Lemma~\ref{lem:subtreefree}, so that for all sufficiently large $h$ it holds that
\begin{equation}\label{eq:thetahsmall}
\vartheta(h)\leq \epsilon.
\end{equation}

We are now ready to proceed to the proof of the lemma. Let $G$ be a graph and consider a path $P$ in $G$, with vertices $u_1,\hdots,u_\ell$, whose $h$-graph neighbourhood contains at least $(1-\epsilon)\ell$ isolated tree components.  By definition, an isolated tree component contains exactly one of the vertices $u_1,\hdots,u_\ell$ and therefore
the set
\[U=\big\{u_i \,\big|\, \mbox{$u_i$ belongs to an isolated tree component and $u_i\neq u_1,u_\ell$}\big\},\]
satisfies $|U|\geq (1-\epsilon)\ell-2\geq (1-2\epsilon)\ell$. Note that we exclude the endpoints of the path $P$ from $U$, even if they belong to isolated tree components. Let also $\overline{U}$ denote the set $\{u_1,\hdots,u_\ell\}\backslash U$, so that $|\overline{U}|\leq 2\epsilon \ell$.  
Recall that $T_h$ is the subtree of the $(\Delta-1)$-ary tree consisting of vertices at distance at most $h$ from the root. By definition of the set $U$, we have that 
\begin{equation}\label{eq:v4t3vecrct46g6}
\mbox{for $u\in U$, $\Gamma_h(G\backslash P, u)$ induces a subgraph in $G$ which is a subtree of $T_h$}.
\end{equation}
For $u\in U$, let for convenience $V_u=\Gamma_h(G\backslash P,u)$, $F_u$ be the subgraph of $G$ induced on $V_u$ and  $L_u$   be the set of vertices that are at distance $h$ from $u$ in $F_u$.  Note that $F_u$ is a tree rooted at $u$ all of whose vertices are at distance at most $h$ from $u$; $L_u$ is thus the set of all leaves in $F_u$ which are at distance $h$ from $u$ (note that there could be other leaves which are closer to $u$ but these will not matter). 

Let $F$ be the union of the subgraphs $F_u$ for $u\in U$; note that $F$ is a disjoint union of copies of trees (each of which is a subtree of $T_h$ from \eqref{eq:v4t3vecrct46g6}). Let also $L:=\bigcup_{u\in U}L_u$. We will also denote by $V_P, E_P$ the vertex and edge set of the path $P$.  Note that $E_F$ is disjoint from $E_P$ (though $V_P$ and $V_F$ intersect at $V_{P}$). Finally, let $H=(V_H,E_H)$ be the subgraph  of $G$ with vertex set $V_P\cup V_F$ and edge set $E_P\cup E_F$. 

To bound the probability that $P$ is bichromatic, we will condition on a worst case boundary configuration on $L$ and $\overline{U}$. In particular, we have the bound
\[\mu_G(\mbox{$P$ is bichromatic})\leq \max_{\tau:L\rightarrow [q], \tau':\overline{U}\rightarrow [q]}\mu_G(\mbox{$P$ is bichromatic}\mid \sigma_L=\tau,\, \sigma_{\overline{U}}=\tau').\]
To bound the r.h.s., fix arbitrary configurations $\tau:L\rightarrow [q], \tau':\overline{U}\rightarrow [q]$ and note that
\[\mu_G(\mbox{$P$ is bichromatic}\mid \sigma_L=\tau,\, \sigma_{\overline{U}}=\tau')=\mu_H(\mbox{$P$ is bichromatic}\mid \sigma_L=\tau,\, \sigma_{\overline{U}}=\tau'),\]
so the lemma will follow (since $\tau,\tau'$ are arbitrary) by showing that, for arbitrary colours $c_1,c_2\in [q]$, it holds that
\begin{equation}\label{eq:3vtv3t565er2f554f}
\mu_H\big(\mathcal{E}(c_1,c_2)\mid \sigma_L=\tau,\, \sigma_{\overline{U}}=\tau'\big)\leq K^\ell/q^2,
\end{equation}
where $\mathcal{E}(c_1,c_2)$ is the event that each vertex in $P$ is coloured with either $c_1$ or $c_2$. Since $\overline{U}$ includes the endpoints of the path $P$, by Corollary~\ref{cor:path2endscond} we have that
\begin{equation}\label{eq:t4vtyy3563r}
\mu_P(\mathcal{E}(c_1,c_2)\mid  \sigma_{\overline{U}}=\tau')\leq (4q/B)^{|\overline{U}|}\chi^{\ell-2|\overline{U}|}=(4q/\chi B)^{|\overline{U}|} \chi^{|U|}.
\end{equation}
Let $Z_P$ be the partition function of the path $P$ and $Z_P(c_1,c_2)$ be the contribution to $Z_P$ from configurations such that $P$ is coloured with $c_1$ or $c_2$, i.e., 
\[Z_{P}=\sum_{\substack{\sigma: V_{P}\rightarrow [q];\ \sigma_{\overline{U}}=\tau'}}w_{P}(\sigma), \qquad Z_{P}(c_1,c_2)=\sum_{\substack{\sigma: V_{P}\rightarrow \{c_1,c_2\};\ \sigma_{\overline{U}}=\tau'}}w_{P}(\sigma),\]
Then, \eqref{eq:t4vtyy3563r} translates into
\begin{equation*}
Z_{P}(c_1,c_2)/Z_P\leq (4q/\chi B)^{|\overline{U}|} \chi^{|U|}.
\end{equation*}
The trees hanging from vertices in $U$ reweight the probability that the path is bichromatic but do not cause significant distortion since the states of the roots are roughly uniformly distributed (because of uniqueness on the tree). To quantify this,  for $u\in U$ and a colour $c\in [q]$, let 
\[Z_u(c)=\sum_{\substack{\sigma: V_{u}\rightarrow [q];\\ \sigma_{L_u}=\tau_{L_u},\ \sigma_u=c}}w_{F_u}(\sigma).\]
Let also $Z_u=\sum_{c\in [q]} Z_{u}(c)$. By Lemma~\ref{lem:subtreefree} and the choice of $h$ in \eqref{eq:thetahsmall}, we have that for all $c\in [q]$ it holds that 
\begin{equation}\label{eq:v4tv5tbtb321}
\Big|\frac{Z_{u}(c)}{Z_u}-\frac{1}{q}\Big|\leq \vartheta(h)\leq \epsilon.
\end{equation}
We can now write $\mu_H\big(\mathcal{E}(c_1,c_2)\mid \sigma_L=\tau,\, \sigma_{\overline{U}}=\tau'\big)$ as
\[\mu_H\big(\mathcal{E}(c_1,c_2)\mid \sigma_L=\tau,\, \sigma_{\overline{U}}=\tau'\big)=\frac{\sum_{\substack{\sigma: V_{P}\rightarrow \{c_1,c_2\};\ \sigma_{\overline{U}}=\tau'}}w_{P}(\sigma)\prod_{u\in U} Z_u(\sigma_u)}{\sum_{\substack{\sigma: V_{P}\rightarrow [q];\ \sigma_{\overline{U}}=\tau'}}w_{P}(\sigma)\prod_{u\in U} Z_u(\sigma_u)}.\]
Dividing both numerator and denominator by $\prod_{u\in U}Z_u$, we obtain using \eqref{eq:v4tv5tbtb321} that 
\[\mu_H\big(\mathcal{E}(c_1,c_2)\mid \sigma_L=\tau,\, \sigma_{\overline{U}}=\tau'\big)\leq \Big(\frac{1/q+\epsilon}{1/q-\epsilon}\Big)^{|U|}(Z_{P}(c_1,c_2)/Z_P)\leq (4q/\chi B)^{|\overline{U}|} \Big(\frac{1/q+\epsilon}{1/q-\epsilon}\chi\Big)^{|U|}.\]
We have $|U|\geq (1-2\epsilon) \ell$ and $|\overline{U}|\leq 2\epsilon \ell$, so from the choice of $\epsilon$ and $\ell$ in \eqref{eq:eps1selection} and \eqref{eq:ell2selection}, we obtain that
\[\mu_H\big(\mathcal{E}(c_1,c_2)\mid \sigma_L=\tau,\, \sigma_{\overline{U}}=\tau'\big)\leq K^{\ell}/(1+\epsilon)^{\ell}\leq K^{\ell}/q^2,\]
which is exactly \eqref{eq:3vtv3t565er2f554f}, as wanted. This concludes the proof of Lemma~\ref{lem:treespathPotts}.
\end{proof}

\section{Correlation decay and sampling for antiferromagnetic Ising}\label{sec:egbe5tb4hh56}
In this section, we prove Theorem~\ref{thm:sampleising} and Lemma~\ref{lem:corrdecayIsing}. The proofs follow relatively easily from correlation decay bounds on trees appearing in \cite{MosselSly}; the bounds in there are stated for the case of the ferromagnetic Ising model, but there is a simple translation of these bounds to the antiferromagnetic case which allows us to conclude the desired results.  In Section~\ref{sec:fr3r4f3}, we import the results from the literature that we need. We then give the proof of Theorem~\ref{thm:sampleising} in Section~\ref{sec:sampleising} and the proof of Lemma~\ref{lem:corrdecayIsing} in Section~\ref{sec:corrdecayIsing}.

\subsection{Preliminaries}\label{sec:fr3r4f3}
Following \cite{MosselSly}, for a graph $G$ and a vertex $u$ in $G$, we consider the self-avoiding walk tree $T=T_{\rm SAW}(G,u)$, which consists of all paths starting from $u$ and not intersecting themselves, except possibly at the terminal vertex of the path.\footnote{More precisely, a self-avoiding walk in our context is a walk $w=(v_1,\hdots,v_t)$ where vertices are pairwise distinct, except that we allow $v_1=v_t$. The tree $T=T_{\rm SAW}(G,u)$ is induced by the set of all walks starting from $u$, where two walks $w_1,w_2$ are connected by an edge if $w_1$ can be obtained from $w_2$ by deleting the last vertex of $w_2$ (or vice versa). Naturally, we can relabel walks in $T$ according to their last vertex, so that each vertex in $T$ corresponds to a vertex in $G$  (in a many-to-one fashion)}  The following lemma originates in the work of Weitz for the hard-core model \cite{Weitz06}; it is well-known that the lemma holds more generally for any 2-state system. The particular version we state here is close to \cite[Lemma 13]{MosselSly}. 
\begin{lemma}[see, e.g., {\cite[Lemma 13]{MosselSly}}]\label{lem:sawtree}
Let $G=(V,E)$ be a graph and $u$ be a vertex in $G$. Consider the self-avoiding walk tree $T=T_{\rm SAW}(G,u)$ starting from $u$ and denote by $A$ the leaves of the tree. Then, there is a configuration $\eta:A\rightarrow\{1,2\}$ (described in \cite{Weitz06}) such that the following holds for any set $\Lambda\subseteq V$ and any configuration $\tau:\Lambda\rightarrow\{1,2\}$.

Let $U_\Lambda$ be the set of vertices in $T$ which correspond to vertices in $\Lambda$  and let $\eta': U_\Lambda\setminus A\rightarrow \{1,2\}$ be the configuration where each vertex in $U_\Lambda\backslash A$ inherits the state of the corresponding vertex in $\Lambda$ under $\tau$. Then, 
\[\pi_G(\sigma_u=1\mid \sigma_\Lambda=\tau)=\pi_T(\sigma_u=1\mid \sigma_A=\eta, \sigma_{U_\Lambda\setminus A}=\eta').\]
\end{lemma}

The following lemma follows  from a strong spatial mixing result in \cite{MosselSly} for the ferromagnetic Ising model, which in turn  builds upon a lemma from \cite{Berger2005}.
\begin{lemma}\label{lem:strongspatialmixing}
Let $B\in (0,1)$. Let $T=(V,E)$ be a tree, $\Lambda\subseteq V$ be a subset of the vertices and $u$ be an arbitrary vertex. Let $\tau_1,\tau_2:\Lambda\rightarrow \{1,2\}$ be two configurations on $\Lambda$ which differ only on a subset $U\subseteq \Lambda$. Then,
\[\big|\pi_T(\sigma_u=1\mid \sigma_\Lambda=\tau_1)-\pi_T(\sigma_u=1\mid \sigma_\Lambda=\tau_2)\big|\leq \sum_{v\in U}\Big(\frac{1-B}{1+B}\Big)^{\dist(u,v)},\]
where $\dist(u,v)$ denotes the distance between $u$ and $v$ in $T$.
\end{lemma}
\begin{proof}
It is well-known that, on bipartite graphs $G=(V,E)$, there is a measure-preserving bijection between configurations of the antiferromagnetic Ising model with parameter $B\in (0,1)$ and configurations  of the ferromagnetic Ising model with parameter $1/B$, obtained by flipping the states of each vertex on one side of the bipartition.  The desired inequality therefore follows from the strong spatial mixing result for the ferromagnetic Ising model given in~\cite[Lemma 14]{MosselSly}. To translate the parameterisation of that result,  note that in \cite{MosselSly} the weight of a (ferromagnetic) Ising configuration $\sigma$ is parameterised to be proportional to  $\exp(2\beta m(\sigma))$ -- therefore  $1/B = e^{2\beta}$, so $\tanh\beta = \frac{1-B}{1+B}$. Using this translation, we obtain the desired inequality.
\end{proof}

\subsection{Proof of Theorem~\ref{thm:sampleising}}\label{sec:sampleising}
To prove Theorem~\ref{thm:sampleising}, we will use the following algorithm that allows us to compute conditional marginal probabilities with very small absolute error.
\begin{lemma}\label{lem:isingmarginal}
    Let $B\in (0,1)$ and $b>0$ be constants such that $b\frac{1-B}{1+B}<1$,  and let $\Delta\geq 3$ be an integer. Then, there exists  $M_0>0$ such that the following holds for all $M>M_0$.
    
    There is a polynomial-time algorithm that, on input: (i) an $n$-vertex graph $G=(V,E)$ with maximum degree at most $\Delta$ and  average growth $b$ up to depth $L=\left\lceil M\log n\right\rceil$, (ii) a subset $\Lambda\subseteq V$ with a configuration $\tau:\Lambda\to\{1,2\}$, and (iii) a vertex $u\in V\setminus \Lambda$, 
outputs  a number $\hat{p}\in [0,1]$ such that  
    \[\big|\hat{p} -p\big|\leq \frac{B^\Delta}{2n^{11}(1+B^{\Delta})} \mbox{\ \, where $p:=\pi_G(\sigma_u=1\mid\sigma_\Lambda=\tau)$},\]
i.e., $\hat{p}$ is within absolute error $O(1/n^{11})$ from the marginal probability that $\sigma_u=1$ conditioned on $\sigma_\Lambda=\tau$, where $\sigma$ is  from the Ising distribution $\pi_G$ on $G$ with parameter $B$.
\end{lemma}
\begin{proof}
Let  $M_0$ be a large constant so that  for $n\geq 2$ and  $L_0:=\left\lceil M_0\log n\right\rceil$, it holds that 
\begin{equation}\label{eq:533vtvtv4}
b^{L_0}\Big(\frac{1-B}{1+B}\Big)^{L_0}\leq \frac{B^\Delta}{2n^{11}(1+B^{\Delta})}.
\end{equation}
Note that such an $M_0$ exists since $b\frac{1-B}{1+B}<1$. Fix  $M$ to be an arbitrary constant larger than $M_0$. 

Let $G$ be an arbitrary $n$-vertex graph with average growth $b$ up to depth $L=\left\lceil M\log n\right\rceil$, $\Lambda\subseteq V$ be a subset of the vertices, $\tau:\Lambda\to\{1,2\}$ be a configuration on $\Lambda$, and $u\in V\setminus \Lambda$ be a vertex. Consider the self-avoiding walk tree $T=T_{\rm SAW}(G,u)$ starting from $u$, and denote by $A$ the leaves of the tree. By Lemma \ref{lem:sawtree},   we have that there is a configuration $\eta: A\rightarrow \{1,2\}$ such that 
\begin{equation}\label{eq:tv3tv3451124}
\pi_G(\sigma_u=1\mid \sigma_\Lambda=\tau)=\pi_T(\sigma_u=1\mid \sigma_A=\eta, \sigma_{U\setminus A}=\eta')
\end{equation}
 where $U:=U_\Lambda$ is the set of vertices in $T$ that correspond to some vertex in $\Lambda$ and $\eta'$ is the assignment on $U\setminus A$ inherited by $\tau$ (see Lemma~\ref{lem:sawtree} for details). For convenience, let $F:=U\cup A$ and $\zeta: F \rightarrow \{1,2\}$ be the configuration which agrees with $\eta$ on $A$ and with $\eta'$ on $U\backslash A$, so that \eqref{eq:tv3tv3451124} can be rewritten as
\[\pi_G(\sigma_u=1\mid \sigma_\Lambda=\tau)=\pi_T(\sigma_u=1\mid \sigma_F=\zeta).\]
Let $T'=(V',E')$ be the subtree of $T$ induced by vertices at distance at most $L-1$ from $u$ and let $F'= F \cap V'$. Let 
\[p:=\pi_T(\sigma_u=1\mid \sigma_F=\zeta) \quad\mbox{and}\quad \hat{p}:=\pi_{T'}(\sigma_u=1\mid \sigma_{F'}=\zeta_{F'}).\]
Note that $\hat{p}$ can be computed in polynomial time since the tree $T'$ and the  configuration $\zeta_{F'}$ can be constructed in polynomial time (since $T'$ is a tree of size $O(\Delta^{L})$). So, the lemma will follow by showing that 
\begin{equation}\label{eq:papproxphat}
\big|\hat{p} -p\big|\leq \frac{B^\Delta}{2n^{11}(1+B^\Delta)}.
\end{equation}
To prove this, let $J$ be the set of vertices in $V'\backslash F'$ whose distance from $u$ is exactly $L-1$ in $T$ and note that $|J|\leq b^{L}$ since $G$ has average growth $b$ up to depth $L$ (cf. Definition~\ref{def:branchingfactor}). Note also  that, conditioned on the configurations on $J$ and $F'$, the probability that $\sigma_u=1$ depends only on $T'$ (and not on the configuration on rest of the tree $T$). Using the law of total probability,   we can therefore expand $p$ as
\begin{equation}\label{eq:3rf34ff3t1223}
p=\sum_{\iota: J\rightarrow\{1,2\}}\pi_{T'}(\sigma_u=1\mid \sigma_{F'}=\zeta_{F'}, \sigma_J=\iota)  \times \pi_{T}(\sigma_J=\iota\mid \sigma_{F}=\zeta_{F}).
\end{equation}
Therefore, \eqref{eq:papproxphat} will follow by showing that, for any configuration $\iota:J\rightarrow\{1,2\}$ it holds that
\begin{equation}\label{eq:rv3r63eververr}
\big|\pi_{T'}\big(\sigma_u=1\mid \sigma_{F'}=\zeta_{F'},\, \sigma_J=\iota\big)-\hat{p}\big|\leq  \frac{B^\Delta}{2n^{11}(1+B^\Delta)}.
\end{equation} 
Note that we can expand $\hat{p}$ analogously to \eqref{eq:3rf34ff3t1223} by conditioning on the configuration on $J$, so to prove \eqref{eq:3rf34ff3t1223} it suffices to show that for any two configurations $\iota_1,\iota_2: J\rightarrow [q]$ it holds that 
\[\kappa\leq  \frac{B^\Delta}{2n^{11}(1+B^\Delta)}\] 
where $\kappa:=\big|\pi_{T'}\big(\sigma_u=1\mid \sigma_{F'}=\zeta_{F'},\, \sigma_J=\iota_1\big)-\pi_{T'}\big(\sigma_u=1\mid \sigma_{F'}=\zeta_{F'},\, \sigma_J=\iota_2\big)\big|$. By the strong spatial mixing result of Lemma~\ref{lem:strongspatialmixing}, we have that
    \begin{equation*}
\kappa\leq\sum_{v\in J}\Big(\frac{1-B}{B+1}\Big)^{\mathrm{dist}(u,v)}=|J|\Big(\frac{1-B}{B+1}\Big)^L\leq b^{L}\Big(\frac{1-B}{1+B}\Big)^L.
    \end{equation*}
Combining this with the choice of $M_0$ (cf. \eqref{eq:533vtvtv4}), we obtain  \eqref{eq:rv3r63eververr}, thus  concluding the proof of the lemma.  
\end{proof}

Using Lemma~\ref{lem:isingmarginal}, the proof of Theorem~\ref{thm:sampleising} follows by standard techniques. We restate it here for convenience.
\begin{thm:sampleising}
\statethmsampleising
\end{thm:sampleising}
\begin{proof}[Proof of Theorem~\ref{thm:sampleising}]
Denote by $v_1, v_2, \ldots, v_n$ the vertices of $G$. The algorithm will sample the state $s_i$ of vertex $v_i$ sequentially for $i=1,\hdots,n$. We just give the details for the first part of the algorithm, the proof of the second part is completely analogous (namely, it suffices to  assume in the following that we first set $v_1 = u, v_2 = v$, we then fix the states $s_1 = 1, s_2 = 2$ and finally sample the states $s_i$ for $i=3,\hdots,n$.) 

 Assume that, at some time $i=1,\hdots, n$, we have sampled the states $s_1,\hdots, s_{i-1}$ (which can take arbitrary values in $\{1,2\}$). Using the algorithm of Lemma~\ref{lem:isingmarginal}, we obtain in polynomial time numbers $\hat p_i(1), \hat p_i(2) \in [0,1]$ such that $\hat p_i(1)+\hat p_i(2)=1$ and, for $s\in \{1,2\}$, it holds that 
\begin{equation}\label{eq:v4vyeveqsq12}
\big|a_i(s)-\hat{p}_i(s)\big|\leq \frac{B^\Delta}{2n^{11}(1+B^\Delta)},\mbox{ where } a_i(s):=\pi_G(\sigma_{v_i}=s\mid\sigma_{v_1}=s_1, \ldots,\sigma_{v_{i-1}}=s_{i-1}).
\end{equation}
We  then sample the state $s_i$ by letting $s_i = 1$ with probability $\hat{p}_i(1)$, or else $s_i = 2$ (note that $s_i=2$ with probability $\hat{p}_i(2)$). Denote by $\tau$ the final configuration and by $\nu_\tau$ its distribution. 

We will show that, for any configuration $\eta:V\rightarrow\{1,2\}$, it holds that 
\begin{equation}\label{eq:4vgvhyerfwef4112}
|\nu_\tau(\eta)-\pi_G(\eta)|\leq \frac{2}{n^{10}}\pi_G(\eta),
\end{equation}
so by summing over $\eta$ we obtain
\[\big\lVert\nu_\tau-\pi_G\big\rVert_{\TV}=\frac{1}{2}\sum_{\eta\colon V\to\{1,2\}}\big|\nu_\tau(\eta)-\pi_G(\eta)\big|\leq  \frac{1}{n^\cons}\sum_{\eta\colon V\to\{1,2\}}\pi_G(\eta)= 1/n^{\cons},
\]
which proves the first part of the theorem; the second part follows analogously. To prove \eqref{eq:4vgvhyerfwef4112}, fix an arbitrary configuration $\eta:V\rightarrow\{1,2\}$ and let
\[p_{i,\eta}:=\pi_G(\sigma_{v_i}=\eta_{v_i}\mid\sigma_{v_1}=\eta_{v_1}, \ldots,\sigma_{v_{i-1}}=\eta_{v_{i-1}}).\] 
Note that 
\begin{equation}\label{eq:4rfvt4b121314}
\pi_G(\eta)=\prod^n_{i=1} p_{i,\eta}\quad \mbox{and}\quad \nu_\tau(\eta)=\prod^n_{i=1} \hat p_i(\eta_{v_i}).
\end{equation}
Moreover, from \eqref{eq:v4vyeveqsq12}, we have that for all $i=1,\hdots,n$ it holds that
\[\big|p_{i,\eta}-\hat{p}_i(\eta_{v_i})\big|\leq \frac{B^{\Delta}}{2n^{11}(1+B^{\Delta})}\leq \frac{p_{i,\eta}}{2n^{11}},\]
where the last inequality follows from  the lower bound of Lemma~\ref{lem:f4f35f63g}. Using the inequalities $1-x\leq {\mathrm e}^{-x}\leq 1-x/2$ which hold for all $x\in [0,1/2]$, we obtain that, for  $\epsilon=1/n^{11}$, it holds that 
\[{\mathrm e}^{-\epsilon}p_{i,\eta}\leq \hat p_i(\eta_{v_i})\leq {\mathrm e}^{\epsilon}p_{i,\eta}.\] Multiplying over $i=1,\hdots,n$ and combining with \eqref{eq:4rfvt4b121314} gives \eqref{eq:4vgvhyerfwef4112}, thus completing the proof of Theorem~\ref{thm:sampleising}.
\end{proof}

\subsection{Proof of Lemma~\ref{lem:corrdecayIsing}}\label{sec:corrdecayIsing}
In this section, we give the proof of Lemma~\ref{lem:corrdecayIsing}, which we restate here for convenience.
\begin{lem:corrdecayIsing}
\statelemcorrdecayIsing
\end{lem:corrdecayIsing}
\begin{proof}[Proof of Lemma~\ref{lem:corrdecayIsing}]
Let  $M_0'$ be sufficiently large so that  for all $n\geq 2$ and  $L_0:=\left\lceil M_0'\log n\right\rceil$, it holds that 
\begin{equation}\label{eq:533vtvtv4b}
b^{L_0}\Big(\frac{1-B}{1+B}\Big)^{L_0}\leq 1/(2n^{\cons}).
\end{equation}
Note that such a constant exists since $b\frac{1-B}{1+B}<1$. Fix arbitrary $M>M_0'$. Let $G$ be an arbitrary graph with average growth $b$ up to depth $L=\left\lceil M\log n\right\rceil$, $\Lambda\subseteq V$ be a subset of the vertices, $\tau:\Lambda\to\{1,2\}$ be a configuration on $\Lambda$, and $u\in V\setminus \Lambda$ be a vertex. 

Consider the self-avoiding walk tree $T=T_{\rm SAW}(G,u)$ starting from $u$, and denote by $A$ the leaves of the tree and by $U$ the set of vertices in $T$ that correspond to $v$. For a subset of vertices $W$ of the tree we denote by $\sigma_W=1$ the event that all vertices in $W$ have the state 1 and analogously for $\sigma_W=2$.  By Lemma \ref{lem:sawtree},   we have that there is a configuration $\eta: A\rightarrow \{1,2\}$ such that for $s\in \{1,2\}$ it holds that 
\begin{equation}\label{eq:tv3tv34544f1124b}
\pi_G(\sigma_u=1\mid \sigma_v=s)=p_s \mbox{ where } p_s:=\pi_T\big(\sigma_u=1\mid \sigma_A=\eta,\, \sigma_{U\setminus A}=s\big).
\end{equation}

Let $T'=(V',E')$ be the subtree of $T$ induced by vertices at distance at most $L-1$ from $u$ and let $A'= A \cap V'$, $U'=U\cap V'$. For $s\in \{1,2\}$, let 
\[p_s':=\pi_{T'}(\sigma_u=1\mid \sigma_{A'}=\eta_{A'}, \sigma_{U'\backslash A'}=s).\]
We will show that
\begin{gather}
|p_1'-p_2'|\leq \sum^{L}_{\ell=1}P_\ell(G,u,v)\Big(\frac{1-B}{1+B}\Big)^{\ell},\label{eq:p1p2b}\\
|p_1-p_1'|\leq 1/(2n^{\cons}),\quad |p_2-p_2'|\leq 1/(2n^{\cons})\label{eq:p1p1p2p2}.
\end{gather}
Assuming these for the moment, we obtain by \eqref{eq:tv3tv34544f1124b} and the triangle inequality that 
\[|\pi_G(\sigma_u=1\mid \sigma_v=1)-\pi_G(\sigma_u=1\mid \sigma_v=2)|=|p_1-p_2|\leq \frac{1}{n^{\cons}}+\sum^{L}_{\ell=1}P_\ell(G,u,v)\Big(\frac{1-B}{1+B}\Big)^{\ell},
\]
thus proving the lemma. It thus remains to prove \eqref{eq:p1p2b} and \eqref{eq:p1p1p2p2}.

To prove \eqref{eq:p1p2b}, note that by Lemma~\ref{lem:strongspatialmixing}, we have that
\[\big|p_1'-p_2'\big|\leq \sum_{w\in U'\backslash A'}\Big(\frac{1-B}{1+B}\Big)^{\dist(u,w)}\leq \sum^{L}_{\ell=1}P_\ell(G,u,v)\Big(\frac{1-B}{1+B}\Big)^{\ell},\]
where the last inequality follows from observing  that each vertex $w\in U'\backslash A'$ with $\dist(u,w)=\ell$ corresponds to a distinct path with $\ell+1$ vertices between $u$ and $v$ in $G$. This proves \eqref{eq:p1p2b}.

To prove \eqref{eq:p1p1p2p2}, we focus on showing that $|p_1-p_1'|\leq 1/(2n^{\cons})$, the other inequality being completely analogous. We follow closely a similar argument which was presented in the proof of Lemma~\ref{lem:isingmarginal}. Let $J$ be the set of vertices in $V'\backslash (U'\cup A')$ whose distance from $u$ is exactly $L-1$ in $T$ and note that $|J|\leq b^{L}$ since $G$ has average growth $b$ up to depth $L$ (cf. Definition~\ref{def:branchingfactor}). Note also  that, conditioned on the configurations on $J$ and $U'\cup A'$, the probability that $\sigma_u=1$ depends only on $T'$ (and not on the configuration on rest of the tree $T$), i.e.,   we can expand $p$ as
\begin{equation}\label{eq:3rf34ff3t1223v2}
p_1=\sum_{\iota: J\rightarrow\{1,2\}}\pi_{T'}(\sigma_u=1\mid \sigma_{A'}=\eta_{A'}, \sigma_{U'\backslash A'}=1, \sigma_J=\iota)  \times \pi_{T}(\sigma_J=\iota\mid \sigma_{A}=\eta, \sigma_{U\backslash A}=1).
\end{equation}
Therefore, $|p_1-p_1'|\leq 1/(2n^{\cons})$ will follow by showing that, for any configuration $\iota:J\rightarrow\{1,2\}$ it holds that
\begin{equation}\label{eq:rv3r63eververrv2}
\big|\pi_{T'}\big(\sigma_u=1\mid \sigma_{A'}=\eta_{A'}, \sigma_{U'\backslash A'}=1,\, \sigma_J=\iota\big)-p_1'\big|\leq  1/(2n^{\cons}).
\end{equation} 
We can expand $p_1'$ analogously to \eqref{eq:3rf34ff3t1223v2} by conditioning on the configuration on $J$, so to prove \eqref{eq:3rf34ff3t1223v2} it suffices to show that for any two configurations $\iota_1,\iota_2: J\rightarrow [q]$ it holds that $\kappa\leq 1/(2n^{\cons})$ where 
\[\kappa:=\big|\pi_{T'}\big(\sigma_u=1\mid \sigma_{A'}=\eta_{A'}, \sigma_{U'\backslash A'}=1,\, \sigma_J=\iota_1\big)-\pi_{T'}\big(\sigma_u=1\mid \sigma_{A'}=\eta_{A'}, \sigma_{U'\backslash A'}=1,\, \sigma_J=\iota_2\big)\big|\] 
By the strong spatial mixing result of Lemma~\ref{lem:strongspatialmixing}, we have that
    \begin{equation*}
\kappa\leq\sum_{v\in J}\Big(\frac{1-B}{B+1}\Big)^{\mathrm{dist}(u,v)}=|J|\Big(\frac{1-B}{B+1}\Big)^L\leq b^{L}\Big(\frac{1-B}{1+B}\Big)^L.
    \end{equation*}
Combining this with the choice of $M_0'$ (cf. \eqref{eq:533vtvtv4b}), we obtain  $\kappa\leq 1/(2n^{\cons})$, thus  concluding the proof of  $|p_1-p_1'|\leq 1/(2n^{\cons})$ and  therefore completing the proof of Lemma~\ref{lem:corrdecayIsing}.  
\end{proof}

\bibliographystyle{plain}
\bibliography{\jobname}

\end{document}